\documentclass[superscriptaddress,longbibliography,nofootinbib,twocolumn,notitlepage,10pt]{revtex4-1}
\usepackage{graphicx}
\usepackage{mathrsfs} 
\usepackage{amssymb,amsmath,amsthm,amsfonts,mathtools,comment}
\usepackage[colorlinks]{hyperref}
\usepackage{bm}

\newtheorem{lemma}{Lemma}

\newcommand{\eq}[1]{Eq.~\hyperref[eq:#1]{(\ref*{eq:#1})}}
\renewcommand{\sec}[1]{\hyperref[sec:#1]{Section~\ref*{sec:#1}}}
\newcommand{\app}[1]{\hyperref[app:#1]{Appendix~\ref*{app:#1}}}
\newcommand{\tab}[1]{\hyperref[tab:#1]{Table~\ref*{tab:#1}}}
\newcommand{\fig}[1]{\hyperref[fig:#1]{Figure~\ref*{fig:#1}}}
\newcommand{\figa}[2]{\hyperref[fig:#1]{Figure~\ref*{fig:#1}#2}}
\newcommand{\figx}[2]{\hyperref[fig:#1]{Figure~\ref*{fig:#1}(#2)}}
\newcommand{\thm}[1]{\hyperref[thm:#1]{Theorem~\ref*{thm:#1}}}
\newcommand{\lem}[1]{\hyperref[lem:#1]{Lemma~\ref*{lem:#1}}}
\newcommand{\cor}[1]{\hyperref[cor:#1]{Corollary~\ref*{cor:#1}}}
\newcommand{\defn}[1]{\hyperref[def:#1]{Definition~\ref*{def:#1}}}
\newcommand{\alg}[1]{\hyperref[alg:#1]{Algorithm~\ref*{alg:#1}}}

\newcommand{\ceil}[1]{\lceil{#1}\rceil}

\newcommand{\F}{\mathbf{F}}

\newcommand\V{{\cal V}}

\newcommand{\be}{\begin{equation}}
\newcommand{\ee}{\end{equation}}
\newcommand{\ba}{\begin{eqnarray}}
\newcommand{\ea}{\end{eqnarray}}

\newcommand{\nn}{\nonumber \\}

\usepackage{physics}

\def\bra#1{\mathinner{\langle{#1}|}}
\def\ket#1{\mathinner{|{#1}\rangle}}

\makeatletter
\renewcommand*\env@matrix[1][\arraystretch]{%
	\edef\arraystretch{#1}%
	\hskip -\arraycolsep
	\let\@ifnextchar\new@ifnextchar
	\array{*\c@MaxMatrixCols c}}
\makeatother
\newcommand{\MQ}{\affiliation{%
Department of Physics and Astronomy,
Macquarie University, Sydney, NSW, Australia}}

\newcommand{\Google}{\affiliation{%
Google Quantum AI, Venice, CA, United States}}

\newcommand{\Columbia}{\affiliation{Department of Chemistry, Columbia University, New York, NY, United States}}

\newcommand{\Sandia}{\affiliation{Quantum Algorithms and Applications Collaboratory,
Sandia National Laboratories, Albuquerque NM, United States}}

\newcommand{\NM}{\affiliation{Department of Physics and Astronomy,
University of New Mexico, Albuquerque, NM, United States}}

\begin{document}
\title{Quantum simulation of exact electron dynamics can be\\ more efficient than classical mean-field methods}

\date{\today}
\author{Ryan Babbush}
\email{corresponding author: ryanbabbush@gmail.com}
\Google

\author{William J.~Huggins}
\Google

\author{Dominic W.~Berry}
\MQ

\author{Shu Fay Ung}
\Columbia

\author{Andrew Zhao}
\Google
\NM

\author{David R.~Reichman}
\Columbia

\author{Hartmut Neven}
\Google

\author{Andrew D.~Baczewski}
\Sandia

\author{Joonho Lee}
\email{corresponding author: linusjoonho@gmail.com}
\Google
\Columbia

\begin{abstract}
Quantum algorithms for simulating electronic ground states are slower than popular classical mean-field algorithms such as Hartree-Fock and density functional theory, but offer higher accuracy. Accordingly, quantum computers have been predominantly regarded as competitors to only the most accurate and costly classical methods for treating electron correlation. However, here we tighten bounds showing that certain first quantized quantum algorithms enable exact time evolution of electronic systems with exponentially less space and polynomially fewer operations in basis set size than conventional real-time time-dependent Hartree-Fock and density functional theory. Although the need to sample observables in the quantum algorithm reduces the speedup, we show that one can estimate all elements of the $k$-particle reduced density matrix with a number of samples scaling only polylogarithmically in basis set size. We also introduce a more efficient quantum algorithm for first quantized mean-field state preparation that is likely cheaper than the cost of time evolution. We conclude that quantum speedup is most pronounced for finite temperature simulations and suggest several practically important electron dynamics problems with potential quantum advantage.
\end{abstract}

\maketitle

\subsection*{Introduction}

Quantum computers were first proposed as tools for dynamics by Feynman \cite{Feynman1982} and later shown to be universal for that purpose by Lloyd \emph{et al.}~\cite{Lloyd1996}. Like those early papers, most work on this topic assumes that the advantage of quantum computers for dynamics is that they provide an approach to simulation with systematically improvable precision but without scaling exponentially. Here, we advance and analyze a different idea: certain (exact) quantum algorithms for dynamics may be more efficient than even classical methods that make uncontrolled approximations. We examine this in the context of simulating interacting fermions -- systems of relevance in fields such as chemistry, physics, and materials science.

It is often the case that practically relevant ground state problems in chemistry and materials science do not exhibit strong correlation. For those problems, many classical heuristic methods work well \cite{Bartlett2007Feb,Mardirossian2017Oct,Lee2022Oct}. Even for some strongly correlated systems, there are successful polynomial-scaling classical methods \cite{EQA}. Here, we argue that even if electronic systems are well described by mean-field theory, quantum algorithms can achieve speedup over classical algorithms for simulating the time evolution of such systems. We focus on comparing to mean-field methods such as real-time time-dependent Hartree-Fock and density functional theory due to their popularity and well-defined scaling. Nonetheless, many of our arguments translate to advantages over other known classical approaches to dynamics that are more expensive but more accurate than mean-field methods. This is a sharp contrast to prior studies of quantum algorithms, which have focused on strongly correlated ground state problems such as FeMoCo \cite{Reiher2017,Li2019,Berry2019,vonBurg2020,Lee2020}, P450 \cite{P450}, chromium dimers \cite{Elfving2020} and jellium \cite{BabbushLow,BabbushSpectra,Kivlichan2019,McArdle2022}, assessing quantum advantage over only the most accurate and costly classical algorithms.

Quantum algorithms competitive with efficient classical algorithms for dynamics have been analyzed in contexts outside of fermionic simulation. For example, work by Somma \cite{Somma2015} showed that certain one-dimensional quantum systems, such as harmonic oscillators, could be simulated with sublinear complexity in system size. Experimentally motivated work by Geller \emph{et al.}~\cite{Geller2015} also proposed simulating quantum systems in a single-excitation subspace, a task for which they suggested a constant factor speedup was plausible. However, neither work is connected to the context studied here.

We begin by analyzing the cost of classical mean-field dynamics and recent exact quantum algorithms in first quantization, focusing on explaining why there is often a quantum speedup in the number of basis functions over classical mean-field methods. Next, we analyze the overheads associated with measuring quantities of interest on a quantum computer and introduce more efficient methods for measuring the one-particle reduced density matrix in first quantization (which characterizes all mean-field observables). Then, we discuss the costs of preparing mean-field states on the quantum computer and describe new methods that make this cost likely negligible compared to the cost of time evolution. Finally, we conclude with a discussion of systems where these techniques might lead to practical quantum advantage over classical mean-field simulations.

\subsection*{Classical mean-field dynamics}

Here we will discuss mean-field classical algorithms for simulating the dynamics of interacting systems of electrons and nuclei. Thus, we will focus on the {\it ab initio} Hamiltonian with $\eta$ particles discretized using $N$ basis functions, which can be expressed as

\begin{equation}
H = \sum_{\mu\nu}^N
h_{\mu \nu}
a_{\mu}^\dagger
a_{\nu}
+\frac12
\sum_{\mu\nu\lambda\sigma}^N
\left(\mu\nu|\lambda\sigma\right)
a_{\mu}^\dagger
a_{\lambda}^\dagger
a_{\sigma}
a_{\nu}
\end{equation}
where $a_{\mu}^{(\dagger)}$ is the fermionic annihilation (creation) operator for the $\mu$-th orbital and integral values are given by
\begin{align}
h_{\mu \nu} & = \int \textrm{d}r \, \phi_\mu^*\left(r\right) \left(-\frac{\nabla^2}{2} + V\left(r\right)\right)  \phi_\nu \left(r\right), \\
\left(\mu\nu|\lambda\sigma\right) & =  \int \textrm{d}r_1 \textrm{d}r_2 \frac{\phi_{\mu}^* \left(r_1\right) \phi_{\nu}\left(r_1\right)  \phi_\lambda^* \left(r_2\right) \phi_\sigma \left(r_2\right) }{\left| r_1 - r_2\right|} \, .
\end{align}
Here, $V(r)$ is the external potential (perhaps arising from the nuclei) and $\phi_\mu(r)$ represents a spatial orbital. 

Exact quantum dynamics is encoded by the time-dependent Schr{\"o}dinger equation given by
\begin{equation}
i\frac{\partial}{\partial t}\ket{\psi\left(t\right)} = H \ket{\psi\left(t\right)} \, .
\label{eq:tdse}
\end{equation}
Mean-field dynamics, such as real-time time-dependent Hartree-Fock (RT-TDHF) \cite{Dreuw2005Nov}, employs a time-dependent variational principle within the space of single Slater determinants (i.e., anti-symmetrized product states) to approximate \eq{tdse}. Other methods with similar cost such as real-time time-dependent density functional theory (RT-TDDFT) rely on a relationship between the interacting system and an auxiliary non-interacting system to define dynamics within a space of single Slater determinants~\cite{runge1984density,van1999mapping,Dreuw2005Nov}. In both methods, there are $\eta$ occupied orbitals, each expressed as a linear combination of $N$ basis functions using the coefficient matrix, $\mathbf C_\text{occ}$. The dimension of $\mathbf C_\text{occ}$ is $N \times \eta$. These orbitals then constitute a Slater determinant (i.e., anti-symmetric product states), $\det(\mathbf C_\text{occ})$. Storing $\mathbf C_\text{occ}$ on a classical computer has space complexity ${\cal O}(N \eta \log(1/\epsilon))$.

As a result of this approximation, we solve the following effective time-dependent equation for the occupied orbital coefficients that specify the Slater determinant $\mathbf C_\text{occ}(t)$ at a given moment in time:
\begin{equation}
i \frac{\partial \mathbf C_\text{occ}\left(t\right)}{\partial t}
=
\mathbf F\left(t\right)\mathbf C_\text{occ}\left(t\right)
\label{eq:tdhf}
\end{equation}
where the effective one-body mean-field operator $\mathbf F(t)$, also known as the time-dependent Fock matrix, is
\begin{equation}\label{eq:Fdef}
F_{\mu\nu}\!\left(t\right) = h_{\mu \nu} + 
\sum_{\lambda\sigma}^{N} \left(\left(\mu\nu|\lambda\sigma\right)-\frac{\left(\mu\sigma|\lambda\nu\right)}{2}\right) P_{\sigma\lambda}\!\left(t\right)
\end{equation}
with $\mathbf P(t) = \mathbf C_\text{occ} (t)(\mathbf C_\text{occ}(t))^\dagger$.
While $\mathbf F(t)$ is an $N \times N$ dimensional matrix, we can apply it to $\mathbf C_\text{occ}(t)$ without explicitly constructing it, thus avoiding a space complexity of ${\cal O}(N^2 \log (1/\epsilon))$. Using the most common methods of applying this matrix to update each of $\eta$ occupied orbitals in $\mathbf C_\text{occ}(t)$ requires $\widetilde{\cal O}(N^2\eta)$ total operations\footnote{Throughout the paper we will use the convention that $\widetilde{\cal O}(\cdot)$ implies big-${\cal O}$ notation suppressing polylogarithmic factors.}.

However, a recent technique referred to as occ-RI-K by Head-Gordon and co-workers \cite{Manzer2015Jul}, and similarly ``Adaptively Compressed Exchange'' (ACE) \cite{Lin2016May,LinTDDFT2019} by Lin and co-workers, further reduces this cost. These methods leverage the observation that, when restricted to the subspace of the $\eta$ occupied orbitals, the effective rank of the Fock operator scales as ${\cal O}(\eta)$. This gives an approach to updating the Fock operator that requires only
\begin{equation}
\label{eq:update_cost}
\widetilde{\cal O}(N \, \eta^2 )
\end{equation}
operations. Below we will use gate complexity and the number of operations interchangeably when discussing the scaling of classical algorithms. Although these techniques are not implemented in every quantum chemistry code, we regard them as the main point of comparison to quantum algorithms. We also note that RT-TDDFT with hybrid functionals \cite{Jia2019Jul} has the same scaling as RT-TDHF. Simpler RT-TDDFT methods (i.e., those without exact exchange) can achieve better scaling, $\widetilde{\cal O}(N \eta)$ in a plane wave basis, but are often less accurate.

For finite-temperature simulation, one often needs to track $M > \eta$ orbitals with appreciable occupations, increasing the space complexity to ${\cal O}(N M \log(1/\epsilon))$. This increases the cost of occ-RI-K or ACE mean-field updates to $\widetilde{\cal O}(N M^2)$. At temperatures well above the Fermi energy, most orbitals have appreciable occupations so $M \simeq N$.
More expensive methods for dynamics that include electron correlation in the dynamics tend to scale at least linearly in the cost of ground state simulation at that level of theory. Thus, speedup over mean-field methods implies speedup over more expensive methods.

In recent years, by leveraging the ``nearsightedness'' of electronic systems \cite{Prodan2005Aug}, ``linear-scaling'' methods have been developed that achieve updates scaling as ${\cal{O}}(N)$ \cite{Kussmann2013Nov}.
For RT-TDHF and RT-TDDFT, linear-scaling comes from the fact that the off-diagonal elements of $\mathbf P$ fall off quickly with distance for the ground state \cite{ORourke2015Sep} and some low-lying excited states \cite{Zuehlsdorff2013Aug} in a localized basis.
One can show that for gapped ground states, the decay rate is exponential, whereas for metallic ground states, it is algebraic \cite{Prodan2005Aug}.
However, often such asymptotic behavior only onsets for very large systems, and the onset can be highly system-dependent. This should be contrasted with the scaling analyzed above and the scaling of quantum algorithms ({\it vide infra}) that onsets already at modest system sizes. Furthermore, the nearsightedness of electrons does not necessarily hold for dynamics of highly excited states and at high temperatures. Due to these limitations, we do not focus on comparing quantum algorithms and classical linear scaling methods.

It has also been suggested that one can exploit a low-rank structure of occupied orbitals using the quantized tensor train format \cite{Khoromskaia2011Jan}. 
Assuming the compression of orbitals in real space is efficient such that the rank does not grow with system size or the number of grid points, the storage cost is reduced to $\tilde{\cal{O}}(\eta)$, and the update cost is $\tilde{\cal{O}}(\eta^2)$. It is unclear how well compression can be realized for dynamics problems and finite-temperature problems, and to our knowledge, it has been never been deployed for those purposes. Accordingly, we do not consider this approach as the point of comparison.

We now discuss how many time steps are required to perform time evolution using classical mean-field approaches. The number of time steps will depend on the target precision as well as the total unitless time
\begin{equation}
    T = \max_{\mathbf C_\text{occ}} \left\| \mathbf{F} \right\| \, t,
    \label{eq:time}
\end{equation}
where $t$ is duration of time-evolution and $\| \cdot \|$ denotes the spectral norm.
This dependence on the norm of $\F$ is similar to what would be obtained in the case of linear differential equations despite the dependence on $\mathbf C_\text{occ}$; see \app{scanorm} for a derivation.
We can upper bound $T$ by considering its scaling in a local basis, and with open boundary conditions. We find
\begin{equation}
\label{eq:norm}
 \max_{\mathbf C_\text{occ}}\! \left\| \mathbf{F}\right \|\! = {\cal O}\!\left(\frac{\eta^{2/3}}{\delta}\! +  \frac{1}{\delta^2} \right)\! = {\cal O}\!\left(N^{1/3}\eta^{1/3}\! + \frac{N^{2/3}}{\eta^{2/3}}\right),
\end{equation}
where $\delta = {\cal O}((\eta / N)^{1/3})$ is the minimum grid spacing.
The first term comes from the Coulomb operator, and the second comes from the kinetic energy operator.
This scaling for $\delta$ comes from taking the computational cell volume proportional to $\eta$.

We briefly describe how this scaling for the norm is obtained and refer the reader to \app{scanorm} for more details. The $1/\delta^2$ term is obtained from the kinetic energy term in $h_{\mu \nu}$.
When diagonalized, that term will be non-zero only when $\mu = \nu$ with entries scaling as ${\cal O}(1/\delta^2)$ due to the $\nabla^2$ in the expression for $h_{\mu\nu}$.
That upper bounds the spectral norm for this diagonal matrix, and the spectral norm is unchanged under change of basis. The $\eta^{2/3}/\delta$ comes from the sum in the expression for $F_{\mu\nu}$. To bound the tensor norm of $(\mu\nu|\lambda\sigma) - (\mu\sigma|\lambda\nu)/2$ we can bound the norms of the two terms separately.
For each, the tensor norm can be upper bounded by noting that the summing over $\mu\nu,\lambda\sigma$ with normalized vectors corresponds to transformations of the individual orbitals in the integral defining $(\mu\nu|\lambda\sigma)$.
Since orbitals cannot be any more compact than width $\delta$, the $1/|r_1-r_2|$ in the integral averages to give ${\cal O}(1/\delta)$.
There is a further factor of $\eta^{2/3}$ when accounting for $\eta$ electrons that cannot be any closer than $\eta^{1/3}\delta$ on average.

The number of time steps required to effect evolution to within error $\epsilon$ depends on the choice of time integrator. Many options are available~\cite{castro2004propagators,jia2018fast,kononov2022electron}, and the optimal choice depends on implementation details like the basis set and pseudization scheme, as well as the desired accuracy \cite{Shepard2021Sep}.
In \app{scanorm}, we argue that the minimum number of time steps $t / \Delta t$ one could hope for by using an arbitrarily high order integration scheme of this sort is $T^{1 + o(1)} / \epsilon^{o(1)}$.
In particular, for an order $k$ integrator, the error can be bounded as $\mathcal{O}((\|\mathbf F\|\Delta t)^{k+1})$, with a possibly $k$-dependent constant factor that is ignored in this expression.
That means the error for $t/\Delta t$ time steps is $\mathcal{O}(t\|\mathbf F\|^{k+1}\Delta t^k)$.
To obtain error no more than $\epsilon$, take $(t/\Delta t)^{k} = \mathcal{O}((t\|\mathbf F\|^{k+1}/\epsilon)$, so the number of time steps is $t/\Delta t = \mathcal{O}(T^{1+1/k}/\epsilon^{1/k})$. Plugging \eq{norm} into \eq{time} and multiplying the update cost in \eq{update_cost} by $T^{1 + o(1)} / \epsilon^{o(1)}$ time steps, we find the number of operations required for classical mean-field time-evolution is
\begin{equation}
\left(N^{4/3}\eta^{7/3}t + N^{5/3} \eta^{4/3} t\right) \left(\frac{N t}{\epsilon}\right)^{o\left(1\right)} \, .
\label{eq:classical_cost}
\end{equation}

Finally, when performing mean-field dynamics, the central quantity of interest is often the one-particle reduced density matrix (1-RDM). The 1-RDM is an $N \times N$ matrix defined as a function of time with matrix elements
\begin{align}
\rho_{\mu\nu}\left(t\right) 
& =  \bra{\psi\left(t\right)}a_\mu^\dagger a_\nu\ket{\psi\left(t\right)}.
\end{align}
The 1-RDM is the central quantity of interest because it can be used to reconstruct any observable associated with a Slater determinant efficiently. For more general states, one would also need higher order RDMs; however, all higher order RDMs can be exactly computed from the 1-RDM via Wick's theorem when the wavefunction is a single Slater determinant \cite{Shavitt2009Aug}. Thus, when mean-field approximations work well, the time-dependent 1-RDM can also be used to compute multi-time correlators such as Green's functions and spectral functions.

\subsection*{Exact quantum dynamics in first quantization}

One of the key advantages of some quantum algorithms over mean-field classical methods is the ability to perform dynamics using the compressed representation of first quantization. First quantized quantum simulations date back to \cite{Wiesner1996SimulationsComputer,Abrams1997,Zalka1998,Boghosian1998}. They were first applied to fermionic systems in \cite{Abrams1997} and developed for molecular systems in \cite{Lidar1999,Kassal2008}. In first quantization, one encodes the wavefunction using $\eta$ different registers (one for each occupied orbital), each of size $\log N$ (to index the basis functions comprising each occupied orbital). The space complexity of first quantized quantum algorithms is ${\cal O}(\eta \log N)$.

As described previously, mean-field classical methods require space complexity of ${\cal O}(N\eta  \log (1/\epsilon))$ where $\epsilon$ is the target precision. Thus, these quantum algorithms require exponentially less space in $N$. Usually, when one thinks of quantum computers more efficiently encoding representations of quantum systems, the advantage comes from the fact that the wavefunction might be specified by a Hilbert space vector of dimension ${N \choose \eta}$ and could require as much space to represent explicitly on a classical computer. However, this alone \emph{cannot} give exponential quantum advantage in storage in $N$ over classical mean-field methods since mean-field methods only resolve entanglement arising from anti-symmetry and do not attempt to represent wavefunction in the full Hilbert space. Instead, the scaling advantage these quantum algorithms have over mean-field methods is related to the ability to store the distribution of each occupied orbital over $N$ basis functions, using only $\log N$ qubits.
But quantum algorithms require more than the compressed representations of first quantization in order to realize a scaling advantage over classical mean-field methods; they must also have sufficiently low gate complexity in the basis size and other parameters. 

Here we will review and tighten bounds for the most efficient known quantum algorithms for simulating the dynamics of interacting electrons. Early first quantized algorithms for simulating chemistry dynamics such as \cite{Lidar1999,Kassal2008} were based on Trotterization of the time-evolution operator in a real space basis and utilized the quantum Fourier transform to switch between a representation where the potential operator was diagonal and the kinetic operator was diagonal. This enabled Trotter steps with gate complexity $\widetilde{\cal O}(\eta^2)$ but the number of Trotter steps required for the approach of those papers scaled worse than linearly in $N$, $\eta$, the simulation time $t$ and the desired inverse error in the evolution, $1/\epsilon$.

Leveraging recent techniques for bounding Trotter error \cite{Childs2019,Su2020,Low2022b}, in \app{trotter_step} we show that using sufficiently high order Trotter formulas, the overall gate complexity of these algorithms can be reduced to
\begin{equation}\label{eq:Trotcomp}
\left(N^{1/3}\eta^{7/3}t +  N^{2/3}\eta^{4/3}t\right)\left( \frac{N t}{\epsilon}\right)^{o\left(1\right)} \, .
\end{equation}
This is the lowest reported scaling of any Trotter based first quantized quantum chemistry simulation.
We remark that the $N^{1/3}\eta^{7/3} t$ scaling is dominant whenever $N < \Theta(\eta^3)$. In that regime, it represents a quartic speedup in basis size for propagation over the classical mean-field scaling given in \eq{classical_cost}.

The first algorithms to achieve sublinear scaling in $N$ were those introduced by Babbush \emph{et al.}~\cite{BabbushContinuum}. That work focused on first quantized simulation in a plane wave basis and leveraged the interaction picture simulation scheme of \cite{Low2018} to give gate complexity scaling as
\begin{equation}
    \widetilde{\cal O}\left(N^{1/3}\eta^{8/3} t\right) \, .
    \label{eq:interaction_scaling}
\end{equation}
When $N > \Theta(\eta^4)$, this algorithm is more efficient than the Trotter based approach. Since that is also the regime where the second term in \eq{classical_cost} dominates that scaling, this represents a quintic speedup in $N$, coupled with a quadratic slowdown in $\eta$, over mean-field classical algorithms.
The work of Su \emph{et al.}~\cite{Su2021} analyzed the constant factors in the scaling of this algorithm for use in ground state preparation via quantum phase estimation \cite{Aspuru-Guzik2005}. In \app{constant_factors} of this work we analyze the constant factors in the scaling of this algorithm when deployed for time-evolution. 
Su \emph{et al.}~\cite{Su2021} also introduced algorithms with the same scaling as \eq{interaction_scaling} but in a grid representation (see Appendix K therein).

A key component of the algorithms of \cite{BabbushContinuum,Su2021} is the realization of block encodings \cite{Low2016} with just $\widetilde{\cal O}(\eta)$ gates. The difficult part of the block encoding is the preparation of a superposition state with amplitudes proportional to the square root of the Hamiltonian term coefficients. A novel quantum algorithm is devised for this purpose in \cite{BabbushContinuum} which scales only polylogarithmically in basis size. The $N^{1/3}$ dependence of \eq{interaction_scaling} enters via the number of times the block encoding must be repeated to perform time evolution, related to the norm of the potential operator. Under certain assumptions, the norm of the potential term can be reduced to a polylogarithmic dependence on $N$ (see \app{softened_potential} for more details). In that case, exponential quantum advantage in $N$ is possible.

We note that second quantized algorithms outperform first quantized quantum algorithms in gate complexity when $N < \Theta(\eta^2)$. This is because while the best scaling Trotter steps in first quantization require $\widetilde{\cal O}(\eta^2)$ gates \cite{Kassal2008}, the best scaling Trotter steps in second quantization require $\widetilde{\cal O}(N)$ gates. As recently shown in \cite{Low2022b}, such approaches lead to a total gate complexity for Trotter based second quantized algorithms scaling as
\begin{equation}
\left(N^{4/3}\eta^{1/3}t +  \frac{N^{5/3}}{\eta^{2/3}}t\right)\left( \frac{N t}{\epsilon}\right)^{o\left(1\right)} \, .
\end{equation}
In the limit that $\eta = \Theta(N)$, this approach has ${\cal O}(N^{5/3})$ gate complexity, which is significantly less than the ${\cal O}(N^{8/3})$ gate complexity of Trotter based first quantized quantum algorithms mentioned here, or the ${\cal O}(N^{11/3})$ gate complexity of classical mean-field algorithms. (See \app{regimes} for discussion on the overall quantum speedup in different regimes of how $N$ scales in $\eta$.) However, these second quantized approaches generally require at least ${\cal O}(N)$ qubits. The approach used in \cite{Low2022b} to implement Trotter steps involves the fast multipole method \cite{Rokhlin1985}, which requires ${\cal O}(N \log N)$ qubits as well as the restriction to a grid-like basis. When using such basis sets, we expect $N \gg \eta$, and so this space complexity would be prohibitive for quantum computers.

Methods such as fast multipole \cite{Rokhlin1985}, Barnes-Hut \cite{BarnesHut}, or particle-mesh Ewald \cite{Ewald1993} compute the Coulomb potential in time $\widetilde{\cal O}(\eta)$ when implemented within the classical random access memory model. If the Coulomb potential could be computed with that complexity on a quantum computer it would speed up the first quantized Trotter algorithms discussed here by a factor of ${\cal O}(\eta)$. However, it is unclear whether such algorithms extend to the quantum circuit model with the same complexity without unfavorable assumptions such as QRAM \cite{Childs2022,Giovannetti2008}, or without restricting the maximum number of electrons within a region of space (see \app{regimes} for details). Thus, we exclude such approaches from our comparisons here.

\subsection*{Quantum measurement costs}

In contrast to classical mean-field simulations, on a quantum computer, all observables must be sampled from the quantum simulation. There are a variety of techniques for doing this, with the optimal choice depending on the target precision in the estimated observable as well as the number and type of observables one wishes to measure. For example, when measuring $W$ unit norm observables to precision $\epsilon$ one could use algorithms introduced in \cite{Huggins2021} which require $\widetilde{\cal O}(\sqrt{W} / \epsilon)$ state preparations and ${\cal O}(W \log(1/\epsilon))$ ancillae.
Thus, to measure all $W = {\cal O}(N^2)$ elements of the 1-RDM to a fixed additive error in each element, this approach would require $\widetilde{\cal O}(N / \epsilon)$ circuit repetitions. While scaling optimally in $\epsilon$ for quantum algorithms, this linear scaling in $N$ would decrease the speedup over classical mean-field algorithms.

Instead, here we will focus on measuring the 1-RDM with a new variation of the classical shadows method.
Classical shadows were introduced in \cite{Huang2020} and adapted for second quantized fermionic systems in \cite{Zhao2021,Wan2022,OGorman2022-ag,Low2022}.
Our approach is to apply a separate random Clifford channel to each of the $\eta$ different $\log N$ sized registers representing an occupied orbital.
Applying a random Clifford on $\log N$ qubits requires ${\cal O}(\log^2\!N)$ gates; thus, ${\cal O}(\eta \log^2\!N)$ gates comprise the full channel (a negligible cost relative to time-evolution).
In \app{shadows} we prove that repeating this procedure $\widetilde{\cal O}(\eta / \epsilon^2)$ times enables estimation of all 1-RDM elements to within additive error $\epsilon$.
More generally, we prove that this same procedure allows for estimating all higher order $k$-particle RDMs elements with $\widetilde{\cal O}(k^k \eta^k / \epsilon^2)$ circuit repetitions.
In the next section and in \app{state_prep}, we describe a way to map second quantized representations to first quantization, effectively extending the applicability of these classical shadows techniques to second quantization as well.

To give some intuition for how this works, we consider the 1-RDM elements in first quantization:
\begin{equation}
    \rho_{\mu\nu}\left(t\right) = \bra{\psi\left(t\right)} \left(\sum_{j=1}^{\eta} \ket{\mu}\!
    \!\bra{\nu}_j\right) \ket{\psi\left(t\right)} \, ,
\end{equation}
where the subscript $j$ indicates which of the $\eta$ registers the orbital-$\nu$ to orbital-$\mu$ transition operator acts upon.
Due to the antisymmetry of the occupied orbital registers in first quantization, we could also obtain the 1-RDM by measuring the expectation value of an operator such as $\eta \ket{p}\!
    \!\bra{q}_1 $, which acts on just one of the $\eta$ registers.
Because $\eta \ket{p}\!
    \!\bra{q}_1 $ has the Hilbert-Schmidt norm of ${\cal O}(\eta)$, the standard classical shadows procedure applied to this $\log N$ sized register would require $\widetilde{\cal O}(\eta^2 / \epsilon^2)$ repetitions.
But we can parallelize the procedure by also collecting classical shadows on the other $\eta - 1$ registers simultaneously.
One way of interpreting the results we prove in \app{shadows} is that, due to antisymmetry, these registers are anticorrelated.
As a result, collecting shadows on all $\eta$ registers simultaneously reduces the overall cost by at least a factor of $\eta$.
To obtain $W$ elements of the 1-RDM one will need to perform an offline classical inversion of the Clifford channel that will scale as $\widetilde{\cal O}(W \eta^2 / \epsilon^2)$; of course, any quantum or classical algorithm for estimating $W$ quantities must have gate complexity of at least $W$.
However, this only needs to be done once and does not scale in $t$.
As a comparison, the cost of computing 1-RDM classically without exploiting sparsity is $\mathcal O(W\eta)$.

When simulating systems that are well described by mean-field theory, all observables can be efficiently obtained from the time-dependent 1-RDM.
However, for observables such as the energy that have a norm growing in system size or basis size, targeting fixed additive error in the 1-RDM elements will not be sufficient for fixed additive error in the observable.
In such situations, it could be preferable to estimate the observable of interest directly using a combination of block encodings \cite{Low2016} and amplitude amplification \cite{Brassard2002} (see e.g., \cite{Rall2020}).
Assuming the cost of block encoding the observable is negligible to the cost of time-evolution (true for many observables, including energy), this results in needing ${\cal O}(\lambda / \epsilon)$ circuit repetitions, where $\lambda$ is the 1-norm associated with the block encoding of the observable.
For example, whereas there are many correlation functions with $\lambda = {\cal O}(1)$, for the energy $\lambda = {\cal O}(N^{1/3} \eta^{5/3} + N^{2/3}\eta^{1/3})$ \cite{BabbushContinuum}.
Multiplying that to the cost of quantum time-evolution further reduces the quantum speedup.

The final measurement cost to consider is that of resolving observables in time.
In some cases, e.g., when computing scattering cross sections or reaction rates, one might be satisfied measuring the state of the simulation at a single point in time $t$.
However, in other situations, one might wish to simulate time-evolution up to a maximum duration of $t$, but sample quantities at $L$ different points in time.
Most quantum simulation methods that accomplish this goal scale as ${\cal O}(L)$ (${\cal O}(L t)$ in the case where the points are evenly spaced in time).
However, the work of \cite{Huggins2021} shows that this cost can be reduced to ${\cal O}(\sqrt{L} t)$, but with an additional additive space complexity of $\widetilde{\cal O}(L)$.
Either way, this is another cost that plagues quantum but not classical algorithms.

\subsection*{Quantum state preparation costs}

Initial state preparation can be as simple or as complex as the state that one desires to begin the simulation in. Since the focus of this paper is outperforming mean-field calculations, we will discuss the cost of preparing Slater determinants within first quantization. For example, one may wish to start in the Hartree-Fock state (the lowest energy Slater determinant). Classical approaches to computing the Hartree-Fock state scale as roughly $\widetilde{\cal O}(N \eta^2)$ in practice \cite{Manzer2015Jul,Lin2016May}. This is a one-time additive classical cost that is not multiplied by the duration of time-evolution so it is likely subdominant to other costs.

Quantum algorithms for preparing Slater determinants have focused on the ``Givens rotation'' approach introduced in \cite{Kivlichan2018QuantumConnectivity} for second quantization. That algorithm requires ${\cal O}(N\eta)$ ``Givens rotation'' unitaries. Such unitaries can be implemented with ${\cal O}(\eta \log N)$ gates in first quantization \cite{Delgado2022,Su2021}, hence combining that with the sequence of rotations called for in \cite{Kivlichan2018QuantumConnectivity} gives an approach to preparing Slater determinants in first quantization with $\widetilde{\cal O}(N \eta^2)$ gates in total, a relatively high cost. Unlike the offline cost to compute the occupied orbital coefficients, this state preparation cost would be multiplied by the number of measurement repetitions.

Here, we develop a new algorithm to prepare arbitrary Slater determinants in first quantization with only $\widetilde{\cal O}(N \eta)$ gates. The approach is to first generate a superposition of all of the configurations of occupied orbitals in the Slater determinant while making sure that electron registers holding the label of the occupied orbitals are always sorted within each configuration so that they are in ascending order. This is necessary because without such structure (or guarantees of something similar), the next step (anti-symmetrization) could not be reversible. For this next step, we apply the anti-symmetrization procedure introduced in \cite{Berry2018}, which requires only ${\cal O}(\eta \log \eta \log N)$ gates (a negligible additive cost). Note that if one did not need the property that the configurations were ordered by the electron register, then it would be relatively trivial to prepare an arbitrary Slater determinant as a product state of $\eta$ different registers, each in an arbitrary superposition over $\log N$ bits (e.g., using the brute-force state preparation of \cite{Shende2006}).

\begin{table*}[t]
\begin{tabular}{c|c|c|c|c}
Processor
& Algorithm
& Observable
& Space
& Gate complexity\\
\hline\hline
classical 
& $T=0$ mean-field with occ-RI-K/ACE \cite{Manzer2015Jul,Lin2016May}
& anything
& $\widetilde{\cal O}(N \eta)$
& $(N^{4/3} \eta^{7/3}t +  N^{5/3} \eta^{4/3}t)(\frac{N t}{\epsilon})^{o(1)}$\\
classical
& $T>0$ mean-field (density matrix) with \cite{Manzer2015Jul,Lin2016May}
& anything
& $\widetilde{\cal O}(N M)$
& $(N^{4/3} M^2 \eta^{1/3}t \! +\!  \frac{N^{5/3} \! M^2 t } {\eta^{2/3}})(\frac{N t}{\epsilon})^{o(1)}$\\
classical
& $T>0$ mean-field (sampled trajectories) with \cite{Manzer2015Jul,Lin2016May}
& anything
& $\widetilde{\cal O}(N \eta)$
& $(\frac{N^{4/3} \eta^{7/3}t}{\epsilon^2} +  \frac{N^{5/3} \eta^{4/3}t}{\epsilon^2})(\frac{N t}{\epsilon})^{o(1)}$\\
quantum
& second quantized Trotter grid algorithm \cite{Low2022b}
& sample $\ket{\psi(t)}$ 
& ${\cal O}(N \log N)$
&  $(N^{4/3}\eta^{1/3}t + \frac{N^{5/3}t}{\eta^{2/3}})   (\frac{N t}{\epsilon})^{o(1)}$ \\
quantum
& first quantized Trotter grid algorithm here
& sample $\ket{\psi(t)}$ 
& ${\cal O}(\eta \log N)$
&  $(N^{1/3}\eta^{7/3}t +\! N^{2/3}\eta^{4/3}t)   (\frac{N t}{\epsilon})^{o(1)}$ \\
quantum
& interaction picture plane wave algorithm \cite{BabbushContinuum}
& sample $\ket{\psi(t)}$
& ${\cal O}(\eta \log N)$
&  $\widetilde{\cal O}(N^{1/3}\eta^{8/3} t)$ \\
quantum
& grid basis algorithm from Appendix K of \cite{Su2021} 
& sample $\ket{\psi(t)}$ 
& ${\cal O}(\eta \log N)$
&  $\widetilde{\cal O}(N^{1/3}\eta^{8/3} t)$ \\
quantum
& new shadows procedure here
& $k$-RDM$(t)$
& ${\cal O}(\eta \log N)$
&  $\widetilde{\cal O}(k^k \eta^{k} L \, {\cal C}_{\rm samp} / \epsilon^2)$ \\
quantum
& gradient measurement \cite{Huggins2021}
& $\bra{\psi(t)} O  \ket{\psi(t)}$
& $\widetilde{\cal O}(\eta + L)$
&  $\widetilde{\cal O}(\sqrt{L} \, {\cal C}_{\rm samp} \,\lambda  / \epsilon)$ \\
quantum
& gradient measurement \cite{Huggins2021}
& $\bra{\psi(t)} H \ket{\psi(t)}$ 
& $\widetilde{\cal O}(\eta + L)$
&  $\widetilde{\cal O}(\frac{\sqrt{L}  {\cal C}_{\rm samp} t (N^{1/3}\eta^{5/3}\! + N^{2/3}\eta^{1/3})}{\epsilon})$ \\
\hline
\end{tabular}
\caption{\label{tab:comparison} Costs of exact quantum algorithms and mean-field classical algorithms for simulating fermionic dynamics. $N$ is the number of basis functions, $\eta$ is the number of particles, $\epsilon$ is target precision, $M$ is the number of appreciably occupied orbitals in a finite temperature ($T$) simulation ($M \simeq N$ for high $T$), $O$ is any observable having norm $\lambda$ that can be block encoded with cost less than time-evolution, $t$ is the duration of evolution, $L$ is the number of time points at which we wish to resolve quantities and ${\cal C}_{\rm samp}$ is the cost of sampling $\ket{\psi(t)}$ with a quantum algorithm. 
For classical algorithms, gate complexity means the number of floating point operations.
We are not accounting for the additive time-independent costs of state preparation ($\widetilde{\cal O}(\eta N)$ gates using the procedure of \app{state_prep}), of classically computing initial occupied orbital coefficients, or of classically reconstructing the $k$-RDM given measurement outcomes. Thus, this table reports gate complexities for long-time $t$ simulations. In \app{regimes} we provide a table clarifying which algorithm has optimal gate complexity as a function of $N / \eta$.}
\end{table*}

A high level description of how the superposition of ``ordered'' configurations comprising the Slater determinant is prepared now follows, with details given in \app{state_prep}. The idea is to generate the Slater determinant in second quantization in an ancilla register using the Givens rotation approach of \cite{Kivlichan2018QuantumConnectivity}, while mapping the second quantized representation to a first quantized representation one second quantized qubit (orbital) at a time. One can get away with storing only $\eta$ non-zero qubits (orbitals) at a time in the second quantized representation because the Givens rotation algorithm gradually produces qubits that do not require further rotations. Whenever one produces a new qubit in the second quantized representation that does not require further rotations, one can convert it to the first quantized representation, which zeros that qubit. Thus, the procedure only requires ${\cal O}(\eta)$ ancilla qubits -- a negligible additive space overhead. A total of ${\cal O}(N \eta \log N)$ gates are required because for each of ${\cal O}(N)$ steps one accesses all ${\cal O}(\eta \log N)$ qubits of the first quantized representation. In \app{state_prep}, we show the Toffoli complexity can be further reduced to ${\cal O}(N \eta)$ with some additional tricks.

Finally, we note that quantum algorithms can also perform finite temperature simulation by sampling initial states from a thermal density matrix in each realization of the circuit. For example, if the system is in a regime that is well treated by mean-field theory, one can initialize the system in a Slater determinant that is sampled from the thermal Hartree-Fock state \cite{Mermin1963}. Since the output of quantum simulations already needs to be sampled this does not meaningfully increase the number of quantum repetitions required. Such an approach would also be viable classically (and would allow one to perform simulations that only ever treat $\eta$ occupied orbitals despite having finite temperature), but would introduce a multiplicative ${\cal O}(1/\epsilon^2)$ sampling cost. For either processor there is the cost of classically computing the thermal Hartree-Fock state, but this is a one-time cost not multiplied by the duration of time-evolution or ${\cal O}(1/\epsilon^2)$.

\subsection*{Discussion}

We have reviewed and provided new analysis of the costs associated with both classical mean-field methods and state-of-the-art exact quantum algorithms for dynamics. We introduced new and more efficient strategies for initializing Slater determinants in first quantization, and for measuring RDMs via classical shadows. We compare these costs in \tab{comparison}. Relative to classical mean-field methods, we see that when the goal is to sample the output of quantum dynamics at zero temperature, the best quantum algorithms deliver a seventh power speedup in particle number when $N < \Theta(\eta^2)$, quartic in basis size  when $\Theta(\eta^2) < N < \Theta(\eta^3)$, super-quadratic in basis size when $\Theta(\eta^3) < N < \Theta(\eta^4)$ and quintic in basis size but with a quadratic slowdown in $\eta$ when $N > \Theta(\eta^4)$. In the extremal regimes of $N < \Theta(\eta^{5/4})$ and $N > \Theta(\eta^4)$, the overall speedup in system size is super-quadratic (see \app{regimes} for details). These are large enough speedups that quantum advantage may persist even despite quantum error-correction overhead \cite{BabbushFocus}. Note that our analyses are based on derivable upper bounds for both classical and quantum algorithms over all possible input states. Tighter bounds derived over restricted inputs would give asymptotically fewer time steps required for both classical and quantum Trotter algorithms \cite{An2021}.

The story becomes more nuanced when we wish to estimate $\epsilon$-accurate quantities via sampling the quantum simulation output at $L$ different time points. For observables with norm scaling as ${\cal O}(1)$ (e.g., simple correlation functions or single RDM elements), or those pertaining to amplitudes of the state (e.g. scattering amplitudes or reaction rates) the scaling advantages in system and basis size are maintained but at the cost of the quantum algorithm slowing down by a multiplicative factor of at least ${\cal O}(\sqrt{L}/\epsilon)$. When targeting the 1-RDM (which characterizes all observables within mean-field theory) we maintain speedup in $N$ but at the cost of an additional linear slowdown in $\eta$. When measuring the total energy, the overall speedup becomes tenuous. Thus, the viability of quantum advantage with respect to zero temperature classical mean-field methods depends sensitively on the target precision and particular observables of interest.

In terms of applications, we expect RT-TDHF to provide qualitatively correct dynamics whenever electron correlation effects are not pronounced. RT-TDDFT includes some aspects of electron correlation but the adiabatic approximation often creates issues~\cite{Provorse2016May} and the method suffers from self-interaction error \cite{Cohen2008Aug}. When the adiabatic approximation is accurate, self-interaction error is not pronounced, and the system does not exhibit strong correlation, we expect RT-TDDFT to generate qualitatively correct dynamics. When there are many excited states to consider for spectral properties, it is often beneficial to resort to real-time dynamics methods instead of linear-response methods. Furthermore, we are often interested in real-time non-equilibrium electronic dynamics. This is the case for photo-excited molecules near metal surfaces \cite{Tully2000Oct}. The time evolution of electron density (i.e., the diagonal of the 1-RDM) near the molecule is of particular interest due to its implications for chemical reactivity and kinetics in the context of heterogeneous catalysis \cite{Wang2013Nov}.
In this application, the simulation of nuclear degrees of freedom may be equally important, which we will leave for future analysis. 

We see from \tab{comparison} that prospects for quantum advantage are considerably increased at finite temperatures. Thus, a promising class of problems to consider for speedup over mean-field methods is the electronic dynamics of either warm dense matter (WDM) \cite{baczewski2016x,magyar2016stopping,andrade2018negative,ding2018ab} or hot dense matter (HDM) \cite{atzeni2004physics}. The WDM regime (where thermal energy is comparable to the Fermi energy) is typified by temperatures and densities that require the accurate treatment of both quantum and thermal effects~\cite{graziani2014frontiers,dornheim2018uniform}. These conditions occur in planetary interiors, experiments involving high-intensity lasers, and in inertial confinement fusion experiments as the ablator and fuel are compressed into the conditions necessary for thermonuclear ignition. Ignition occurs in the hot dense matter (HDM) regime (where thermal energy far exceeds the Fermi energy). While certain aspects of these systems are conspicuously classical, they still present spectra that can be challenging to model~\cite{bailey2015higher,nagayama2019systematic}.


Simulations in either the WDM or HDM regime typically rely on large plane wave basis sets and the inclusion of ten to one-hundred times more partially occupied orbitals per atom than would be required at lower temperatures. Often, the attendant costs are so great that it is impractical to implement RT-TDDFT with hybrid functionals. Therefore, many calculations necessarily use adiabatic semi-local approximations, even on large classical high-performance computing systems~\cite{baczewski2016x}. Thus, the level of practically achievable accuracy can be quite low, and the prospect of exactly simulating the dynamics on a quantum computer is particularly compelling.

Although we have focused on assessing quantum speedup over mean-field theory, we view the main contribution of this work as more general. In particular, if exact quantum simulations are sometimes more efficient than classical mean-field methods, then all levels of theory in between mean-field and exact diagonalization are in scope for possible quantum advantage. Targeting systems that require more correlated calculations narrows the application space but improves prospects for quantum advantage due to the unfavorable scaling of the requisite classical algorithms. Thus, it may turn out that the domain of systems requiring, say, coupled cluster dynamics~\cite{Huber2011Feb,Sato2018Feb,Shushkov2019Oct,White2019Nov}, might be an even more ideal regime for practical quantum advantage, striking a balance in the trade-off between the breadth of possible applications and the cost of the classical competition.


\subsection*{Acknowledgments} 

The authors thank Alina Kononov, Garnet Kin-Lic Chan, Robin Kothari, Alicia Magann, Fionn Malone, Jarrod McClean, Thomas O'Brien, Nicholas Rubin, Henry Schurkus, Rolando Somma, and Yuan Su for helpful discussions and feedback. We thank Lin Lin for bringing our attention to the quantized tensor train format in \cite{Khoromskaia2011Jan} and thank Yuehaw Khoo for a discussion related to this. DWB worked on this project under a sponsored research agreement with Google Quantum AI. DWB is supported by Australian Research Council Discovery Projects DP190102633 and DP210101367. ADB acknowledges support from the Advanced Simulation and Computing Program and the Sandia LDRD Program. Some work on this project occurred while in residence at The Kavli Institute for Theoretical Physics, supported in part by the National Science Foundation under Grant No. NSF PHY-1748958. Sandia National Laboratories is a multi-mission laboratory managed and operated by National Technology and Engineering Solutions of Sandia, LLC, a wholly owned subsidiary of Honeywell International, Inc., for DOE’s National Nuclear Security Administration under contract DE-NA0003525.

\bibliography{ryan_Mendeley,extra}

\begin{thebibliography}{98}%
\makeatletter
\providecommand \@ifxundefined [1]{%
 \@ifx{#1\undefined}
}%
\providecommand \@ifnum [1]{%
 \ifnum #1\expandafter \@firstoftwo
 \else \expandafter \@secondoftwo
 \fi
}%
\providecommand \@ifx [1]{%
 \ifx #1\expandafter \@firstoftwo
 \else \expandafter \@secondoftwo
 \fi
}%
\providecommand \natexlab [1]{#1}%
\providecommand \enquote  [1]{``#1''}%
\providecommand \bibnamefont  [1]{#1}%
\providecommand \bibfnamefont [1]{#1}%
\providecommand \citenamefont [1]{#1}%
\providecommand \href@noop [0]{\@secondoftwo}%
\providecommand \href [0]{\begingroup \@sanitize@url \@href}%
\providecommand \@href[1]{\@@startlink{#1}\@@href}%
\providecommand \@@href[1]{\endgroup#1\@@endlink}%
\providecommand \@sanitize@url [0]{\catcode `\\12\catcode `\$12\catcode
  `\&12\catcode `\#12\catcode `\^12\catcode `\_12\catcode `\%12\relax}%
\providecommand \@@startlink[1]{}%
\providecommand \@@endlink[0]{}%
\providecommand \url  [0]{\begingroup\@sanitize@url \@url }%
\providecommand \@url [1]{\endgroup\@href {#1}{\urlprefix }}%
\providecommand \urlprefix  [0]{URL }%
\providecommand \Eprint [0]{\href }%
\providecommand \doibase [0]{http://dx.doi.org/}%
\providecommand \selectlanguage [0]{\@gobble}%
\providecommand \bibinfo  [0]{\@secondoftwo}%
\providecommand \bibfield  [0]{\@secondoftwo}%
\providecommand \translation [1]{[#1]}%
\providecommand \BibitemOpen [0]{}%
\providecommand \bibitemStop [0]{}%
\providecommand \bibitemNoStop [0]{.\EOS\space}%
\providecommand \EOS [0]{\spacefactor3000\relax}%
\providecommand \BibitemShut  [1]{\csname bibitem#1\endcsname}%
\let\auto@bib@innerbib\@empty
\bibitem [{\citenamefont {Feynman}(1982)}]{Feynman1982}%
  \BibitemOpen
  \bibfield  {author} {\bibinfo {author} {\bibfnamefont {Richard~P}\
  \bibnamefont {Feynman}},\ }\bibfield  {title} {\enquote {\bibinfo {title}
  {{Simulating physics with computers}},}\ }\href {\doibase 10.1007/BF02650179}
  {\bibfield  {journal} {\bibinfo  {journal} {International Journal of
  Theoretical Physics}\ }\textbf {\bibinfo {volume} {21}},\ \bibinfo {pages}
  {467--488} (\bibinfo {year} {1982})}\BibitemShut {NoStop}%
\bibitem [{\citenamefont {Lloyd}(1996)}]{Lloyd1996}%
  \BibitemOpen
  \bibfield  {author} {\bibinfo {author} {\bibfnamefont {Seth}\ \bibnamefont
  {Lloyd}},\ }\bibfield  {title} {\enquote {\bibinfo {title} {{Universal
  Quantum Simulators}},}\ }\href {\doibase 10.1126/science.273.5278.1073}
  {\bibfield  {journal} {\bibinfo  {journal} {Science}\ }\textbf {\bibinfo
  {volume} {273}},\ \bibinfo {pages} {1073--1078} (\bibinfo {year}
  {1996})}\BibitemShut {NoStop}%
\bibitem [{\citenamefont {Bartlett}\ and\ \citenamefont
  {Musial}(2007)}]{Bartlett2007Feb}%
  \BibitemOpen
  \bibfield  {author} {\bibinfo {author} {\bibfnamefont {Rodney~J.}\
  \bibnamefont {Bartlett}}\ and\ \bibinfo {author} {\bibfnamefont {Monika}\
  \bibnamefont {Musial}},\ }\bibfield  {title} {\enquote {\bibinfo {title}
  {{Coupled-cluster theory in quantum chemistry}},}\ }\href {\doibase
  10.1103/RevModPhys.79.291} {\bibfield  {journal} {\bibinfo  {journal} {Rev.
  Mod. Phys.}\ }\textbf {\bibinfo {volume} {79}},\ \bibinfo {pages} {291--352}
  (\bibinfo {year} {2007})}\BibitemShut {NoStop}%
\bibitem [{\citenamefont {Mardirossian}\ and\ \citenamefont
  {Head-Gordon}(2017)}]{Mardirossian2017Oct}%
  \BibitemOpen
  \bibfield  {author} {\bibinfo {author} {\bibfnamefont {Narbe}\ \bibnamefont
  {Mardirossian}}\ and\ \bibinfo {author} {\bibfnamefont {Martin}\ \bibnamefont
  {Head-Gordon}},\ }\bibfield  {title} {\enquote {\bibinfo {title} {{Thirty
  years of density functional theory in computational chemistry: an overview
  and extensive assessment of 200 density functionals}},}\ }\href {\doibase
  10.1080/00268976.2017.1333644} {\bibfield  {journal} {\bibinfo  {journal}
  {Mol. Phys.}\ }\textbf {\bibinfo {volume} {115}},\ \bibinfo {pages}
  {2315--2372} (\bibinfo {year} {2017})}\BibitemShut {NoStop}%
\bibitem [{\citenamefont {Lee}\ \emph {et~al.}(2022{\natexlab{a}})\citenamefont
  {Lee}, \citenamefont {Pham},\ and\ \citenamefont {Reichman}}]{Lee2022Oct}%
  \BibitemOpen
  \bibfield  {author} {\bibinfo {author} {\bibfnamefont {Joonho}\ \bibnamefont
  {Lee}}, \bibinfo {author} {\bibfnamefont {Hung~Q.}\ \bibnamefont {Pham}}, \
  and\ \bibinfo {author} {\bibfnamefont {David~R.}\ \bibnamefont {Reichman}},\
  }\bibfield  {title} {\enquote {\bibinfo {title} {{Twenty Years of
  Auxiliary-Field Quantum Monte Carlo in Quantum Chemistry: An Overview and
  Assessment on Main Group Chemistry and Bond-Breaking}},}\ }\href {\doibase
  10.1021/acs.jctc.2c00802} {\bibfield  {journal} {\bibinfo  {journal} {J.
  Chem. Theory Comput.}\ }\textbf {\bibinfo {volume} {2022}} (\bibinfo {year}
  {2022}{\natexlab{a}}),\ 10.1021/acs.jctc.2c00802}\BibitemShut {NoStop}%
\bibitem [{\citenamefont {Lee}\ \emph {et~al.}(2022{\natexlab{b}})\citenamefont
  {Lee}, \citenamefont {Lee}, \citenamefont {Zhai}, \citenamefont {Tong},
  \citenamefont {Dalzell}, \citenamefont {Kumar}, \citenamefont {Helms},
  \citenamefont {Gray}, \citenamefont {Cui}, \citenamefont {Liu}, \citenamefont
  {Kastoryano}, \citenamefont {Babbush}, \citenamefont {Preskill},
  \citenamefont {Reichman}, \citenamefont {Campbell}, \citenamefont {Valeev},
  \citenamefont {Lin},\ and\ \citenamefont {Chan}}]{EQA}%
  \BibitemOpen
  \bibfield  {author} {\bibinfo {author} {\bibfnamefont {Seunghoon}\
  \bibnamefont {Lee}}, \bibinfo {author} {\bibfnamefont {Joonho}\ \bibnamefont
  {Lee}}, \bibinfo {author} {\bibfnamefont {Huanchen}\ \bibnamefont {Zhai}},
  \bibinfo {author} {\bibfnamefont {Yu}~\bibnamefont {Tong}}, \bibinfo {author}
  {\bibfnamefont {Alexander~M.}\ \bibnamefont {Dalzell}}, \bibinfo {author}
  {\bibfnamefont {Ashutosh}\ \bibnamefont {Kumar}}, \bibinfo {author}
  {\bibfnamefont {Phillip}\ \bibnamefont {Helms}}, \bibinfo {author}
  {\bibfnamefont {Johnnie}\ \bibnamefont {Gray}}, \bibinfo {author}
  {\bibfnamefont {Zhi-Hao}\ \bibnamefont {Cui}}, \bibinfo {author}
  {\bibfnamefont {Wenyuan}\ \bibnamefont {Liu}}, \bibinfo {author}
  {\bibfnamefont {Michael}\ \bibnamefont {Kastoryano}}, \bibinfo {author}
  {\bibfnamefont {Ryan}\ \bibnamefont {Babbush}}, \bibinfo {author}
  {\bibfnamefont {John}\ \bibnamefont {Preskill}}, \bibinfo {author}
  {\bibfnamefont {David~R.}\ \bibnamefont {Reichman}}, \bibinfo {author}
  {\bibfnamefont {Earl~T.}\ \bibnamefont {Campbell}}, \bibinfo {author}
  {\bibfnamefont {Edward~F.}\ \bibnamefont {Valeev}}, \bibinfo {author}
  {\bibfnamefont {Lin}\ \bibnamefont {Lin}}, \ and\ \bibinfo {author}
  {\bibfnamefont {Garnet Kin-Lic}\ \bibnamefont {Chan}},\ }\bibfield  {title}
  {\enquote {\bibinfo {title} {{Is There Evidence for Exponential Quantum
  Advantage in Quantum Chemistry?}}}\ }\href {\doibase
  10.48550/arxiv.2208.02199} {\bibfield  {journal} {\bibinfo  {journal}
  {arXiv:2208.02199}\ } (\bibinfo {year} {2022}{\natexlab{b}}),\
  10.48550/arxiv.2208.02199}\BibitemShut {NoStop}%
\bibitem [{\citenamefont {Reiher}\ \emph {et~al.}(2017)\citenamefont {Reiher},
  \citenamefont {Wiebe}, \citenamefont {Svore}, \citenamefont {Wecker},\ and\
  \citenamefont {Troyer}}]{Reiher2017}%
  \BibitemOpen
  \bibfield  {author} {\bibinfo {author} {\bibfnamefont {Markus}\ \bibnamefont
  {Reiher}}, \bibinfo {author} {\bibfnamefont {Nathan}\ \bibnamefont {Wiebe}},
  \bibinfo {author} {\bibfnamefont {Krysta~M}\ \bibnamefont {Svore}}, \bibinfo
  {author} {\bibfnamefont {Dave}\ \bibnamefont {Wecker}}, \ and\ \bibinfo
  {author} {\bibfnamefont {Matthias}\ \bibnamefont {Troyer}},\ }\bibfield
  {title} {\enquote {\bibinfo {title} {{Elucidating Reaction Mechanisms on
  Quantum Computers}},}\ }\href
  {http://www.pnas.org/content/114/29/7555.abstract} {\bibfield  {journal}
  {\bibinfo  {journal} {Proceedings of the National Academy of Sciences}\
  }\textbf {\bibinfo {volume} {114}},\ \bibinfo {pages} {7555--7560} (\bibinfo
  {year} {2017})}\BibitemShut {NoStop}%
\bibitem [{\citenamefont {Li}\ \emph {et~al.}(2019)\citenamefont {Li},
  \citenamefont {Li}, \citenamefont {Dattani}, \citenamefont {Umrigar},\ and\
  \citenamefont {Chan}}]{Li2019}%
  \BibitemOpen
  \bibfield  {author} {\bibinfo {author} {\bibfnamefont {Zhendong}\
  \bibnamefont {Li}}, \bibinfo {author} {\bibfnamefont {Junhao}\ \bibnamefont
  {Li}}, \bibinfo {author} {\bibfnamefont {Nikesh~S.}\ \bibnamefont {Dattani}},
  \bibinfo {author} {\bibfnamefont {C.~J.}\ \bibnamefont {Umrigar}}, \ and\
  \bibinfo {author} {\bibfnamefont {Garnet Kin-Lic}\ \bibnamefont {Chan}},\
  }\bibfield  {title} {\enquote {\bibinfo {title} {{The electronic complexity
  of the ground-state of the FeMo cofactor of nitrogenase as relevant to
  quantum simulations}},}\ }\href {\doibase 10.1063/1.5063376} {\bibfield
  {journal} {\bibinfo  {journal} {The Journal of Chemical Physics}\ }\textbf
  {\bibinfo {volume} {150}},\ \bibinfo {pages} {024302} (\bibinfo {year}
  {2019})}\BibitemShut {NoStop}%
\bibitem [{\citenamefont {Berry}\ \emph {et~al.}(2019)\citenamefont {Berry},
  \citenamefont {Gidney}, \citenamefont {Motta}, \citenamefont {McClean},\ and\
  \citenamefont {Babbush}}]{Berry2019}%
  \BibitemOpen
  \bibfield  {author} {\bibinfo {author} {\bibfnamefont {Dominic}\ \bibnamefont
  {Berry}}, \bibinfo {author} {\bibfnamefont {Craig}\ \bibnamefont {Gidney}},
  \bibinfo {author} {\bibfnamefont {Mario}\ \bibnamefont {Motta}}, \bibinfo
  {author} {\bibfnamefont {Jarrod}\ \bibnamefont {McClean}}, \ and\ \bibinfo
  {author} {\bibfnamefont {Ryan}\ \bibnamefont {Babbush}},\ }\bibfield  {title}
  {\enquote {\bibinfo {title} {{Qubitization of Arbitrary Basis Quantum
  Chemistry Leveraging Sparsity and Low Rank Factorization}},}\ }\href
  {https://quantum-journal.org/papers/q-2019-12-02-208/} {\bibfield  {journal}
  {\bibinfo  {journal} {Quantum}\ }\textbf {\bibinfo {volume} {3}},\ \bibinfo
  {pages} {208} (\bibinfo {year} {2019})}\BibitemShut {NoStop}%
\bibitem [{\citenamefont {von Burg}\ \emph {et~al.}(2021)\citenamefont {von
  Burg}, \citenamefont {Low}, \citenamefont {H{\"{a}}ner}, \citenamefont
  {Steiger}, \citenamefont {Reiher}, \citenamefont {Roetteler},\ and\
  \citenamefont {Troyer}}]{vonBurg2020}%
  \BibitemOpen
  \bibfield  {author} {\bibinfo {author} {\bibfnamefont {Vera}\ \bibnamefont
  {von Burg}}, \bibinfo {author} {\bibfnamefont {Guang~Hao}\ \bibnamefont
  {Low}}, \bibinfo {author} {\bibfnamefont {Thomas}\ \bibnamefont
  {H{\"{a}}ner}}, \bibinfo {author} {\bibfnamefont {Damian~S.}\ \bibnamefont
  {Steiger}}, \bibinfo {author} {\bibfnamefont {Markus}\ \bibnamefont
  {Reiher}}, \bibinfo {author} {\bibfnamefont {Martin}\ \bibnamefont
  {Roetteler}}, \ and\ \bibinfo {author} {\bibfnamefont {Matthias}\
  \bibnamefont {Troyer}},\ }\bibfield  {title} {\enquote {\bibinfo {title}
  {{Quantum computing enhanced computational catalysis}},}\ }\href
  {https://journals.aps.org/prresearch/abstract/10.1103/PhysRevResearch.3.033055}
  {\bibfield  {journal} {\bibinfo  {journal} {Physical Review Research}\
  }\textbf {\bibinfo {volume} {3}},\ \bibinfo {pages} {033055--033071}
  (\bibinfo {year} {2021})}\BibitemShut {NoStop}%
\bibitem [{\citenamefont {Lee}\ \emph {et~al.}(2021)\citenamefont {Lee},
  \citenamefont {Berry}, \citenamefont {Gidney}, \citenamefont {Huggins},
  \citenamefont {McClean}, \citenamefont {Wiebe},\ and\ \citenamefont
  {Babbush}}]{Lee2020}%
  \BibitemOpen
  \bibfield  {author} {\bibinfo {author} {\bibfnamefont {Joonho}\ \bibnamefont
  {Lee}}, \bibinfo {author} {\bibfnamefont {Dominic~W.}\ \bibnamefont {Berry}},
  \bibinfo {author} {\bibfnamefont {Craig}\ \bibnamefont {Gidney}}, \bibinfo
  {author} {\bibfnamefont {William~J.}\ \bibnamefont {Huggins}}, \bibinfo
  {author} {\bibfnamefont {Jarrod~R.}\ \bibnamefont {McClean}}, \bibinfo
  {author} {\bibfnamefont {Nathan}\ \bibnamefont {Wiebe}}, \ and\ \bibinfo
  {author} {\bibfnamefont {Ryan}\ \bibnamefont {Babbush}},\ }\bibfield  {title}
  {\enquote {\bibinfo {title} {{Even More Efficient Quantum Computations of
  Chemistry Through Tensor Hypercontraction}},}\ }\href {\doibase
  10.1103/PRXQuantum.2.030305} {\bibfield  {journal} {\bibinfo  {journal} {PRX
  Quantum}\ }\textbf {\bibinfo {volume} {2}},\ \bibinfo {pages} {030305}
  (\bibinfo {year} {2021})}\BibitemShut {NoStop}%
\bibitem [{\citenamefont {Goings}\ \emph {et~al.}(2022)\citenamefont {Goings},
  \citenamefont {White}, \citenamefont {Lee}, \citenamefont {Tautermann},
  \citenamefont {Degroote}, \citenamefont {Gidney}, \citenamefont {Shiozaki},
  \citenamefont {Babbush},\ and\ \citenamefont {Rubin}}]{P450}%
  \BibitemOpen
  \bibfield  {author} {\bibinfo {author} {\bibfnamefont {Joshua~J.}\
  \bibnamefont {Goings}}, \bibinfo {author} {\bibfnamefont {Alec}\ \bibnamefont
  {White}}, \bibinfo {author} {\bibfnamefont {Joonho}\ \bibnamefont {Lee}},
  \bibinfo {author} {\bibfnamefont {Christofer~S.}\ \bibnamefont {Tautermann}},
  \bibinfo {author} {\bibfnamefont {Matthias}\ \bibnamefont {Degroote}},
  \bibinfo {author} {\bibfnamefont {Craig}\ \bibnamefont {Gidney}}, \bibinfo
  {author} {\bibfnamefont {Toru}\ \bibnamefont {Shiozaki}}, \bibinfo {author}
  {\bibfnamefont {Ryan}\ \bibnamefont {Babbush}}, \ and\ \bibinfo {author}
  {\bibfnamefont {Nicholas~C.}\ \bibnamefont {Rubin}},\ }\bibfield  {title}
  {\enquote {\bibinfo {title} {{Reliably Assessing the Electronic Structure of
  Cytochrome P450 on Today’s Classical Computers and Tomorrow’s Quantum
  Computers}},}\ }\href {\doibase 10.1073/pnas.2203533119} {\bibfield
  {journal} {\bibinfo  {journal} {Proceedings of the National Academy of
  Sciences}\ }\textbf {\bibinfo {volume} {119}} (\bibinfo {year} {2022}),\
  10.1073/pnas.2203533119}\BibitemShut {NoStop}%
\bibitem [{\citenamefont {Elfving}\ \emph {et~al.}(2020)\citenamefont
  {Elfving}, \citenamefont {Broer}, \citenamefont {Webber}, \citenamefont
  {Gavartin}, \citenamefont {Halls}, \citenamefont {Lorton},\ and\
  \citenamefont {Bochevarov}}]{Elfving2020}%
  \BibitemOpen
  \bibfield  {author} {\bibinfo {author} {\bibfnamefont {V.~E.}\ \bibnamefont
  {Elfving}}, \bibinfo {author} {\bibfnamefont {B.~W.}\ \bibnamefont {Broer}},
  \bibinfo {author} {\bibfnamefont {M.}~\bibnamefont {Webber}}, \bibinfo
  {author} {\bibfnamefont {J.}~\bibnamefont {Gavartin}}, \bibinfo {author}
  {\bibfnamefont {M.~D.}\ \bibnamefont {Halls}}, \bibinfo {author}
  {\bibfnamefont {K.~P.}\ \bibnamefont {Lorton}}, \ and\ \bibinfo {author}
  {\bibfnamefont {A.}~\bibnamefont {Bochevarov}},\ }\bibfield  {title}
  {\enquote {\bibinfo {title} {{How will quantum computers provide an
  industrially relevant computational advantage in quantum chemistry?}}}\
  }\href {http://arxiv.org/abs/2009.12472} {\  (\bibinfo {year}
  {2020})}\BibitemShut {NoStop}%
\bibitem [{\citenamefont {Babbush}\ \emph
  {et~al.}(2018{\natexlab{a}})\citenamefont {Babbush}, \citenamefont {Wiebe},
  \citenamefont {McClean}, \citenamefont {McClain}, \citenamefont {Neven},\
  and\ \citenamefont {Chan}}]{BabbushLow}%
  \BibitemOpen
  \bibfield  {author} {\bibinfo {author} {\bibfnamefont {Ryan}\ \bibnamefont
  {Babbush}}, \bibinfo {author} {\bibfnamefont {Nathan}\ \bibnamefont {Wiebe}},
  \bibinfo {author} {\bibfnamefont {Jarrod}\ \bibnamefont {McClean}}, \bibinfo
  {author} {\bibfnamefont {James}\ \bibnamefont {McClain}}, \bibinfo {author}
  {\bibfnamefont {Hartmut}\ \bibnamefont {Neven}}, \ and\ \bibinfo {author}
  {\bibfnamefont {Garnet Kin-Lic}\ \bibnamefont {Chan}},\ }\bibfield  {title}
  {\enquote {\bibinfo {title} {{Low-Depth Quantum Simulation of Materials}},}\
  }\href {https://journals.aps.org/prx/abstract/10.1103/PhysRevX.8.011044}
  {\bibfield  {journal} {\bibinfo  {journal} {Physical Review X}\ }\textbf
  {\bibinfo {volume} {8}},\ \bibinfo {pages} {011044} (\bibinfo {year}
  {2018}{\natexlab{a}})}\BibitemShut {NoStop}%
\bibitem [{\citenamefont {Babbush}\ \emph
  {et~al.}(2018{\natexlab{b}})\citenamefont {Babbush}, \citenamefont {Gidney},
  \citenamefont {Berry}, \citenamefont {Wiebe}, \citenamefont {McClean},
  \citenamefont {Paler}, \citenamefont {Fowler},\ and\ \citenamefont
  {Neven}}]{BabbushSpectra}%
  \BibitemOpen
  \bibfield  {author} {\bibinfo {author} {\bibfnamefont {Ryan}\ \bibnamefont
  {Babbush}}, \bibinfo {author} {\bibfnamefont {Craig}\ \bibnamefont {Gidney}},
  \bibinfo {author} {\bibfnamefont {Dominic}\ \bibnamefont {Berry}}, \bibinfo
  {author} {\bibfnamefont {Nathan}\ \bibnamefont {Wiebe}}, \bibinfo {author}
  {\bibfnamefont {Jarrod}\ \bibnamefont {McClean}}, \bibinfo {author}
  {\bibfnamefont {Alexandru}\ \bibnamefont {Paler}}, \bibinfo {author}
  {\bibfnamefont {Austin}\ \bibnamefont {Fowler}}, \ and\ \bibinfo {author}
  {\bibfnamefont {Hartmut}\ \bibnamefont {Neven}},\ }\bibfield  {title}
  {\enquote {\bibinfo {title} {{Encoding Electronic Spectra in Quantum Circuits
  with Linear T Complexity}},}\ }\href
  {https://journals.aps.org/prx/abstract/10.1103/PhysRevX.8.041015} {\bibfield
  {journal} {\bibinfo  {journal} {Physical Review X}\ }\textbf {\bibinfo
  {volume} {8}},\ \bibinfo {pages} {041015} (\bibinfo {year}
  {2018}{\natexlab{b}})}\BibitemShut {NoStop}%
\bibitem [{\citenamefont {Kivlichan}\ \emph {et~al.}(2020)\citenamefont
  {Kivlichan}, \citenamefont {Gidney}, \citenamefont {Berry}, \citenamefont
  {Wiebe}, \citenamefont {McClean}, \citenamefont {Sun}, \citenamefont {Jiang},
  \citenamefont {Rubin}, \citenamefont {Fowler}, \citenamefont {Aspuru-Guzik},
  \citenamefont {Neven},\ and\ \citenamefont {Babbush}}]{Kivlichan2019}%
  \BibitemOpen
  \bibfield  {author} {\bibinfo {author} {\bibfnamefont {Ian~D.}\ \bibnamefont
  {Kivlichan}}, \bibinfo {author} {\bibfnamefont {Craig}\ \bibnamefont
  {Gidney}}, \bibinfo {author} {\bibfnamefont {Dominic~W.}\ \bibnamefont
  {Berry}}, \bibinfo {author} {\bibfnamefont {Nathan}\ \bibnamefont {Wiebe}},
  \bibinfo {author} {\bibfnamefont {Jarrod}\ \bibnamefont {McClean}}, \bibinfo
  {author} {\bibfnamefont {Wei}\ \bibnamefont {Sun}}, \bibinfo {author}
  {\bibfnamefont {Zhang}\ \bibnamefont {Jiang}}, \bibinfo {author}
  {\bibfnamefont {Nicholas}\ \bibnamefont {Rubin}}, \bibinfo {author}
  {\bibfnamefont {Austin}\ \bibnamefont {Fowler}}, \bibinfo {author}
  {\bibfnamefont {Alán}\ \bibnamefont {Aspuru-Guzik}}, \bibinfo {author}
  {\bibfnamefont {Hartmut}\ \bibnamefont {Neven}}, \ and\ \bibinfo {author}
  {\bibfnamefont {Ryan}\ \bibnamefont {Babbush}},\ }\bibfield  {title}
  {\enquote {\bibinfo {title} {{Improved Fault-Tolerant Quantum Simulation of
  Condensed-Phase Correlated Electrons via Trotterization}},}\ }\href
  {https://quantum-journal.org/papers/q-2020-07-16-296/} {\bibfield  {journal}
  {\bibinfo  {journal} {Quantum}\ }\textbf {\bibinfo {volume} {4}},\ \bibinfo
  {pages} {296} (\bibinfo {year} {2020})}\BibitemShut {NoStop}%
\bibitem [{\citenamefont {McArdle}\ \emph {et~al.}(2022)\citenamefont
  {McArdle}, \citenamefont {Campbell},\ and\ \citenamefont {Su}}]{McArdle2022}%
  \BibitemOpen
  \bibfield  {author} {\bibinfo {author} {\bibfnamefont {Sam}\ \bibnamefont
  {McArdle}}, \bibinfo {author} {\bibfnamefont {Earl}\ \bibnamefont
  {Campbell}}, \ and\ \bibinfo {author} {\bibfnamefont {Yuan}\ \bibnamefont
  {Su}},\ }\bibfield  {title} {\enquote {\bibinfo {title} {{Exploiting fermion
  number in factorized decompositions of the electronic structure
  Hamiltonian}},}\ }\href {\doibase 10.1103/PhysRevA.105.012403} {\bibfield
  {journal} {\bibinfo  {journal} {Physical Review A}\ }\textbf {\bibinfo
  {volume} {105}},\ \bibinfo {pages} {012403} (\bibinfo {year}
  {2022})}\BibitemShut {NoStop}%
\bibitem [{\citenamefont {Somma}(2015)}]{Somma2015}%
  \BibitemOpen
  \bibfield  {author} {\bibinfo {author} {\bibfnamefont {Rolando~D.}\
  \bibnamefont {Somma}},\ }\bibfield  {title} {\enquote {\bibinfo {title}
  {{Quantum Simulations of One Dimensional Quantum Systems}},}\ }\href
  {https://arxiv.org/abs/1503.06319} {\bibfield  {journal} {\bibinfo  {journal}
  {arXiv:2203.17006}\ } (\bibinfo {year} {2015})}\BibitemShut {NoStop}%
\bibitem [{\citenamefont {Geller}\ \emph {et~al.}(2015)\citenamefont {Geller},
  \citenamefont {Martinis}, \citenamefont {Sornborger}, \citenamefont
  {Stancil}, \citenamefont {Pritchett}, \citenamefont {You},\ and\
  \citenamefont {Galiautdinov}}]{Geller2015}%
  \BibitemOpen
  \bibfield  {author} {\bibinfo {author} {\bibfnamefont {Michael~R.}\
  \bibnamefont {Geller}}, \bibinfo {author} {\bibfnamefont {John~M.}\
  \bibnamefont {Martinis}}, \bibinfo {author} {\bibfnamefont {Andrew~T.}\
  \bibnamefont {Sornborger}}, \bibinfo {author} {\bibfnamefont {Phillip~C.}\
  \bibnamefont {Stancil}}, \bibinfo {author} {\bibfnamefont {Emily~J.}\
  \bibnamefont {Pritchett}}, \bibinfo {author} {\bibfnamefont {Hao}\
  \bibnamefont {You}}, \ and\ \bibinfo {author} {\bibfnamefont {Andrei}\
  \bibnamefont {Galiautdinov}},\ }\bibfield  {title} {\enquote {\bibinfo
  {title} {{Universal Quantum Simulation with Prethreshold Superconducting
  Qubits: Single-Excitation Subspace Method}},}\ }\href
  {http://arxiv.org/abs/1505.04990} {\bibfield  {journal} {\bibinfo  {journal}
  {arXiv:1505.04990}\ } (\bibinfo {year} {2015})}\BibitemShut {NoStop}%
\bibitem [{\citenamefont {Dreuw}\ and\ \citenamefont
  {Head-Gordon}(2005)}]{Dreuw2005Nov}%
  \BibitemOpen
  \bibfield  {author} {\bibinfo {author} {\bibfnamefont {Andreas}\ \bibnamefont
  {Dreuw}}\ and\ \bibinfo {author} {\bibfnamefont {Martin}\ \bibnamefont
  {Head-Gordon}},\ }\bibfield  {title} {\enquote {\bibinfo {title}
  {{Single-Reference ab Initio Methods for the Calculation of Excited States of
  Large Molecules}},}\ }\href {\doibase 10.1021/cr0505627} {\bibfield
  {journal} {\bibinfo  {journal} {Chem. Rev.}\ }\textbf {\bibinfo {volume}
  {105}},\ \bibinfo {pages} {4009--4037} (\bibinfo {year} {2005})}\BibitemShut
  {NoStop}%
\bibitem [{\citenamefont {Runge}\ and\ \citenamefont
  {Gross}(1984)}]{runge1984density}%
  \BibitemOpen
  \bibfield  {author} {\bibinfo {author} {\bibfnamefont {Erich}\ \bibnamefont
  {Runge}}\ and\ \bibinfo {author} {\bibfnamefont {Eberhard~KU}\ \bibnamefont
  {Gross}},\ }\bibfield  {title} {\enquote {\bibinfo {title}
  {Density-functional theory for time-dependent systems},}\ }\href {\doibase
  10.1103/PhysRevLett.52.997} {\bibfield  {journal} {\bibinfo  {journal}
  {Physical review letters}\ }\textbf {\bibinfo {volume} {52}},\ \bibinfo
  {pages} {997} (\bibinfo {year} {1984})}\BibitemShut {NoStop}%
\bibitem [{\citenamefont {Van~Leeuwen}(1999)}]{van1999mapping}%
  \BibitemOpen
  \bibfield  {author} {\bibinfo {author} {\bibfnamefont {Robert}\ \bibnamefont
  {Van~Leeuwen}},\ }\bibfield  {title} {\enquote {\bibinfo {title} {Mapping
  from densities to potentials in time-dependent density-functional theory},}\
  }\href {\doibase 10.1103/PhysRevLett.82.3863} {\bibfield  {journal} {\bibinfo
   {journal} {Physical review letters}\ }\textbf {\bibinfo {volume} {82}},\
  \bibinfo {pages} {3863} (\bibinfo {year} {1999})}\BibitemShut {NoStop}%
\bibitem [{\citenamefont {Manzer}\ \emph {et~al.}(2015)\citenamefont {Manzer},
  \citenamefont {Horn}, \citenamefont {Mardirossian},\ and\ \citenamefont
  {Head-Gordon}}]{Manzer2015Jul}%
  \BibitemOpen
  \bibfield  {author} {\bibinfo {author} {\bibfnamefont {Samuel}\ \bibnamefont
  {Manzer}}, \bibinfo {author} {\bibfnamefont {Paul~R.}\ \bibnamefont {Horn}},
  \bibinfo {author} {\bibfnamefont {Narbe}\ \bibnamefont {Mardirossian}}, \
  and\ \bibinfo {author} {\bibfnamefont {Martin}\ \bibnamefont {Head-Gordon}},\
  }\bibfield  {title} {\enquote {\bibinfo {title} {{Fast, accurate evaluation
  of exact exchange: The occ-RI-K algorithm}},}\ }\href {\doibase
  10.1063/1.4923369} {\bibfield  {journal} {\bibinfo  {journal} {J. Chem.
  Phys.}\ }\textbf {\bibinfo {volume} {143}},\ \bibinfo {pages} {024113}
  (\bibinfo {year} {2015})}\BibitemShut {NoStop}%
\bibitem [{\citenamefont {Lin}(2016)}]{Lin2016May}%
  \BibitemOpen
  \bibfield  {author} {\bibinfo {author} {\bibfnamefont {Lin}\ \bibnamefont
  {Lin}},\ }\bibfield  {title} {\enquote {\bibinfo {title} {{Adaptively
  Compressed Exchange Operator}},}\ }\href {\doibase 10.1021/acs.jctc.6b00092}
  {\bibfield  {journal} {\bibinfo  {journal} {J. Chem. Theory Comput.}\
  }\textbf {\bibinfo {volume} {12}},\ \bibinfo {pages} {2242--2249} (\bibinfo
  {year} {2016})}\BibitemShut {NoStop}%
\bibitem [{\citenamefont {Jia}\ and\ \citenamefont
  {Lin}(2019{\natexlab{a}})}]{LinTDDFT2019}%
  \BibitemOpen
  \bibfield  {author} {\bibinfo {author} {\bibfnamefont {Weile}\ \bibnamefont
  {Jia}}\ and\ \bibinfo {author} {\bibfnamefont {Lin}\ \bibnamefont {Lin}},\
  }\bibfield  {title} {\enquote {\bibinfo {title} {{Fast real-time
  time-dependent hybrid functional calculations with the parallel transport
  gauge and the adaptively compressed exchange formulation}},}\ }\href
  {\doibase 10.1016/j.cpc.2019.02.009} {\bibfield  {journal} {\bibinfo
  {journal} {Computer Physics Communications}\ }\textbf {\bibinfo {volume}
  {240}},\ \bibinfo {pages} {21--29} (\bibinfo {year}
  {2019}{\natexlab{a}})}\BibitemShut {NoStop}%
\bibitem [{\citenamefont {Jia}\ and\ \citenamefont
  {Lin}(2019{\natexlab{b}})}]{Jia2019Jul}%
  \BibitemOpen
  \bibfield  {author} {\bibinfo {author} {\bibfnamefont {Weile}\ \bibnamefont
  {Jia}}\ and\ \bibinfo {author} {\bibfnamefont {Lin}\ \bibnamefont {Lin}},\
  }\bibfield  {title} {\enquote {\bibinfo {title} {{Fast real-time
  time-dependent hybrid functional calculations with the parallel transport
  gauge and the adaptively compressed exchange formulation}},}\ }\href
  {\doibase 10.1016/j.cpc.2019.02.009} {\bibfield  {journal} {\bibinfo
  {journal} {Comput. Phys. Commun.}\ }\textbf {\bibinfo {volume} {240}},\
  \bibinfo {pages} {21--29} (\bibinfo {year} {2019}{\natexlab{b}})}\BibitemShut
  {NoStop}%
\bibitem [{\citenamefont {Prodan}\ and\ \citenamefont
  {Kohn}(2005)}]{Prodan2005Aug}%
  \BibitemOpen
  \bibfield  {author} {\bibinfo {author} {\bibfnamefont {E.}~\bibnamefont
  {Prodan}}\ and\ \bibinfo {author} {\bibfnamefont {W.}~\bibnamefont {Kohn}},\
  }\bibfield  {title} {\enquote {\bibinfo {title} {{Nearsightedness of
  electronic matter}},}\ }\href {\doibase 10.1073/pnas.0505436102} {\bibfield
  {journal} {\bibinfo  {journal} {Proc. Natl. Acad. Sci. U.S.A.}\ }\textbf
  {\bibinfo {volume} {102}},\ \bibinfo {pages} {11635--11638} (\bibinfo {year}
  {2005})}\BibitemShut {NoStop}%
\bibitem [{\citenamefont {Kussmann}\ \emph {et~al.}(2013)\citenamefont
  {Kussmann}, \citenamefont {Beer},\ and\ \citenamefont
  {Ochsenfeld}}]{Kussmann2013Nov}%
  \BibitemOpen
  \bibfield  {author} {\bibinfo {author} {\bibfnamefont
  {J{\ifmmode\ddot{o}\else\"{o}\fi}rg}\ \bibnamefont {Kussmann}}, \bibinfo
  {author} {\bibfnamefont {Matthias}\ \bibnamefont {Beer}}, \ and\ \bibinfo
  {author} {\bibfnamefont {Christian}\ \bibnamefont {Ochsenfeld}},\ }\bibfield
  {title} {\enquote {\bibinfo {title} {{Linear-scaling self-consistent field
  methods for large molecules}},}\ }\href {\doibase 10.1002/wcms.1138}
  {\bibfield  {journal} {\bibinfo  {journal} {WIREs Comput. Mol. Sci.}\
  }\textbf {\bibinfo {volume} {3}},\ \bibinfo {pages} {614--636} (\bibinfo
  {year} {2013})}\BibitemShut {NoStop}%
\bibitem [{\citenamefont {O{'}Rourke}\ and\ \citenamefont
  {Bowler}(2015)}]{ORourke2015Sep}%
  \BibitemOpen
  \bibfield  {author} {\bibinfo {author} {\bibfnamefont {Conn}\ \bibnamefont
  {O{'}Rourke}}\ and\ \bibinfo {author} {\bibfnamefont {David~R.}\ \bibnamefont
  {Bowler}},\ }\bibfield  {title} {\enquote {\bibinfo {title} {{Linear scaling
  density matrix real time TDDFT: Propagator unitarity and matrix
  truncation}},}\ }\href {\doibase 10.1063/1.4919128} {\bibfield  {journal}
  {\bibinfo  {journal} {J. Chem. Phys.}\ }\textbf {\bibinfo {volume} {143}},\
  \bibinfo {pages} {102801} (\bibinfo {year} {2015})}\BibitemShut {NoStop}%
\bibitem [{\citenamefont {Zuehlsdorff}\ \emph {et~al.}(2013)\citenamefont
  {Zuehlsdorff}, \citenamefont {Hine}, \citenamefont {Spencer}, \citenamefont
  {Harrison}, \citenamefont {Riley},\ and\ \citenamefont
  {Haynes}}]{Zuehlsdorff2013Aug}%
  \BibitemOpen
  \bibfield  {author} {\bibinfo {author} {\bibfnamefont {T.~J.}\ \bibnamefont
  {Zuehlsdorff}}, \bibinfo {author} {\bibfnamefont {N.~D.~M.}\ \bibnamefont
  {Hine}}, \bibinfo {author} {\bibfnamefont {J.~S.}\ \bibnamefont {Spencer}},
  \bibinfo {author} {\bibfnamefont {N.~M.}\ \bibnamefont {Harrison}}, \bibinfo
  {author} {\bibfnamefont {D.~J.}\ \bibnamefont {Riley}}, \ and\ \bibinfo
  {author} {\bibfnamefont {P.~D.}\ \bibnamefont {Haynes}},\ }\bibfield  {title}
  {\enquote {\bibinfo {title} {{Linear-scaling time-dependent
  density-functional theory in the linear response formalism}},}\ }\href
  {\doibase 10.1063/1.4817330} {\bibfield  {journal} {\bibinfo  {journal} {J.
  Chem. Phys.}\ }\textbf {\bibinfo {volume} {139}},\ \bibinfo {pages} {064104}
  (\bibinfo {year} {2013})}\BibitemShut {NoStop}%
\bibitem [{\citenamefont {Khoromskaia}\ \emph {et~al.}(2011)\citenamefont
  {Khoromskaia}, \citenamefont {Khoromskij},\ and\ \citenamefont
  {Schneider}}]{Khoromskaia2011Jan}%
  \BibitemOpen
  \bibfield  {author} {\bibinfo {author} {\bibfnamefont {Venera}\ \bibnamefont
  {Khoromskaia}}, \bibinfo {author} {\bibfnamefont {Boris}\ \bibnamefont
  {Khoromskij}}, \ and\ \bibinfo {author} {\bibfnamefont {Reinhold}\
  \bibnamefont {Schneider}},\ }\bibfield  {title} {\enquote {\bibinfo {title}
  {{QTT Representation of the Hartree and Exchange Operators in Electronic
  Structure Calculations}},}\ }\href {\doibase 10.2478/cmam-2011-0018}
  {\bibfield  {journal} {\bibinfo  {journal} {Comput. Methods Appl. Math.}\
  }\textbf {\bibinfo {volume} {11}},\ \bibinfo {pages} {327--341} (\bibinfo
  {year} {2011})}\BibitemShut {NoStop}%
\bibitem [{\citenamefont {Castro}\ \emph {et~al.}(2004)\citenamefont {Castro},
  \citenamefont {Marques},\ and\ \citenamefont
  {Rubio}}]{castro2004propagators}%
  \BibitemOpen
  \bibfield  {author} {\bibinfo {author} {\bibfnamefont {Alberto}\ \bibnamefont
  {Castro}}, \bibinfo {author} {\bibfnamefont {Miguel~AL}\ \bibnamefont
  {Marques}}, \ and\ \bibinfo {author} {\bibfnamefont {Angel}\ \bibnamefont
  {Rubio}},\ }\bibfield  {title} {\enquote {\bibinfo {title} {{Propagators for
  the time-dependent Kohn--Sham equations}},}\ }\href {\doibase
  10.1063/1.1774980} {\bibfield  {journal} {\bibinfo  {journal} {The Journal of
  chemical physics}\ }\textbf {\bibinfo {volume} {121}},\ \bibinfo {pages}
  {3425--3433} (\bibinfo {year} {2004})}\BibitemShut {NoStop}%
\bibitem [{\citenamefont {Jia}\ \emph {et~al.}(2018)\citenamefont {Jia},
  \citenamefont {An}, \citenamefont {Wang},\ and\ \citenamefont
  {Lin}}]{jia2018fast}%
  \BibitemOpen
  \bibfield  {author} {\bibinfo {author} {\bibfnamefont {Weile}\ \bibnamefont
  {Jia}}, \bibinfo {author} {\bibfnamefont {Dong}\ \bibnamefont {An}}, \bibinfo
  {author} {\bibfnamefont {Lin-Wang}\ \bibnamefont {Wang}}, \ and\ \bibinfo
  {author} {\bibfnamefont {Lin}\ \bibnamefont {Lin}},\ }\bibfield  {title}
  {\enquote {\bibinfo {title} {Fast real-time time-dependent density functional
  theory calculations with the parallel transport gauge},}\ }\href {\doibase
  10.1021/acs.jctc.8b00580} {\bibfield  {journal} {\bibinfo  {journal} {Journal
  of Chemical Theory and Computation}\ }\textbf {\bibinfo {volume} {14}},\
  \bibinfo {pages} {5645--5652} (\bibinfo {year} {2018})}\BibitemShut {NoStop}%
\bibitem [{\citenamefont {Kononov}\ \emph {et~al.}(2022)\citenamefont
  {Kononov}, \citenamefont {Lee}, \citenamefont {dos Santos}, \citenamefont
  {Robinson}, \citenamefont {Yao}, \citenamefont {Yao}, \citenamefont
  {Andrade}, \citenamefont {Baczewski}, \citenamefont {Constantinescu},
  \citenamefont {Correa}, \citenamefont {Kanai}, \citenamefont {Modine},\ and\
  \citenamefont {Schleife}}]{kononov2022electron}%
  \BibitemOpen
  \bibfield  {author} {\bibinfo {author} {\bibfnamefont {Alina}\ \bibnamefont
  {Kononov}}, \bibinfo {author} {\bibfnamefont {Cheng-Wei}\ \bibnamefont
  {Lee}}, \bibinfo {author} {\bibfnamefont {Tatiane~Pereira}\ \bibnamefont {dos
  Santos}}, \bibinfo {author} {\bibfnamefont {Brian}\ \bibnamefont {Robinson}},
  \bibinfo {author} {\bibfnamefont {Yifan}\ \bibnamefont {Yao}}, \bibinfo
  {author} {\bibfnamefont {Yi}~\bibnamefont {Yao}}, \bibinfo {author}
  {\bibfnamefont {Xavier}\ \bibnamefont {Andrade}}, \bibinfo {author}
  {\bibfnamefont {Andrew~David}\ \bibnamefont {Baczewski}}, \bibinfo {author}
  {\bibfnamefont {Emil}\ \bibnamefont {Constantinescu}}, \bibinfo {author}
  {\bibfnamefont {Alfredo~A}\ \bibnamefont {Correa}}, \bibinfo {author}
  {\bibfnamefont {Yosuke}\ \bibnamefont {Kanai}}, \bibinfo {author}
  {\bibfnamefont {Normand}\ \bibnamefont {Modine}}, \ and\ \bibinfo {author}
  {\bibfnamefont {Andr\'e}\ \bibnamefont {Schleife}},\ }\bibfield  {title}
  {\enquote {\bibinfo {title} {Electron dynamics in extended systems within
  real-time time-dependent density-functional theory},}\ }\href {\doibase
  10.1557/s43579-022-00273-7} {\bibfield  {journal} {\bibinfo  {journal} {MRS
  Communications}\ ,\ \bibinfo {pages} {1--13}} (\bibinfo {year}
  {2022})}\BibitemShut {NoStop}%
\bibitem [{\citenamefont {Shepard}\ \emph {et~al.}(2021)\citenamefont
  {Shepard}, \citenamefont {Zhou}, \citenamefont {Yost}, \citenamefont {Yao},\
  and\ \citenamefont {Kanai}}]{Shepard2021Sep}%
  \BibitemOpen
  \bibfield  {author} {\bibinfo {author} {\bibfnamefont {Christopher}\
  \bibnamefont {Shepard}}, \bibinfo {author} {\bibfnamefont {Ruiyi}\
  \bibnamefont {Zhou}}, \bibinfo {author} {\bibfnamefont {Dillon~C.}\
  \bibnamefont {Yost}}, \bibinfo {author} {\bibfnamefont {Yi}~\bibnamefont
  {Yao}}, \ and\ \bibinfo {author} {\bibfnamefont {Yosuke}\ \bibnamefont
  {Kanai}},\ }\bibfield  {title} {\enquote {\bibinfo {title} {{Simulating
  electronic excitation and dynamics with real-time propagation approach to
  TDDFT within plane-wave pseudopotential formulation}},}\ }\href {\doibase
  10.1063/5.0057587} {\bibfield  {journal} {\bibinfo  {journal} {J. Chem.
  Phys.}\ }\textbf {\bibinfo {volume} {155}},\ \bibinfo {pages} {100901}
  (\bibinfo {year} {2021})}\BibitemShut {NoStop}%
\bibitem [{\citenamefont {Shavitt}\ and\ \citenamefont
  {Bartlett}(2009)}]{Shavitt2009Aug}%
  \BibitemOpen
  \bibfield  {author} {\bibinfo {author} {\bibfnamefont {Isaiah}\ \bibnamefont
  {Shavitt}}\ and\ \bibinfo {author} {\bibfnamefont {Rodney~J.}\ \bibnamefont
  {Bartlett}},\ }\href {\doibase 10.1017/CBO9780511596834} {\emph {\bibinfo
  {title} {{Many-Body Methods in Chemistry and Physics: MBPT and
  Coupled-Cluster Theory}}}}\ (\bibinfo  {publisher} {Cambridge University
  Press},\ \bibinfo {address} {Cambridge, England, UK},\ \bibinfo {year}
  {2009})\BibitemShut {NoStop}%
\bibitem [{\citenamefont {Wiesner}(1996)}]{Wiesner1996SimulationsComputer}%
  \BibitemOpen
  \bibfield  {author} {\bibinfo {author} {\bibfnamefont {Stephen}\ \bibnamefont
  {Wiesner}},\ }\bibfield  {title} {\enquote {\bibinfo {title} {{Simulations of
  Many-Body Quantum Systems by a Quantum Computer}},}\ }\href
  {https://arxiv.org/abs/quant-ph/9603028} {\bibfield  {journal} {\bibinfo
  {journal} {arXiv:quant-ph/9603028}\ } (\bibinfo {year} {1996})}\BibitemShut
  {NoStop}%
\bibitem [{\citenamefont {Abrams}\ and\ \citenamefont
  {Lloyd}(1997)}]{Abrams1997}%
  \BibitemOpen
  \bibfield  {author} {\bibinfo {author} {\bibfnamefont {Daniel~S}\
  \bibnamefont {Abrams}}\ and\ \bibinfo {author} {\bibfnamefont {Seth}\
  \bibnamefont {Lloyd}},\ }\bibfield  {title} {\enquote {\bibinfo {title}
  {{Simulation of Many-Body Fermi Systems on a Universal Quantum Computer}},}\
  }\href {https://journals.aps.org/prl/pdf/10.1103/PhysRevLett.79.2586}
  {\bibfield  {journal} {\bibinfo  {journal} {Physical Review Letters}\
  }\textbf {\bibinfo {volume} {79}},\ \bibinfo {pages} {4} (\bibinfo {year}
  {1997})}\BibitemShut {NoStop}%
\bibitem [{\citenamefont {Zalka}(1998)}]{Zalka1998}%
  \BibitemOpen
  \bibfield  {author} {\bibinfo {author} {\bibfnamefont {Christof}\
  \bibnamefont {Zalka}},\ }\bibfield  {title} {\enquote {\bibinfo {title}
  {{Efficient Simulation of Quantum Systems by Quantum Computers}},}\ }\href
  {\doibase 10.1002/(SICI)1521-3978(199811)46:6/8<877::AID-PROP877>3.0.CO;2-A}
  {\bibfield  {journal} {\bibinfo  {journal} {Fortschritte der Physik}\
  }\textbf {\bibinfo {volume} {46}},\ \bibinfo {pages} {877--879} (\bibinfo
  {year} {1998})}\BibitemShut {NoStop}%
\bibitem [{\citenamefont {Boghosian}\ and\ \citenamefont
  {Taylor}(1998)}]{Boghosian1998}%
  \BibitemOpen
  \bibfield  {author} {\bibinfo {author} {\bibfnamefont {Bruce~M}\ \bibnamefont
  {Boghosian}}\ and\ \bibinfo {author} {\bibfnamefont {Washington}\
  \bibnamefont {Taylor}},\ }\bibfield  {title} {\enquote {\bibinfo {title}
  {{Simulating quantum mechanics on a quantum computer}},}\ }\href {\doibase
  10.1016/S0167-2789(98)00042-6} {\bibfield  {journal} {\bibinfo  {journal}
  {Physica D-Nonlinear Phenomena}\ }\textbf {\bibinfo {volume} {120}},\
  \bibinfo {pages} {30--42} (\bibinfo {year} {1998})}\BibitemShut {NoStop}%
\bibitem [{\citenamefont {Lidar}\ and\ \citenamefont {Wang}(1999)}]{Lidar1999}%
  \BibitemOpen
  \bibfield  {author} {\bibinfo {author} {\bibfnamefont {Daniel~A}\
  \bibnamefont {Lidar}}\ and\ \bibinfo {author} {\bibfnamefont {Haobin}\
  \bibnamefont {Wang}},\ }\bibfield  {title} {\enquote {\bibinfo {title}
  {{Calculating the thermal rate constant with exponential speedup on a quantum
  computer}},}\ }\href {\doibase 10.1103/PhysRevE.59.2429} {\bibfield
  {journal} {\bibinfo  {journal} {Physical Review E}\ }\textbf {\bibinfo
  {volume} {59}},\ \bibinfo {pages} {2429--2438} (\bibinfo {year}
  {1999})}\BibitemShut {NoStop}%
\bibitem [{\citenamefont {Kassal}\ \emph {et~al.}(2008)\citenamefont {Kassal},
  \citenamefont {Jordan}, \citenamefont {Love}, \citenamefont {Mohseni},\ and\
  \citenamefont {Aspuru-Guzik}}]{Kassal2008}%
  \BibitemOpen
  \bibfield  {author} {\bibinfo {author} {\bibfnamefont {Ivan}\ \bibnamefont
  {Kassal}}, \bibinfo {author} {\bibfnamefont {Stephen~P}\ \bibnamefont
  {Jordan}}, \bibinfo {author} {\bibfnamefont {Peter~J}\ \bibnamefont {Love}},
  \bibinfo {author} {\bibfnamefont {Masoud}\ \bibnamefont {Mohseni}}, \ and\
  \bibinfo {author} {\bibfnamefont {Alan}\ \bibnamefont {Aspuru-Guzik}},\
  }\bibfield  {title} {\enquote {\bibinfo {title} {{Polynomial-time quantum
  algorithm for the simulation of chemical dynamics}},}\ }\href
  {http://www.pnas.org/content/105/48/18681.abstract} {\bibfield  {journal}
  {\bibinfo  {journal} {Proceedings of the National Academy of Sciences}\
  }\textbf {\bibinfo {volume} {105}},\ \bibinfo {pages} {18681--18686}
  (\bibinfo {year} {2008})}\BibitemShut {NoStop}%
\bibitem [{\citenamefont {Childs}\ and\ \citenamefont {Su}(2019)}]{Childs2019}%
  \BibitemOpen
  \bibfield  {author} {\bibinfo {author} {\bibfnamefont {Andrew}\ \bibnamefont
  {Childs}}\ and\ \bibinfo {author} {\bibfnamefont {Yuan}\ \bibnamefont {Su}},\
  }\bibfield  {title} {\enquote {\bibinfo {title} {{Nearly optimal lattice
  simulation by product formulas}},}\ }\href {\doibase
  10.1103/PhysRevLett.123.050503} {\bibfield  {journal} {\bibinfo  {journal}
  {Physical Review Letters}\ }\textbf {\bibinfo {volume} {123}},\ \bibinfo
  {pages} {050503} (\bibinfo {year} {2019})}\BibitemShut {NoStop}%
\bibitem [{\citenamefont {Su}\ \emph {et~al.}(2021{\natexlab{a}})\citenamefont
  {Su}, \citenamefont {Huang},\ and\ \citenamefont {Campbell}}]{Su2020}%
  \BibitemOpen
  \bibfield  {author} {\bibinfo {author} {\bibfnamefont {Yuan}\ \bibnamefont
  {Su}}, \bibinfo {author} {\bibfnamefont {Hsin-Yuan}\ \bibnamefont {Huang}}, \
  and\ \bibinfo {author} {\bibfnamefont {Earl~T.}\ \bibnamefont {Campbell}},\
  }\bibfield  {title} {\enquote {\bibinfo {title} {{Nearly tight Trotterization
  of interacting electrons}},}\ }\href {\doibase 10.22331/q-2021-07-05-495}
  {\bibfield  {journal} {\bibinfo  {journal} {Quantum}\ }\textbf {\bibinfo
  {volume} {5}},\ \bibinfo {pages} {495} (\bibinfo {year}
  {2021}{\natexlab{a}})}\BibitemShut {NoStop}%
\bibitem [{\citenamefont {Low}\ \emph {et~al.}(2022)\citenamefont {Low},
  \citenamefont {Su}, \citenamefont {Tong},\ and\ \citenamefont
  {Tran}}]{Low2022b}%
  \BibitemOpen
  \bibfield  {author} {\bibinfo {author} {\bibfnamefont {Guang~Hao}\
  \bibnamefont {Low}}, \bibinfo {author} {\bibfnamefont {Yuan}\ \bibnamefont
  {Su}}, \bibinfo {author} {\bibfnamefont {Yu}~\bibnamefont {Tong}}, \ and\
  \bibinfo {author} {\bibfnamefont {Minh~C.}\ \bibnamefont {Tran}},\ }\bibfield
   {title} {\enquote {\bibinfo {title} {{On the complexity of implementing
  Trotter steps}},}\ }\href {\doibase 10.48550/arxiv.2211.09133} {\  (\bibinfo
  {year} {2022}),\ 10.48550/arxiv.2211.09133}\BibitemShut {NoStop}%
\bibitem [{\citenamefont {Babbush}\ \emph {et~al.}(2019)\citenamefont
  {Babbush}, \citenamefont {Berry}, \citenamefont {McClean},\ and\
  \citenamefont {Neven}}]{BabbushContinuum}%
  \BibitemOpen
  \bibfield  {author} {\bibinfo {author} {\bibfnamefont {Ryan}\ \bibnamefont
  {Babbush}}, \bibinfo {author} {\bibfnamefont {Dominic~W.}\ \bibnamefont
  {Berry}}, \bibinfo {author} {\bibfnamefont {Jarrod~R.}\ \bibnamefont
  {McClean}}, \ and\ \bibinfo {author} {\bibfnamefont {Hartmut}\ \bibnamefont
  {Neven}},\ }\bibfield  {title} {\enquote {\bibinfo {title} {{Quantum
  Simulation of Chemistry with Sublinear Scaling in Basis Size}},}\ }\href
  {https://www.nature.com/articles/s41534-019-0199-y} {\bibfield  {journal}
  {\bibinfo  {journal} {npj Quantum Information}\ }\textbf {\bibinfo {volume}
  {5}},\ \bibinfo {pages} {92} (\bibinfo {year} {2019})}\BibitemShut {NoStop}%
\bibitem [{\citenamefont {Low}\ and\ \citenamefont {Wiebe}(2018)}]{Low2018}%
  \BibitemOpen
  \bibfield  {author} {\bibinfo {author} {\bibfnamefont {Guang~Hao}\
  \bibnamefont {Low}}\ and\ \bibinfo {author} {\bibfnamefont {Nathan}\
  \bibnamefont {Wiebe}},\ }\bibfield  {title} {\enquote {\bibinfo {title}
  {{Hamiltonian Simulation in the Interaction Picture}},}\ }\href
  {http://arxiv.org/abs/1805.00675} {\bibfield  {journal} {\bibinfo  {journal}
  {arXiv:1805.00675}\ } (\bibinfo {year} {2018})}\BibitemShut {NoStop}%
\bibitem [{\citenamefont {Su}\ \emph {et~al.}(2021{\natexlab{b}})\citenamefont
  {Su}, \citenamefont {Berry}, \citenamefont {Wiebe}, \citenamefont {Rubin},\
  and\ \citenamefont {Babbush}}]{Su2021}%
  \BibitemOpen
  \bibfield  {author} {\bibinfo {author} {\bibfnamefont {Yuan}\ \bibnamefont
  {Su}}, \bibinfo {author} {\bibfnamefont {Dominic}\ \bibnamefont {Berry}},
  \bibinfo {author} {\bibfnamefont {Nathan}\ \bibnamefont {Wiebe}}, \bibinfo
  {author} {\bibfnamefont {Nicholas}\ \bibnamefont {Rubin}}, \ and\ \bibinfo
  {author} {\bibfnamefont {Ryan}\ \bibnamefont {Babbush}},\ }\bibfield  {title}
  {\enquote {\bibinfo {title} {{Fault-tolerant quantum simulations of chemistry
  in first quantization}},}\ }\href
  {https://journals.aps.org/prxquantum/abstract/10.1103/PRXQuantum.2.040332}
  {\bibfield  {journal} {\bibinfo  {journal} {PRX Quantum}\ }\textbf {\bibinfo
  {volume} {4}},\ \bibinfo {pages} {040332} (\bibinfo {year}
  {2021}{\natexlab{b}})}\BibitemShut {NoStop}%
\bibitem [{\citenamefont {Aspuru-Guzik}\ \emph {et~al.}(2005)\citenamefont
  {Aspuru-Guzik}, \citenamefont {Dutoi}, \citenamefont {Love},\ and\
  \citenamefont {Head-Gordon}}]{Aspuru-Guzik2005}%
  \BibitemOpen
  \bibfield  {author} {\bibinfo {author} {\bibfnamefont {Alan}\ \bibnamefont
  {Aspuru-Guzik}}, \bibinfo {author} {\bibfnamefont {Anthony~D}\ \bibnamefont
  {Dutoi}}, \bibinfo {author} {\bibfnamefont {Peter~J}\ \bibnamefont {Love}}, \
  and\ \bibinfo {author} {\bibfnamefont {Martin}\ \bibnamefont {Head-Gordon}},\
  }\bibfield  {title} {\enquote {\bibinfo {title} {{Simulated Quantum
  Computation of Molecular Energies}},}\ }\href {\doibase
  10.1126/science.1113479} {\bibfield  {journal} {\bibinfo  {journal}
  {Science}\ }\textbf {\bibinfo {volume} {309}},\ \bibinfo {pages} {1704}
  (\bibinfo {year} {2005})}\BibitemShut {NoStop}%
\bibitem [{\citenamefont {Low}\ and\ \citenamefont {Chuang}(2019)}]{Low2016}%
  \BibitemOpen
  \bibfield  {author} {\bibinfo {author} {\bibfnamefont {Guang~Hao}\
  \bibnamefont {Low}}\ and\ \bibinfo {author} {\bibfnamefont {Isaac~L}\
  \bibnamefont {Chuang}},\ }\bibfield  {title} {\enquote {\bibinfo {title}
  {{Hamiltonian Simulation by Qubitization}},}\ }\href
  {https://doi.org/10.22331/q-2019-07-12-163} {\bibfield  {journal} {\bibinfo
  {journal} {Quantum}\ }\textbf {\bibinfo {volume} {3}},\ \bibinfo {pages}
  {163} (\bibinfo {year} {2019})}\BibitemShut {NoStop}%
\bibitem [{\citenamefont {Rokhlin}(1985)}]{Rokhlin1985}%
  \BibitemOpen
  \bibfield  {author} {\bibinfo {author} {\bibfnamefont {V}~\bibnamefont
  {Rokhlin}},\ }\bibfield  {title} {\enquote {\bibinfo {title} {{Rapid solution
  of integral equations of classical potential theory}},}\ }\href {\doibase
  10.1016/0021-9991(85)90002-6} {\bibfield  {journal} {\bibinfo  {journal}
  {Journal of Computational Physics}\ }\textbf {\bibinfo {volume} {60}},\
  \bibinfo {pages} {187--207} (\bibinfo {year} {1985})}\BibitemShut {NoStop}%
\bibitem [{\citenamefont {Barnes}\ and\ \citenamefont {Hut}(1986)}]{BarnesHut}%
  \BibitemOpen
  \bibfield  {author} {\bibinfo {author} {\bibfnamefont {Josh}\ \bibnamefont
  {Barnes}}\ and\ \bibinfo {author} {\bibfnamefont {Piet}\ \bibnamefont
  {Hut}},\ }\bibfield  {title} {\enquote {\bibinfo {title} {{A hierarchical O(N
  log N) force-calculation algorithm}},}\ }\href {\doibase 10.1038/324446a0}
  {\bibfield  {journal} {\bibinfo  {journal} {Nature}\ }\textbf {\bibinfo
  {volume} {324}},\ \bibinfo {pages} {446--449} (\bibinfo {year}
  {1986})}\BibitemShut {NoStop}%
\bibitem [{\citenamefont {Darden}\ \emph {et~al.}(1993)\citenamefont {Darden},
  \citenamefont {York},\ and\ \citenamefont {Pedersen}}]{Ewald1993}%
  \BibitemOpen
  \bibfield  {author} {\bibinfo {author} {\bibfnamefont {Tom}\ \bibnamefont
  {Darden}}, \bibinfo {author} {\bibfnamefont {Darrin}\ \bibnamefont {York}}, \
  and\ \bibinfo {author} {\bibfnamefont {Lee}\ \bibnamefont {Pedersen}},\
  }\bibfield  {title} {\enquote {\bibinfo {title} {{Particle mesh Ewald: An N
  log N method for Ewald sums in large systems}},}\ }\href {\doibase
  10.1063/1.464397} {\bibfield  {journal} {\bibinfo  {journal} {The Journal of
  Chemical Physics}\ }\textbf {\bibinfo {volume} {98}},\ \bibinfo {pages}
  {10089--10092} (\bibinfo {year} {1993})}\BibitemShut {NoStop}%
\bibitem [{\citenamefont {Childs}\ \emph {et~al.}(2022)\citenamefont {Childs},
  \citenamefont {Leng}, \citenamefont {Li}, \citenamefont {Liu},\ and\
  \citenamefont {Zhang}}]{Childs2022}%
  \BibitemOpen
  \bibfield  {author} {\bibinfo {author} {\bibfnamefont {Andrew~M.}\
  \bibnamefont {Childs}}, \bibinfo {author} {\bibfnamefont {Jiaqi}\
  \bibnamefont {Leng}}, \bibinfo {author} {\bibfnamefont {Tongyang}\
  \bibnamefont {Li}}, \bibinfo {author} {\bibfnamefont {Jin-Peng}\ \bibnamefont
  {Liu}}, \ and\ \bibinfo {author} {\bibfnamefont {Chenyi}\ \bibnamefont
  {Zhang}},\ }\bibfield  {title} {\enquote {\bibinfo {title} {{Quantum
  Simulation of Real-Space Dynamics}},}\ }\href
  {https://arxiv.org/abs/2203.17006} {\bibfield  {journal} {\bibinfo  {journal}
  {arXiv:2203.17006}\ } (\bibinfo {year} {2022})}\BibitemShut {NoStop}%
\bibitem [{\citenamefont {Giovannetti}\ \emph {et~al.}(2008)\citenamefont
  {Giovannetti}, \citenamefont {Lloyd},\ and\ \citenamefont
  {Maccone}}]{Giovannetti2008}%
  \BibitemOpen
  \bibfield  {author} {\bibinfo {author} {\bibfnamefont {Vittorio}\
  \bibnamefont {Giovannetti}}, \bibinfo {author} {\bibfnamefont {Seth}\
  \bibnamefont {Lloyd}}, \ and\ \bibinfo {author} {\bibfnamefont {Lorenzo}\
  \bibnamefont {Maccone}},\ }\bibfield  {title} {\enquote {\bibinfo {title}
  {{Quantum Random Access Memory}},}\ }\href {\doibase
  10.1103/PhysRevLett.100.160501} {\bibfield  {journal} {\bibinfo  {journal}
  {Physical Review Letters}\ }\textbf {\bibinfo {volume} {100}},\ \bibinfo
  {pages} {160501} (\bibinfo {year} {2008})}\BibitemShut {NoStop}%
\bibitem [{\citenamefont {Huggins}\ \emph {et~al.}(2021)\citenamefont
  {Huggins}, \citenamefont {Wan}, \citenamefont {McClean}, \citenamefont
  {O'Brien}, \citenamefont {Wiebe},\ and\ \citenamefont
  {Babbush}}]{Huggins2021}%
  \BibitemOpen
  \bibfield  {author} {\bibinfo {author} {\bibfnamefont {William~J.}\
  \bibnamefont {Huggins}}, \bibinfo {author} {\bibfnamefont {Kianna}\
  \bibnamefont {Wan}}, \bibinfo {author} {\bibfnamefont {Jarrod}\ \bibnamefont
  {McClean}}, \bibinfo {author} {\bibfnamefont {Thomas~E.}\ \bibnamefont
  {O'Brien}}, \bibinfo {author} {\bibfnamefont {Nathan}\ \bibnamefont {Wiebe}},
  \ and\ \bibinfo {author} {\bibfnamefont {Ryan}\ \bibnamefont {Babbush}},\
  }\bibfield  {title} {\enquote {\bibinfo {title} {{Nearly Optimal Quantum
  Algorithm for Estimating Multiple Expectation Values}},}\ }\href {\doibase
  10.48550/arxiv.2111.09283} {\bibfield  {journal} {\bibinfo  {journal}
  {arXiv:2111.09283}\ } (\bibinfo {year} {2021}),\
  10.48550/arxiv.2111.09283}\BibitemShut {NoStop}%
\bibitem [{\citenamefont {Huang}\ \emph {et~al.}(2020)\citenamefont {Huang},
  \citenamefont {Kueng},\ and\ \citenamefont {Preskill}}]{Huang2020}%
  \BibitemOpen
  \bibfield  {author} {\bibinfo {author} {\bibfnamefont {Hsin-Yuan}\
  \bibnamefont {Huang}}, \bibinfo {author} {\bibfnamefont {Richard}\
  \bibnamefont {Kueng}}, \ and\ \bibinfo {author} {\bibfnamefont {John}\
  \bibnamefont {Preskill}},\ }\bibfield  {title} {\enquote {\bibinfo {title}
  {{Predicting many properties of a quantum system from very few
  measurements}},}\ }\href {\doibase 10.1038/s41567-020-0932-7} {\bibfield
  {journal} {\bibinfo  {journal} {Nature Physics}\ }\textbf {\bibinfo {volume}
  {16}},\ \bibinfo {pages} {1050--1057} (\bibinfo {year} {2020})}\BibitemShut
  {NoStop}%
\bibitem [{\citenamefont {Zhao}\ \emph {et~al.}(2021)\citenamefont {Zhao},
  \citenamefont {Rubin},\ and\ \citenamefont {Miyake}}]{Zhao2021}%
  \BibitemOpen
  \bibfield  {author} {\bibinfo {author} {\bibfnamefont {Andrew}\ \bibnamefont
  {Zhao}}, \bibinfo {author} {\bibfnamefont {Nicholas~C.}\ \bibnamefont
  {Rubin}}, \ and\ \bibinfo {author} {\bibfnamefont {Akimasa}\ \bibnamefont
  {Miyake}},\ }\bibfield  {title} {\enquote {\bibinfo {title} {{Fermionic
  Partial Tomography via Classical Shadows}},}\ }\href {\doibase
  10.1103/PhysRevLett.127.110504} {\bibfield  {journal} {\bibinfo  {journal}
  {Physical Review Letters}\ }\textbf {\bibinfo {volume} {127}},\ \bibinfo
  {pages} {110504} (\bibinfo {year} {2021})}\BibitemShut {NoStop}%
\bibitem [{\citenamefont {Wan}\ \emph {et~al.}(2022)\citenamefont {Wan},
  \citenamefont {Huggins}, \citenamefont {Lee},\ and\ \citenamefont
  {Babbush}}]{Wan2022}%
  \BibitemOpen
  \bibfield  {author} {\bibinfo {author} {\bibfnamefont {Kianna}\ \bibnamefont
  {Wan}}, \bibinfo {author} {\bibfnamefont {William~J.}\ \bibnamefont
  {Huggins}}, \bibinfo {author} {\bibfnamefont {Joonho}\ \bibnamefont {Lee}}, \
  and\ \bibinfo {author} {\bibfnamefont {Ryan}\ \bibnamefont {Babbush}},\
  }\bibfield  {title} {\enquote {\bibinfo {title} {{Matchgate Shadows for
  Fermionic Quantum Simulation}},}\ }\href {\doibase 10.48550/arxiv.2207.13723}
  {\bibfield  {journal} {\bibinfo  {journal} {arXiv:2207.13723}\ } (\bibinfo
  {year} {2022}),\ 10.48550/arxiv.2207.13723}\BibitemShut {NoStop}%
\bibitem [{\citenamefont {O'Gorman}(2022)}]{OGorman2022-ag}%
  \BibitemOpen
  \bibfield  {author} {\bibinfo {author} {\bibfnamefont {B}~\bibnamefont
  {O'Gorman}},\ }\bibfield  {title} {\enquote {\bibinfo {title} {Fermionic
  tomography and learning},}\ }\href {http://arxiv.org/abs/2207.14787}
  {\bibfield  {journal} {\bibinfo  {journal} {arXiv:2207.14787}\ } (\bibinfo
  {year} {2022})}\BibitemShut {NoStop}%
\bibitem [{\citenamefont {Low}(2022)}]{Low2022}%
  \BibitemOpen
  \bibfield  {author} {\bibinfo {author} {\bibfnamefont {Guang~Hao}\
  \bibnamefont {Low}},\ }\bibfield  {title} {\enquote {\bibinfo {title}
  {{Classical shadows of fermions with particle number symmetry}},}\ }\href
  {\doibase 10.48550/arxiv.2208.08964} {\bibfield  {journal} {\bibinfo
  {journal} {arXiv:2208.08964}\ } (\bibinfo {year} {2022}),\
  10.48550/arxiv.2208.08964}\BibitemShut {NoStop}%
\bibitem [{\citenamefont {Brassard}\ \emph {et~al.}(2002)\citenamefont
  {Brassard}, \citenamefont {H{\o}yer}, \citenamefont {Mosca},\ and\
  \citenamefont {Tapp}}]{Brassard2002}%
  \BibitemOpen
  \bibfield  {author} {\bibinfo {author} {\bibfnamefont {Gilles}\ \bibnamefont
  {Brassard}}, \bibinfo {author} {\bibfnamefont {Peter}\ \bibnamefont
  {H{\o}yer}}, \bibinfo {author} {\bibfnamefont {Michele}\ \bibnamefont
  {Mosca}}, \ and\ \bibinfo {author} {\bibfnamefont {Alain}\ \bibnamefont
  {Tapp}},\ }\bibfield  {title} {\enquote {\bibinfo {title} {{Quantum amplitude
  amplification and estimation}},}\ }in\ \href {\doibase
  10.1090/conm/305/05215} {\emph {\bibinfo {booktitle} {Quantum Computation and
  Information}}},\ \bibinfo {editor} {edited by\ \bibinfo {editor}
  {\bibnamefont {{Vitaly I Voloshin}}}, \bibinfo {editor} {\bibnamefont
  {{Samuel J. Lomonaco}}}, \ and\ \bibinfo {editor} {\bibnamefont {{Howard E.
  Brandt}}}}\ (\bibinfo  {publisher} {American Mathematical Society},\ \bibinfo
  {address} {Washington D.C.},\ \bibinfo {year} {2002})\ Chap.~\bibinfo
  {chapter} {3}, pp.\ \bibinfo {pages} {53--74}\BibitemShut {NoStop}%
\bibitem [{\citenamefont {Rall}(2020)}]{Rall2020}%
  \BibitemOpen
  \bibfield  {author} {\bibinfo {author} {\bibfnamefont {Patrick}\ \bibnamefont
  {Rall}},\ }\bibfield  {title} {\enquote {\bibinfo {title} {{Quantum
  algorithms for estimating physical quantities using block encodings}},}\
  }\href {\doibase 10.1103/PhysRevA.102.022408} {\bibfield  {journal} {\bibinfo
   {journal} {Physical Review A}\ }\textbf {\bibinfo {volume} {102}},\ \bibinfo
  {pages} {022408} (\bibinfo {year} {2020})}\BibitemShut {NoStop}%
\bibitem [{\citenamefont {Kivlichan}\ \emph {et~al.}(2018)\citenamefont
  {Kivlichan}, \citenamefont {McClean}, \citenamefont {Wiebe}, \citenamefont
  {Gidney}, \citenamefont {Aspuru-Guzik}, \citenamefont {Chan},\ and\
  \citenamefont {Babbush}}]{Kivlichan2018QuantumConnectivity}%
  \BibitemOpen
  \bibfield  {author} {\bibinfo {author} {\bibfnamefont {I.D.}\ \bibnamefont
  {Kivlichan}}, \bibinfo {author} {\bibfnamefont {J.}~\bibnamefont {McClean}},
  \bibinfo {author} {\bibfnamefont {N.}~\bibnamefont {Wiebe}}, \bibinfo
  {author} {\bibfnamefont {C.}~\bibnamefont {Gidney}}, \bibinfo {author}
  {\bibfnamefont {A.}~\bibnamefont {Aspuru-Guzik}}, \bibinfo {author}
  {\bibfnamefont {G.K.-L.}\ \bibnamefont {Chan}}, \ and\ \bibinfo {author}
  {\bibfnamefont {R.}~\bibnamefont {Babbush}},\ }\bibfield  {title} {\enquote
  {\bibinfo {title} {{Quantum Simulation of Electronic Structure with Linear
  Depth and Connectivity}},}\ }\href {\doibase 10.1103/PhysRevLett.120.110501}
  {\bibfield  {journal} {\bibinfo  {journal} {Physical Review Letters}\
  }\textbf {\bibinfo {volume} {120}} (\bibinfo {year} {2018}),\
  10.1103/PhysRevLett.120.110501}\BibitemShut {NoStop}%
\bibitem [{\citenamefont {Delgado}\ \emph {et~al.}(2022)\citenamefont
  {Delgado}, \citenamefont {Casares}, \citenamefont {dos Reis}, \citenamefont
  {Zini}, \citenamefont {Campos}, \citenamefont {Cruz-Hern{\'{a}}ndez},
  \citenamefont {Voigt}, \citenamefont {Lowe}, \citenamefont {Jahangiri},
  \citenamefont {Martin-Delgado}, \citenamefont {Mueller},\ and\ \citenamefont
  {Arrazola}}]{Delgado2022}%
  \BibitemOpen
  \bibfield  {author} {\bibinfo {author} {\bibfnamefont {Alain}\ \bibnamefont
  {Delgado}}, \bibinfo {author} {\bibfnamefont {Pablo A.~M.}\ \bibnamefont
  {Casares}}, \bibinfo {author} {\bibfnamefont {Roberto}\ \bibnamefont {dos
  Reis}}, \bibinfo {author} {\bibfnamefont {Modjtaba~Shokrian}\ \bibnamefont
  {Zini}}, \bibinfo {author} {\bibfnamefont {Roberto}\ \bibnamefont {Campos}},
  \bibinfo {author} {\bibfnamefont {Norge}\ \bibnamefont
  {Cruz-Hern{\'{a}}ndez}}, \bibinfo {author} {\bibfnamefont {Arne-Christian}\
  \bibnamefont {Voigt}}, \bibinfo {author} {\bibfnamefont {Angus}\ \bibnamefont
  {Lowe}}, \bibinfo {author} {\bibfnamefont {Soran}\ \bibnamefont {Jahangiri}},
  \bibinfo {author} {\bibfnamefont {M.~A.}\ \bibnamefont {Martin-Delgado}},
  \bibinfo {author} {\bibfnamefont {Jonathan~E.}\ \bibnamefont {Mueller}}, \
  and\ \bibinfo {author} {\bibfnamefont {Juan~Miguel}\ \bibnamefont
  {Arrazola}},\ }\bibfield  {title} {\enquote {\bibinfo {title} {{Simulating
  key properties of lithium-ion batteries with a fault-tolerant quantum
  computer}},}\ }\href {\doibase 10.1103/PhysRevA.106.032428} {\bibfield
  {journal} {\bibinfo  {journal} {Physical Review A}\ }\textbf {\bibinfo
  {volume} {106}},\ \bibinfo {pages} {032428} (\bibinfo {year}
  {2022})}\BibitemShut {NoStop}%
\bibitem [{\citenamefont {Berry}\ \emph {et~al.}(2018)\citenamefont {Berry},
  \citenamefont {Kieferov{\'{a}}}, \citenamefont {Scherer}, \citenamefont
  {Sanders}, \citenamefont {Low}, \citenamefont {Wiebe}, \citenamefont
  {Gidney},\ and\ \citenamefont {Babbush}}]{Berry2018}%
  \BibitemOpen
  \bibfield  {author} {\bibinfo {author} {\bibfnamefont {Dominic~W}\
  \bibnamefont {Berry}}, \bibinfo {author} {\bibfnamefont {Maria}\ \bibnamefont
  {Kieferov{\'{a}}}}, \bibinfo {author} {\bibfnamefont {Artur}\ \bibnamefont
  {Scherer}}, \bibinfo {author} {\bibfnamefont {Yuval~R}\ \bibnamefont
  {Sanders}}, \bibinfo {author} {\bibfnamefont {Guang~Hao}\ \bibnamefont
  {Low}}, \bibinfo {author} {\bibfnamefont {Nathan}\ \bibnamefont {Wiebe}},
  \bibinfo {author} {\bibfnamefont {Craig}\ \bibnamefont {Gidney}}, \ and\
  \bibinfo {author} {\bibfnamefont {Ryan}\ \bibnamefont {Babbush}},\ }\bibfield
   {title} {\enquote {\bibinfo {title} {{Improved Techniques for Preparing
  Eigenstates of Fermionic Hamiltonians}},}\ }\href {\doibase
  10.1038/s41534-018-0071-5} {\bibfield  {journal} {\bibinfo  {journal} {npj
  Quantum Information}\ }\textbf {\bibinfo {volume} {4}},\ \bibinfo {pages}
  {22} (\bibinfo {year} {2018})}\BibitemShut {NoStop}%
\bibitem [{\citenamefont {Shende}\ \emph {et~al.}(2006)\citenamefont {Shende},
  \citenamefont {Bullock},\ and\ \citenamefont {Markov}}]{Shende2006}%
  \BibitemOpen
  \bibfield  {author} {\bibinfo {author} {\bibfnamefont {V~V}\ \bibnamefont
  {Shende}}, \bibinfo {author} {\bibfnamefont {S~S}\ \bibnamefont {Bullock}}, \
  and\ \bibinfo {author} {\bibfnamefont {I~L}\ \bibnamefont {Markov}},\
  }\bibfield  {title} {\enquote {\bibinfo {title} {{Synthesis of quantum-logic
  circuits}},}\ }\href {\doibase 10.1109/TCAD.2005.855930} {\bibfield
  {journal} {\bibinfo  {journal} {IEEE Transactions on Computer-Aided Design of
  Integrated Circuits and Systems}\ }\textbf {\bibinfo {volume} {25}},\
  \bibinfo {pages} {1000--1010} (\bibinfo {year} {2006})}\BibitemShut {NoStop}%
\bibitem [{\citenamefont {Mermin}(1963)}]{Mermin1963}%
  \BibitemOpen
  \bibfield  {author} {\bibinfo {author} {\bibfnamefont {N.~David}\
  \bibnamefont {Mermin}},\ }\bibfield  {title} {\enquote {\bibinfo {title}
  {{Stability of the thermal Hartree-Fock approximation}},}\ }\href {\doibase
  10.1016/0003-4916(63)90226-4} {\bibfield  {journal} {\bibinfo  {journal}
  {Annals of Physics}\ }\textbf {\bibinfo {volume} {21}},\ \bibinfo {pages}
  {99--121} (\bibinfo {year} {1963})}\BibitemShut {NoStop}%
\bibitem [{\citenamefont {Babbush}\ \emph {et~al.}(2021)\citenamefont
  {Babbush}, \citenamefont {McClean}, \citenamefont {Newman}, \citenamefont
  {Gidney}, \citenamefont {Boixo},\ and\ \citenamefont {Neven}}]{BabbushFocus}%
  \BibitemOpen
  \bibfield  {author} {\bibinfo {author} {\bibfnamefont {Ryan}\ \bibnamefont
  {Babbush}}, \bibinfo {author} {\bibfnamefont {Jarrod~R.}\ \bibnamefont
  {McClean}}, \bibinfo {author} {\bibfnamefont {Michael}\ \bibnamefont
  {Newman}}, \bibinfo {author} {\bibfnamefont {Craig}\ \bibnamefont {Gidney}},
  \bibinfo {author} {\bibfnamefont {Sergio}\ \bibnamefont {Boixo}}, \ and\
  \bibinfo {author} {\bibfnamefont {Hartmut}\ \bibnamefont {Neven}},\
  }\bibfield  {title} {\enquote {\bibinfo {title} {{Focus beyond Quadratic
  Speedups for Error-Corrected Quantum Advantage}},}\ }\href {\doibase
  10.1103/PRXQuantum.2.010103} {\bibfield  {journal} {\bibinfo  {journal} {PRX
  Quantum}\ }\textbf {\bibinfo {volume} {2}},\ \bibinfo {pages} {010103}
  (\bibinfo {year} {2021})}\BibitemShut {NoStop}%
\bibitem [{\citenamefont {An}\ \emph {et~al.}(2021)\citenamefont {An},
  \citenamefont {Fang},\ and\ \citenamefont {Lin}}]{An2021}%
  \BibitemOpen
  \bibfield  {author} {\bibinfo {author} {\bibfnamefont {Dong}\ \bibnamefont
  {An}}, \bibinfo {author} {\bibfnamefont {Di}~\bibnamefont {Fang}}, \ and\
  \bibinfo {author} {\bibfnamefont {Lin}\ \bibnamefont {Lin}},\ }\bibfield
  {title} {\enquote {\bibinfo {title} {{Time-dependent unbounded Hamiltonian
  simulation with vector norm scaling}},}\ }\href {\doibase
  10.22331/q-2021-05-26-459} {\bibfield  {journal} {\bibinfo  {journal}
  {Quantum}\ }\textbf {\bibinfo {volume} {5}},\ \bibinfo {pages} {459}
  (\bibinfo {year} {2021})}\BibitemShut {NoStop}%
\bibitem [{\citenamefont {Provorse}\ and\ \citenamefont
  {Isborn}(2016)}]{Provorse2016May}%
  \BibitemOpen
  \bibfield  {author} {\bibinfo {author} {\bibfnamefont {Makenzie~R.}\
  \bibnamefont {Provorse}}\ and\ \bibinfo {author} {\bibfnamefont
  {Christine~M.}\ \bibnamefont {Isborn}},\ }\bibfield  {title} {\enquote
  {\bibinfo {title} {{Electron dynamics with real-time time-dependent density
  functional theory}},}\ }\href {\doibase 10.1002/qua.25096} {\bibfield
  {journal} {\bibinfo  {journal} {Int. J. Quantum Chem.}\ }\textbf {\bibinfo
  {volume} {116}},\ \bibinfo {pages} {739--749} (\bibinfo {year}
  {2016})}\BibitemShut {NoStop}%
\bibitem [{\citenamefont {Cohen}\ \emph {et~al.}(2008)\citenamefont {Cohen},
  \citenamefont {Mori-Sa'nchez},\ and\ \citenamefont {Yang}}]{Cohen2008Aug}%
  \BibitemOpen
  \bibfield  {author} {\bibinfo {author} {\bibfnamefont {Aron~J.}\ \bibnamefont
  {Cohen}}, \bibinfo {author} {\bibfnamefont {Paula}\ \bibnamefont
  {Mori-Sa'nchez}}, \ and\ \bibinfo {author} {\bibfnamefont {Weitao}\
  \bibnamefont {Yang}},\ }\bibfield  {title} {\enquote {\bibinfo {title}
  {{Insights into Current Limitations of Density Functional Theory}},}\ }\href
  {\doibase 10.1126/science.1158722} {\bibfield  {journal} {\bibinfo  {journal}
  {Science}\ }\textbf {\bibinfo {volume} {321}},\ \bibinfo {pages} {792--794}
  (\bibinfo {year} {2008})}\BibitemShut {NoStop}%
\bibitem [{\citenamefont {Tully}(2000)}]{Tully2000Oct}%
  \BibitemOpen
  \bibfield  {author} {\bibinfo {author} {\bibfnamefont {John~C.}\ \bibnamefont
  {Tully}},\ }\bibfield  {title} {\enquote {\bibinfo {title} {{Chemical
  Dynamics at Metal Surfaces}},}\ }\href {\doibase
  10.1146/annurev.physchem.51.1.153} {\bibfield  {journal} {\bibinfo  {journal}
  {Annu. Rev. Phys. Chem.}\ }\textbf {\bibinfo {volume} {51}},\ \bibinfo
  {pages} {153--178} (\bibinfo {year} {2000})}\BibitemShut {NoStop}%
\bibitem [{\citenamefont {Wang}\ \emph {et~al.}(2013)\citenamefont {Wang},
  \citenamefont {Hou},\ and\ \citenamefont {Zheng}}]{Wang2013Nov}%
  \BibitemOpen
  \bibfield  {author} {\bibinfo {author} {\bibfnamefont {Rulin}\ \bibnamefont
  {Wang}}, \bibinfo {author} {\bibfnamefont {Dong}\ \bibnamefont {Hou}}, \ and\
  \bibinfo {author} {\bibfnamefont {Xiao}\ \bibnamefont {Zheng}},\ }\bibfield
  {title} {\enquote {\bibinfo {title} {{Time-dependent density-functional
  theory for real-time electronic dynamics on material surfaces}},}\ }\href
  {\doibase 10.1103/PhysRevB.88.205126} {\bibfield  {journal} {\bibinfo
  {journal} {Phys. Rev. B}\ }\textbf {\bibinfo {volume} {88}},\ \bibinfo
  {pages} {205126} (\bibinfo {year} {2013})}\BibitemShut {NoStop}%
\bibitem [{\citenamefont {Baczewski}\ \emph {et~al.}(2016)\citenamefont
  {Baczewski}, \citenamefont {Shulenburger}, \citenamefont {Desjarlais},
  \citenamefont {Hansen},\ and\ \citenamefont {Magyar}}]{baczewski2016x}%
  \BibitemOpen
  \bibfield  {author} {\bibinfo {author} {\bibfnamefont {Andrew~David}\
  \bibnamefont {Baczewski}}, \bibinfo {author} {\bibfnamefont {L}~\bibnamefont
  {Shulenburger}}, \bibinfo {author} {\bibfnamefont {MP}~\bibnamefont
  {Desjarlais}}, \bibinfo {author} {\bibfnamefont {SB}~\bibnamefont {Hansen}},
  \ and\ \bibinfo {author} {\bibfnamefont {RJ}~\bibnamefont {Magyar}},\
  }\bibfield  {title} {\enquote {\bibinfo {title} {X-ray thomson scattering in
  warm dense matter without the chihara decomposition},}\ }\href {\doibase
  10.1103/PhysRevLett.116.115004} {\bibfield  {journal} {\bibinfo  {journal}
  {Physical review letters}\ }\textbf {\bibinfo {volume} {116}},\ \bibinfo
  {pages} {115004} (\bibinfo {year} {2016})}\BibitemShut {NoStop}%
\bibitem [{\citenamefont {Magyar}\ \emph {et~al.}(2016)\citenamefont {Magyar},
  \citenamefont {Shulenburger},\ and\ \citenamefont
  {Baczewski}}]{magyar2016stopping}%
  \BibitemOpen
  \bibfield  {author} {\bibinfo {author} {\bibfnamefont {Rudolph~J}\
  \bibnamefont {Magyar}}, \bibinfo {author} {\bibfnamefont {L}~\bibnamefont
  {Shulenburger}}, \ and\ \bibinfo {author} {\bibfnamefont {AD}~\bibnamefont
  {Baczewski}},\ }\href {\doibase 10.1002/ctpp.201500143} {\enquote {\bibinfo
  {title} {Stopping of deuterium in warm dense deuterium from ehrenfest
  time-dependent density functional theory},}\ } (\bibinfo {year}
  {2016})\BibitemShut {NoStop}%
\bibitem [{\citenamefont {Andrade}\ \emph {et~al.}(2018)\citenamefont
  {Andrade}, \citenamefont {Hamel},\ and\ \citenamefont
  {Correa}}]{andrade2018negative}%
  \BibitemOpen
  \bibfield  {author} {\bibinfo {author} {\bibfnamefont {Xavier}\ \bibnamefont
  {Andrade}}, \bibinfo {author} {\bibfnamefont {S{\'e}bastien}\ \bibnamefont
  {Hamel}}, \ and\ \bibinfo {author} {\bibfnamefont {Alfredo~A}\ \bibnamefont
  {Correa}},\ }\bibfield  {title} {\enquote {\bibinfo {title} {Negative
  differential conductivity in liquid aluminum from real-time quantum
  simulations},}\ }\href {\doibase 10.1140/epjb/e2018-90291-5} {\bibfield
  {journal} {\bibinfo  {journal} {The European Physical Journal B}\ }\textbf
  {\bibinfo {volume} {91}},\ \bibinfo {pages} {1--7} (\bibinfo {year}
  {2018})}\BibitemShut {NoStop}%
\bibitem [{\citenamefont {Ding}\ \emph {et~al.}(2018)\citenamefont {Ding},
  \citenamefont {White}, \citenamefont {Hu}, \citenamefont {Certik},\ and\
  \citenamefont {Collins}}]{ding2018ab}%
  \BibitemOpen
  \bibfield  {author} {\bibinfo {author} {\bibfnamefont {YH}~\bibnamefont
  {Ding}}, \bibinfo {author} {\bibfnamefont {Alexander~James}\ \bibnamefont
  {White}}, \bibinfo {author} {\bibfnamefont {SX}~\bibnamefont {Hu}}, \bibinfo
  {author} {\bibfnamefont {Ondrej}\ \bibnamefont {Certik}}, \ and\ \bibinfo
  {author} {\bibfnamefont {Lee~A}\ \bibnamefont {Collins}},\ }\bibfield
  {title} {\enquote {\bibinfo {title} {Ab initio studies on the stopping power
  of warm dense matter with time-dependent orbital-free density functional
  theory},}\ }\href {\doibase 10.1103/PhysRevLett.121.145001} {\bibfield
  {journal} {\bibinfo  {journal} {Physical Review Letters}\ }\textbf {\bibinfo
  {volume} {121}},\ \bibinfo {pages} {145001} (\bibinfo {year}
  {2018})}\BibitemShut {NoStop}%
\bibitem [{\citenamefont {Atzeni}\ and\ \citenamefont {Meyer-ter
  Vehn}(2004)}]{atzeni2004physics}%
  \BibitemOpen
  \bibfield  {author} {\bibinfo {author} {\bibfnamefont {Stefano}\ \bibnamefont
  {Atzeni}}\ and\ \bibinfo {author} {\bibfnamefont {J{\"u}rgen}\ \bibnamefont
  {Meyer-ter Vehn}},\ }\href@noop {} {\emph {\bibinfo {title} {The physics of
  inertial fusion: beam plasma interaction, hydrodynamics, hot dense
  matter}}},\ Vol.\ \bibinfo {volume} {125}\ (\bibinfo  {publisher} {OUP
  Oxford},\ \bibinfo {year} {2004})\BibitemShut {NoStop}%
\bibitem [{\citenamefont {Graziani}\ \emph {et~al.}(2014)\citenamefont
  {Graziani}, \citenamefont {Desjarlais}, \citenamefont {Redmer},\ and\
  \citenamefont {Trickey}}]{graziani2014frontiers}%
  \BibitemOpen
  \bibfield  {author} {\bibinfo {author} {\bibfnamefont {Frank}\ \bibnamefont
  {Graziani}}, \bibinfo {author} {\bibfnamefont {Michael~P}\ \bibnamefont
  {Desjarlais}}, \bibinfo {author} {\bibfnamefont {Ronald}\ \bibnamefont
  {Redmer}}, \ and\ \bibinfo {author} {\bibfnamefont {Samuel~B}\ \bibnamefont
  {Trickey}},\ }\href@noop {} {\emph {\bibinfo {title} {Frontiers and
  challenges in warm dense matter}}},\ Vol.~\bibinfo {volume} {96}\ (\bibinfo
  {publisher} {Springer Science \& Business},\ \bibinfo {year}
  {2014})\BibitemShut {NoStop}%
\bibitem [{\citenamefont {Dornheim}\ \emph {et~al.}(2018)\citenamefont
  {Dornheim}, \citenamefont {Groth},\ and\ \citenamefont
  {Bonitz}}]{dornheim2018uniform}%
  \BibitemOpen
  \bibfield  {author} {\bibinfo {author} {\bibfnamefont {Tobias}\ \bibnamefont
  {Dornheim}}, \bibinfo {author} {\bibfnamefont {Simon}\ \bibnamefont {Groth}},
  \ and\ \bibinfo {author} {\bibfnamefont {Michael}\ \bibnamefont {Bonitz}},\
  }\bibfield  {title} {\enquote {\bibinfo {title} {The uniform electron gas at
  warm dense matter conditions},}\ }\href {\doibase
  10.1016/j.physrep.2018.04.001} {\bibfield  {journal} {\bibinfo  {journal}
  {Physics Reports}\ }\textbf {\bibinfo {volume} {744}},\ \bibinfo {pages}
  {1--86} (\bibinfo {year} {2018})}\BibitemShut {NoStop}%
\bibitem [{\citenamefont {Bailey}\ \emph {et~al.}(2015)\citenamefont {Bailey},
  \citenamefont {Nagayama}, \citenamefont {Loisel}, \citenamefont {Rochau},
  \citenamefont {Blancard}, \citenamefont {Colgan}, \citenamefont {Cosse},
  \citenamefont {Faussurier}, \citenamefont {Fontes}, \citenamefont {Gilleron}
  \emph {et~al.}}]{bailey2015higher}%
  \BibitemOpen
  \bibfield  {author} {\bibinfo {author} {\bibfnamefont {James~E}\ \bibnamefont
  {Bailey}}, \bibinfo {author} {\bibfnamefont {Taisuke}\ \bibnamefont
  {Nagayama}}, \bibinfo {author} {\bibfnamefont {Guillaume~Pascal}\
  \bibnamefont {Loisel}}, \bibinfo {author} {\bibfnamefont {Gregory~Alan}\
  \bibnamefont {Rochau}}, \bibinfo {author} {\bibfnamefont {C}~\bibnamefont
  {Blancard}}, \bibinfo {author} {\bibfnamefont {James}\ \bibnamefont
  {Colgan}}, \bibinfo {author} {\bibfnamefont {Ph}~\bibnamefont {Cosse}},
  \bibinfo {author} {\bibfnamefont {G}~\bibnamefont {Faussurier}}, \bibinfo
  {author} {\bibfnamefont {CJ}~\bibnamefont {Fontes}}, \bibinfo {author}
  {\bibfnamefont {F}~\bibnamefont {Gilleron}},  \emph {et~al.},\ }\bibfield
  {title} {\enquote {\bibinfo {title} {A higher-than-predicted measurement of
  iron opacity at solar interior temperatures},}\ }\href {\doibase
  10.1038/nature14048} {\bibfield  {journal} {\bibinfo  {journal} {Nature}\
  }\textbf {\bibinfo {volume} {517}},\ \bibinfo {pages} {56--59} (\bibinfo
  {year} {2015})}\BibitemShut {NoStop}%
\bibitem [{\citenamefont {Nagayama}\ \emph {et~al.}(2019)\citenamefont
  {Nagayama}, \citenamefont {Bailey}, \citenamefont {Loisel}, \citenamefont
  {Dunham}, \citenamefont {Rochau}, \citenamefont {Blancard}, \citenamefont
  {Colgan}, \citenamefont {Coss{\'e}}, \citenamefont {Faussurier},
  \citenamefont {Fontes} \emph {et~al.}}]{nagayama2019systematic}%
  \BibitemOpen
  \bibfield  {author} {\bibinfo {author} {\bibfnamefont {Taisuke}\ \bibnamefont
  {Nagayama}}, \bibinfo {author} {\bibfnamefont {JE}~\bibnamefont {Bailey}},
  \bibinfo {author} {\bibfnamefont {GP}~\bibnamefont {Loisel}}, \bibinfo
  {author} {\bibfnamefont {GS}~\bibnamefont {Dunham}}, \bibinfo {author}
  {\bibfnamefont {GA}~\bibnamefont {Rochau}}, \bibinfo {author} {\bibfnamefont
  {C}~\bibnamefont {Blancard}}, \bibinfo {author} {\bibfnamefont
  {J}~\bibnamefont {Colgan}}, \bibinfo {author} {\bibfnamefont
  {Ph}~\bibnamefont {Coss{\'e}}}, \bibinfo {author} {\bibfnamefont
  {G}~\bibnamefont {Faussurier}}, \bibinfo {author} {\bibfnamefont
  {Christopher~John}\ \bibnamefont {Fontes}},  \emph {et~al.},\ }\bibfield
  {title} {\enquote {\bibinfo {title} {Systematic study of l-shell opacity at
  stellar interior temperatures},}\ }\href {\doibase
  10.1103/PhysRevLett.122.235001} {\bibfield  {journal} {\bibinfo  {journal}
  {Physical review letters}\ }\textbf {\bibinfo {volume} {122}},\ \bibinfo
  {pages} {235001} (\bibinfo {year} {2019})}\BibitemShut {NoStop}%
\bibitem [{\citenamefont {Huber}\ and\ \citenamefont
  {Klamroth}(2011)}]{Huber2011Feb}%
  \BibitemOpen
  \bibfield  {author} {\bibinfo {author} {\bibfnamefont {Christian}\
  \bibnamefont {Huber}}\ and\ \bibinfo {author} {\bibfnamefont {Tillmann}\
  \bibnamefont {Klamroth}},\ }\bibfield  {title} {\enquote {\bibinfo {title}
  {{Explicitly time-dependent coupled cluster singles doubles calculations of
  laser-driven many-electron dynamics}},}\ }\href {\doibase 10.1063/1.3530807}
  {\bibfield  {journal} {\bibinfo  {journal} {J. Chem. Phys.}\ }\textbf
  {\bibinfo {volume} {134}},\ \bibinfo {pages} {054113} (\bibinfo {year}
  {2011})}\BibitemShut {NoStop}%
\bibitem [{\citenamefont {Sato}\ \emph {et~al.}(2018)\citenamefont {Sato},
  \citenamefont {Pathak}, \citenamefont {Orimo},\ and\ \citenamefont
  {Ishikawa}}]{Sato2018Feb}%
  \BibitemOpen
  \bibfield  {author} {\bibinfo {author} {\bibfnamefont {Takeshi}\ \bibnamefont
  {Sato}}, \bibinfo {author} {\bibfnamefont {Himadri}\ \bibnamefont {Pathak}},
  \bibinfo {author} {\bibfnamefont {Yuki}\ \bibnamefont {Orimo}}, \ and\
  \bibinfo {author} {\bibfnamefont {Kenichi~L.}\ \bibnamefont {Ishikawa}},\
  }\bibfield  {title} {\enquote {\bibinfo {title} {{Communication:
  Time-dependent optimized coupled-cluster method for multielectron
  dynamics}},}\ }\href {\doibase 10.1063/1.5020633} {\bibfield  {journal}
  {\bibinfo  {journal} {J. Chem. Phys.}\ }\textbf {\bibinfo {volume} {148}},\
  \bibinfo {pages} {051101} (\bibinfo {year} {2018})}\BibitemShut {NoStop}%
\bibitem [{\citenamefont {Shushkov}\ and\ \citenamefont
  {Miller}(2019)}]{Shushkov2019Oct}%
  \BibitemOpen
  \bibfield  {author} {\bibinfo {author} {\bibfnamefont {Philip}\ \bibnamefont
  {Shushkov}}\ and\ \bibinfo {author} {\bibfnamefont {Thomas~F.}\ \bibnamefont
  {Miller}},\ }\bibfield  {title} {\enquote {\bibinfo {title} {{Real-time
  density-matrix coupled-cluster approach for closed and open systems at finite
  temperature}},}\ }\href {\doibase 10.1063/1.5121749} {\bibfield  {journal}
  {\bibinfo  {journal} {J. Chem. Phys.}\ }\textbf {\bibinfo {volume} {151}},\
  \bibinfo {pages} {134107} (\bibinfo {year} {2019})}\BibitemShut {NoStop}%
\bibitem [{\citenamefont {White}\ and\ \citenamefont
  {Chan}(2019)}]{White2019Nov}%
  \BibitemOpen
  \bibfield  {author} {\bibinfo {author} {\bibfnamefont {Alec~F.}\ \bibnamefont
  {White}}\ and\ \bibinfo {author} {\bibfnamefont {Garnet Kin-Lic}\
  \bibnamefont {Chan}},\ }\bibfield  {title} {\enquote {\bibinfo {title}
  {{Time-Dependent Coupled Cluster Theory on the Keldysh Contour for
  Nonequilibrium Systems}},}\ }\href {\doibase 10.1021/acs.jctc.9b00750}
  {\bibfield  {journal} {\bibinfo  {journal} {J. Chem. Theory Comput.}\
  }\textbf {\bibinfo {volume} {15}},\ \bibinfo {pages} {6137--6153} (\bibinfo
  {year} {2019})}\BibitemShut {NoStop}%
\bibitem [{\citenamefont {Kivlichan}\ \emph {et~al.}(2017)\citenamefont
  {Kivlichan}, \citenamefont {Wiebe}, \citenamefont {Babbush},\ and\
  \citenamefont {Aspuru-Guzik}}]{Kivlichan2016}%
  \BibitemOpen
  \bibfield  {author} {\bibinfo {author} {\bibfnamefont {Ian~D}\ \bibnamefont
  {Kivlichan}}, \bibinfo {author} {\bibfnamefont {Nathan}\ \bibnamefont
  {Wiebe}}, \bibinfo {author} {\bibfnamefont {Ryan}\ \bibnamefont {Babbush}}, \
  and\ \bibinfo {author} {\bibfnamefont {Alan}\ \bibnamefont {Aspuru-Guzik}},\
  }\bibfield  {title} {\enquote {\bibinfo {title} {{Bounding the costs of
  quantum simulation of many-body physics in real space}},}\ }\href
  {http://iopscience.iop.org/article/10.1088/1751-8121/aa77b8} {\bibfield
  {journal} {\bibinfo  {journal} {Journal of Physics A: Mathematical and
  Theoretical}\ }\textbf {\bibinfo {volume} {50}},\ \bibinfo {pages} {305301}
  (\bibinfo {year} {2017})}\BibitemShut {NoStop}%
\bibitem [{\citenamefont {Chen}\ and\ \citenamefont
  {Weeks}(2006)}]{smoothCoulomb2006}%
  \BibitemOpen
  \bibfield  {author} {\bibinfo {author} {\bibfnamefont {Yng-Gwei}\
  \bibnamefont {Chen}}\ and\ \bibinfo {author} {\bibfnamefont {John~D.}\
  \bibnamefont {Weeks}},\ }\bibfield  {title} {\enquote {\bibinfo {title}
  {{Local molecular field theory for effective attractions between like charged
  objects in systems with strong Coulomb interactions}},}\ }\href {\doibase
  10.1073/pnas.0600282103} {\bibfield  {journal} {\bibinfo  {journal}
  {Proceedings of the National Academy of Sciences}\ }\textbf {\bibinfo
  {volume} {103}},\ \bibinfo {pages} {7560--7565} (\bibinfo {year}
  {2006})}\BibitemShut {NoStop}%
\bibitem [{\citenamefont {Gonz{\'{a}}lez-Espinoza}\ \emph
  {et~al.}(2016)\citenamefont {Gonz{\'{a}}lez-Espinoza}, \citenamefont {Ayers},
  \citenamefont {Karwowski},\ and\ \citenamefont {Savin}}]{smoothCoulomb2016}%
  \BibitemOpen
  \bibfield  {author} {\bibinfo {author} {\bibfnamefont {Cristina~E.}\
  \bibnamefont {Gonz{\'{a}}lez-Espinoza}}, \bibinfo {author} {\bibfnamefont
  {Paul~W.}\ \bibnamefont {Ayers}}, \bibinfo {author} {\bibfnamefont {Jacek}\
  \bibnamefont {Karwowski}}, \ and\ \bibinfo {author} {\bibfnamefont {Andreas}\
  \bibnamefont {Savin}},\ }\bibfield  {title} {\enquote {\bibinfo {title}
  {{Smooth models for the Coulomb potential}},}\ }\href {\doibase
  10.1007/s00214-016-2007-5} {\bibfield  {journal} {\bibinfo  {journal}
  {Theoretical Chemistry Accounts}\ }\textbf {\bibinfo {volume} {135}},\
  \bibinfo {pages} {256} (\bibinfo {year} {2016})}\BibitemShut {NoStop}%
\bibitem [{\citenamefont {Carrier}\ \emph {et~al.}(1988)\citenamefont
  {Carrier}, \citenamefont {Greengard},\ and\ \citenamefont
  {Rokhlin}}]{adaptiveFMM}%
  \BibitemOpen
  \bibfield  {author} {\bibinfo {author} {\bibfnamefont {J.}~\bibnamefont
  {Carrier}}, \bibinfo {author} {\bibfnamefont {L.}~\bibnamefont {Greengard}},
  \ and\ \bibinfo {author} {\bibfnamefont {V.}~\bibnamefont {Rokhlin}},\
  }\bibfield  {title} {\enquote {\bibinfo {title} {{A Fast Adaptive Multipole
  Algorithm for Particle Simulations}},}\ }\href {\doibase 10.1137/0909044}
  {\bibfield  {journal} {\bibinfo  {journal} {SIAM Journal on Scientific and
  Statistical Computing}\ }\textbf {\bibinfo {volume} {9}},\ \bibinfo {pages}
  {669--686} (\bibinfo {year} {1988})}\BibitemShut {NoStop}%
\bibitem [{\citenamefont {Bravyi}\ and\ \citenamefont
  {Maslov}(2021)}]{Bravyi2021-rn}%
  \BibitemOpen
  \bibfield  {author} {\bibinfo {author} {\bibfnamefont {Sergey}\ \bibnamefont
  {Bravyi}}\ and\ \bibinfo {author} {\bibfnamefont {Dmitri}\ \bibnamefont
  {Maslov}},\ }\bibfield  {title} {\enquote {\bibinfo {title} {Hadamard-free
  circuits expose the structure of the clifford group},}\ }\href {\doibase
  10.1109/tit.2021.3081415} {\bibfield  {journal} {\bibinfo  {journal} {IEEE
  Trans. Inf. Theory}\ }\textbf {\bibinfo {volume} {67}},\ \bibinfo {pages}
  {4546--4563} (\bibinfo {year} {2021})}\BibitemShut {NoStop}%
\bibitem [{\citenamefont {Huggins}\ \emph {et~al.}(2022)\citenamefont
  {Huggins}, \citenamefont {O'Gorman}, \citenamefont {Rubin}, \citenamefont
  {Reichman}, \citenamefont {Babbush},\ and\ \citenamefont
  {Lee}}]{Huggins2022-nc}%
  \BibitemOpen
  \bibfield  {author} {\bibinfo {author} {\bibfnamefont {William~J}\
  \bibnamefont {Huggins}}, \bibinfo {author} {\bibfnamefont {Bryan~A}\
  \bibnamefont {O'Gorman}}, \bibinfo {author} {\bibfnamefont {Nicholas~C}\
  \bibnamefont {Rubin}}, \bibinfo {author} {\bibfnamefont {David~R}\
  \bibnamefont {Reichman}}, \bibinfo {author} {\bibfnamefont {Ryan}\
  \bibnamefont {Babbush}}, \ and\ \bibinfo {author} {\bibfnamefont {Joonho}\
  \bibnamefont {Lee}},\ }\bibfield  {title} {\enquote {\bibinfo {title}
  {Unbiasing fermionic quantum monte carlo with a quantum computer},}\ }\href
  {\doibase 10.1038/s41586-021-04351-z} {\bibfield  {journal} {\bibinfo
  {journal} {Nature}\ }\textbf {\bibinfo {volume} {603}},\ \bibinfo {pages}
  {416--420} (\bibinfo {year} {2022})},\ \Eprint
  {http://arxiv.org/abs/2106.16235} {arXiv:2106.16235 [quant-ph]} \BibitemShut
  {NoStop}%
\bibitem [{\citenamefont {Lerasle}(2019)}]{Lerasle2019-zm}%
  \BibitemOpen
  \bibfield  {author} {\bibinfo {author} {\bibfnamefont {Matthieu}\
  \bibnamefont {Lerasle}},\ }\bibfield  {title} {\enquote {\bibinfo {title}
  {Lecture notes: Selected topics on robust statistical learning theory},}\
  }\href {http://arxiv.org/abs/1908.10761} {\bibfield  {journal} {\bibinfo
  {journal} {arXiv:1908.10761}\ } (\bibinfo {year} {2019})}\BibitemShut
  {NoStop}%
\bibitem [{\citenamefont {Garc{\'\i}a}\ \emph {et~al.}(2017)\citenamefont
  {Garc{\'\i}a}, \citenamefont {Markov},\ and\ \citenamefont
  {Cross}}]{Garcia2017-rx}%
  \BibitemOpen
  \bibfield  {author} {\bibinfo {author} {\bibfnamefont {H{\'e}ctor~J}\
  \bibnamefont {Garc{\'\i}a}}, \bibinfo {author} {\bibfnamefont {Igor~L}\
  \bibnamefont {Markov}}, \ and\ \bibinfo {author} {\bibfnamefont {Andrew~W}\
  \bibnamefont {Cross}},\ }\bibfield  {title} {\enquote {\bibinfo {title} {On
  the geometry of stabilizer states},}\ }\href
  {http://arxiv.org/abs/1711.07848} {\bibfield  {journal} {\bibinfo  {journal}
  {arXiv:1711.07848}\ } (\bibinfo {year} {2017})}\BibitemShut {NoStop}%
\bibitem [{\citenamefont {Aaronson}\ and\ \citenamefont
  {Gottesman}(2004)}]{Aaronson2004-yb}%
  \BibitemOpen
  \bibfield  {author} {\bibinfo {author} {\bibfnamefont {Scott}\ \bibnamefont
  {Aaronson}}\ and\ \bibinfo {author} {\bibfnamefont {Daniel}\ \bibnamefont
  {Gottesman}},\ }\bibfield  {title} {\enquote {\bibinfo {title} {Improved
  simulation of stabilizer circuits},}\ }\href {\doibase
  10.1103/physreva.70.052328} {\bibfield  {journal} {\bibinfo  {journal} {Phys.
  Rev. A}\ }\textbf {\bibinfo {volume} {70}} (\bibinfo {year} {2004}),\
  10.1103/physreva.70.052328}\BibitemShut {NoStop}%
\bibitem [{\citenamefont {Gottesman}(1998)}]{Gottesman1998-rv}%
  \BibitemOpen
  \bibfield  {author} {\bibinfo {author} {\bibfnamefont {D}~\bibnamefont
  {Gottesman}},\ }\href
  {https://www.osti.gov/biblio/319738-heisenberg-representation-quantum-computers}
  {\emph {\bibinfo {title} {The Heisenberg representation of quantum
  computers}}},\ \bibinfo {type} {Tech. Rep.}\ \bibinfo {number}
  {LA-UR-98-2848; CONF-980788-}\ (\bibinfo  {institution} {Los Alamos National
  Lab., NM (United States)},\ \bibinfo {year} {1998})\BibitemShut {NoStop}%
\bibitem [{\citenamefont {Gross}\ \emph {et~al.}(2015)\citenamefont {Gross},
  \citenamefont {Krahmer},\ and\ \citenamefont {Kueng}}]{Gross2015-ck}%
  \BibitemOpen
  \bibfield  {author} {\bibinfo {author} {\bibfnamefont {D}~\bibnamefont
  {Gross}}, \bibinfo {author} {\bibfnamefont {F}~\bibnamefont {Krahmer}}, \
  and\ \bibinfo {author} {\bibfnamefont {R}~\bibnamefont {Kueng}},\ }\bibfield
  {title} {\enquote {\bibinfo {title} {A partial derandomization of {PhaseLift}
  using spherical designs},}\ }\href {\doibase 10.1007/s00041-014-9361-2}
  {\bibfield  {journal} {\bibinfo  {journal} {J. Fourier Anal. Appl.}\ }\textbf
  {\bibinfo {volume} {21}},\ \bibinfo {pages} {229--266} (\bibinfo {year}
  {2015})}\BibitemShut {NoStop}%
\end{thebibliography}%

\onecolumngrid

\appendix

\section{Norms and scaling for the nonlinear differential equation governing mean-field evolution}
\label{app:scanorm}

We have the differential equation
\begin{equation}
i \frac{\partial \mathbf C_\text{occ}\left(t\right)}{\partial t}
=
\mathbf F\left(t\right)\mathbf C_\text{occ}\left(t\right)
\end{equation}
where 
\begin{equation}
F_{\mu\nu}\!\left(t\right) = h_{\mu \nu} + 
\sum_{\lambda\sigma}^{N} \left(\left(\mu\nu|\lambda\sigma\right)-\frac{\left(\mu\sigma|\lambda\nu\right)}{2}\right) P_{\sigma\lambda}\!\left(t\right)
\end{equation}
with $\mathbf P(t) = \mathbf C_\text{occ} (t)\mathbf C_\text{occ}(t)^\dagger$.
If $\mathbf F$ were independent of $\mathbf C_\text{occ}$, then it would imply that taking the $n$th derivative gives
\begin{equation}
(i)^n \frac{\partial^n \mathbf C_\text{occ}\left(t\right)}{\partial t^n}
=
\mathbf F^n \mathbf C_\text{occ}\left(t\right).
\end{equation}
That means the norm of the $n$th derivative would scale as $\|\mathbf F\|^n$ (with $\mathbf C_\text{occ}$ normalized).

Then higher-order methods will typically have an error that scales as the norm of the higher-order derivatives.
For example, if one were to use a Taylor series up to order $k$ to approximate a time step, then the error for a time step of length $\delta t$ would scale as
\begin{equation}
\frac 1{(k+1)!} \|\mathbf F\|^{k+1}  \delta t^{k+1}.
\end{equation}
This means that if the size of the time step is taken as proportional to $1/\|\mathbf F\|$, then the error may be made exponentially small in $k$.
As a result, the total number of time steps used scales as ${\cal O}(\|\mathbf F\| t)$.
Similar considerations hold for other higher-order methods for integration.
The dependence of the complexity on $\|\mathbf F\|$ can also be expected from principles of scaling, where if $\mathbf F$ is divided by $\|\mathbf F\|$ but $t$ is also multiplied by $\|\mathbf F\|$, then the same differential equation is obtained.

In our case where $\mathbf F$ is dependent on $\mathbf C_\text{occ}$, the situation is more complicated.
This is because taking higher-order derivatives of $\mathbf C_\text{occ}$ yields more terms due to the derivatives of $\mathbf C_\text{occ}$ in $\mathbf F$.
To describe this, let us write, omitting $h_{\mu\nu}$ for simplicity,
\begin{equation}
    F_{\mu\nu}\!\left(t\right) = V_{\mu\nu\sigma\lambda} C_{\sigma a}C^*_{\lambda a},
\end{equation}
with $C_{\sigma a}$ the matrix entries of $\mathbf C_\text{occ}$.
We are taking a convention that Greek indices are over all orbitals, English letters are over electrons, and repeated indices are summed over.
Then we would give the derivative as
\begin{equation}
    i\frac{\partial C_{\mu b}}{\partial t} =  V_{\mu\nu\sigma\lambda} C_{\sigma a}C^*_{\lambda a} C_{\nu b}.
\end{equation}
We can define an $\eta$-norm of $\V$ as
\begin{equation}
    \| \mathbf V \|_\eta = \max_{\mathbf{x},\mathbf{y},\mathbf{z}} V_{\mu\nu\sigma\lambda} x_\mu y^*_\nu z_{\sigma\lambda},
\label{eq:etanorm}
\end{equation}
with $\|\mathbf{x}\|=\|\mathbf{y}\|=\|\mathbf{z}\|=1$ (i.e., spectral norms are normalized), and $\mathbf{z}$ of rank $\eta$.
To bound this norm, we can consider the first term for $V_{\mu\nu\sigma\lambda}$, which is
\begin{equation}
\left(\mu\nu|\lambda\sigma\right) =  \int \textrm{d}\mathbf r_1 \, \textrm{d}\mathbf r_2 \frac{\phi_{\mu}^* \left(\mathbf r_1\right) \phi_{\nu}\left(\mathbf r_1\right)  \phi_\lambda^* \left(\mathbf r_2\right) \phi_\sigma \left(\mathbf r_2\right) }{\left| \mathbf r_1 - \mathbf r_2\right|} \, .
\end{equation}
The multiplication by $x_\mu$ and sum over $\mu$ corresponds to a transformation of $\phi_{\mu}$ to a new orbital, and similarly, the sum over $\nu$ transforms $\phi_{\nu}$ to another new orbital.
Since $\mathbf{z}$ is of rank $\eta$, the sum over $\lambda$ and $\sigma$ corresponds to transforming the orbital basis for both $\phi_\lambda$ and $\phi_\sigma$, and summing over $\eta$ of these basis states.

That is, we can write
\begin{equation}
\sum_{a=1}^\eta \int \textrm{d}\mathbf r_1 \, \textrm{d}\mathbf r_2 \frac{\phi^* \left(\mathbf r_1\right) \chi\left(\mathbf r_1\right)  \psi_a^* \left(\mathbf r_2\right) \theta_a \left(\mathbf r_2\right) }{\left| \mathbf r_1 - \mathbf r_2\right|} \, ,
\end{equation}
for some transformed orbitals $\phi,\chi,\psi_a,\theta_a$.
We can then use the fact that $|\phi^* \chi| \le |\phi|^2 + |\chi|^2$, and similarly for $\psi_a$ and $\theta_a$ to upper bound this expression by
\begin{equation}
4\sum_{a=1}^\eta \int \textrm{d}\mathbf r_1 \, \textrm{d}\mathbf r_2 \frac{|\phi \left(\mathbf r_1\right)|^2 |\psi_a \left(\mathbf r_2\right)|^2 }{\left| \mathbf r_1 - \mathbf r_2\right|} \, ,
\end{equation}
for some choice of $\phi$ and $\psi_a$.
This integral can be maximized when the $\psi_a$ are orbitals that are clustered as close as possible to $\phi$.
With neighboring grid points separated by $\delta$, the smallest the average separation can be is ${\cal O}(\eta^{1/3}\delta)$. Then the factor of $1/|\mathbf r_1-\mathbf r_2|$ in the integrals will give ${\cal O}(1/[\eta^{1/3}\delta])$.
Multiplying the sum by $\eta$ gives ${\cal O}(\eta^{2/3}/\delta)$.

The second term for $V_{\mu\nu\sigma\lambda}$ is $\left(\mu\sigma|\lambda\nu\right)$.
This is similar to $\left(\mu\nu|\lambda\sigma\right)$, but with $\nu$ and $\sigma$ swapped.
Then the transformation of orbitals gives
\begin{equation}
\sum_{a=1}^\eta \int \textrm{d}\mathbf r_1 \, \textrm{d}\mathbf r_2 \frac{\phi^* \left(\mathbf r_1\right) \chi_a\left(\mathbf r_1\right)  \psi_a^* \left(\mathbf r_2\right) \theta \left(\mathbf r_2\right) }{\left| \mathbf r_1 - \mathbf r_2\right|} \, ,
\end{equation}
for some choice of $\phi,\chi_a,\psi_a,\theta$.
The same argument holds, where the sum is maximized with orbitals over a region of volume $\eta\delta^3$ so there are contributions from all $\eta$ terms in the sum, but $1/|\mathbf r_1-\mathbf r_2|$ averages to give ${\cal O}(1/[\eta^{1/3}\delta])$.
This gives the same scaling for the second term for $V_{\mu\nu\sigma\lambda}$, and so
\begin{equation}
    \| \mathbf V \|_\eta = {\cal O}(\eta^{2/3}/\delta).
\end{equation}
What this means is that, whenever we have a contraction of the $\sigma,\lambda$ indices in $V_{\mu\nu\sigma\lambda}$ with a normalized matrix of rank $\eta$, the remaining matrix has norm ${\cal O}(\eta^{2/3}/\delta)$.
That immediately implies that $\|\F\|$ has this norm.
Then applying $\F$ to the normalized matrix $\mathbf C_\text{occ}$ gives an upper bound on the first derivative ${\mathcal O}(\eta^{2/3}/\delta)$.

Taking the second derivative then yields an expression with 3 terms, where each has $\mathbf V$ appearing twice and $\mathbf C_\text{occ}$ appearing five times.
In particular,
\begin{align}
    -\frac{\partial^2 C_{\mu b}}{\partial t^2} &=  V_{\mu\nu\sigma\lambda} 
    [(V_{\sigma\epsilon\zeta\eta} C_{\zeta c}C^*_{\eta c}) C_{\epsilon a}
    C^*_{\lambda a}] C_{\nu b} \nn
    & \quad + V_{\mu\nu\sigma\lambda} [C_{\sigma a}
    (V_{\lambda\epsilon\zeta\eta} C^*_{\zeta c}C_{\eta c}) C^*_{\epsilon a}]
    C_{\nu b} \nn
    & \quad + (V_{\mu\nu\sigma\lambda} C_{\sigma a}C^*_{\lambda a})
    (V_{\nu\epsilon\zeta\eta} C_{\zeta c}C^*_{\eta c}) C_{\epsilon b}
\end{align}
Only the third line has a simple interpretation as $\mathbf F$ squared times $\mathbf C_\text{occ}$ (indicated by the brackets).

The first line has $\mathbf V$ contracted with $\mathbf C_\text{occ}$ using $\zeta,\eta$, so the expression in round brackets is a matrix with norm ${\mathcal O}(\eta^{2/3}/\delta)$.
Then in matrix terms, it is multiplied by $C_{\epsilon a}C^*_{\lambda a}$ (summed over $a$), which is a matrix of norm 1 and rank $\eta$.
As a result, the expression in square brackets is of norm ${\mathcal O}(\eta^{2/3}/\delta)$ and rank $\eta$.
We can then see that the first $\mathbf V$ is contracted over $\sigma,\lambda$ with a matrix of rank $\eta$ and norm ${\mathcal O}(\eta^{2/3}/\delta)$.
That implies that the norm of the resulting matrix is upper bounded by the square of ${\mathcal O}(\eta^{2/3}/\delta)$.
That is then multiplied by $C_{\nu b}$ which is of norm 1, resulting in the overall norm of this line being upper bounded by the square of ${\mathcal O}(\eta^{2/3}/\delta)$.
Similar considerations hold for the second line, so we can upper bound the entire second derivative by an order scaling that is the square of that for $\|\F\|$.

In this, the general principle is that wherever we have something of the form $C_{\sigma a}C^*_{\lambda a}$, it is a matrix of norm 1 and rank $\eta$, and taking the derivative of it yields something that is still of rank $\eta$, but with a norm upper bounded by ${\mathcal O}(\eta^{2/3}/\delta)$.
Because we have bounded the norm when contracting $\mathbf V$ with a general matrix of rank $\eta$, that yields a factor of ${\mathcal O}(\eta^{2/3}/\delta)$ on whatever result we had for the lower-order derivative.
The other scenario is where we take the derivative of $C_{\nu b}$, which is effectively like multiplying it by $\F$ which increases the norm (but not the rank).

This reasoning holds in general whenever we take the derivative of an expression for the derivative of some order to give the derivative of higher order.
The norm is multiplied by ${\mathcal O}(\eta^{2/3}/\delta)$ for each of the terms.
The number of terms will increase exponentially with the order.
The third derivative has $3\times 5$ terms, where each of the three original terms yields five due to the derivatives of $\mathbf C_\text{occ}$ at each location.
Then the fourth-order derivative has $3\times 5 \times 7$ terms and so on. In describing the scaling we can ignore this exponential number of terms, and give the upper bound on the $n$th order derivative as $\mathcal{O}(\eta^{2/3}/\delta)$ to the power of $n$.
This implies that the appropriate scaling of the time should again be $T=\|\F\|t$.

Finally we bound the norm of $\|\mathbf h\|$.
When using a plane wave basis, $h_{\mu \nu}$ will be non-zero only when $\mu = \nu$ with entries scaling as ${\cal O}(1/\delta^2)$ due to the $\nabla^2$ in the expression for $h_{\mu\nu}$. That gives the scaling of the spectral norm for this component, which would be unchanged under a unitary transformation, such as the Fourier transform which maps plane waves to an approximately local basis.

For the dependence of $h_{\mu \nu}$ on $V(r)$, the potential will come from nuclei, and for charge-neutral systems the total nuclear charge will be the same as the number of electrons.
If the nuclear charge were entirely at one location and we have a charge-neutral system, then the largest contribution to $h_{\mu \nu}$ would be for an approximately local basis, where the contribution would scale as $\eta/\delta$, with the factor of $\eta$ from the nuclear charge and $1/\delta$ from the inverse distance.

In most cases that we would be interested in, there would be a more even distribution of nuclear charges through the volume.
In that case, if the volume scales as $\eta$, there would be an average distance ${\cal O}(\eta^{1/3})$.
That would result in a contribution to $h_{\mu \nu}$ of ${\cal O}(\eta^{2/3})$.
An orbital localized near one nucleus would give a contribution of ${\cal O}(1/\delta)$ just from that nucleus, which may be larger than ${\cal O}(\eta^{2/3})$ if $N>\eta^3$ but may be ignored in comparison to $1/\delta^2$.

As a result of these considerations, we can give the upper bound on $\mathbf{F}$ in the case without $V(r)$ as
\begin{equation}
 \left\| \mathbf{F}\right \| = {\cal O}\!\left(\frac{\eta^{2/3}}{\delta} +  \frac{1}{\delta^2} \right).
\end{equation}
In the case with nuclei we obtain the same result,
provided the nuclear charges are not clustered any closer than the grid spacing. Here $\delta = {\cal O}((\eta / N)^{1/3})$ is the minimum grid spacing.
This scaling for $\delta$ comes from taking the computational cell volume proportional to $\eta$ (a reasonable assumption for both condensed-phase and molecular systems).
Thus, the scaling becomes
\begin{equation}
\|\mathbf{F}\| = \mathcal{O} \left( N^{1/3} \eta^{1/3} + \frac{N^{2/3}}{\eta^{2/3}} \right).
\label{eq:fock_norm}
\end{equation}
In this case we can see that the first term is dominant unless $N>\eta^3$.

\section{Proving sublinear gate complexity in basis size for Trotter based methods}
\label{app:trotter_step}

Here we derive the complexity for quantum simulation of the electronic structure problem given in \eq{Trotcomp}.
We consider the simulation of the electronic structure problem defined on a spatial grid in first quantization. Such a Hamiltonian can be expressed as
\begin{align}
\label{eq:real_space}
H & = T + U + V + \sum_{\ell \neq \kappa=1}^L\frac{\zeta_\ell \zeta_\kappa}{2\left\|R_\ell - R_\kappa\right\|} \\
T & \approx  \sum_{i=1}^{\eta} {\rm QFT}_j \left( \sum_{p\in G} \frac{\left \| k_p\right\|^2}{2} \ket{p}\!\!\bra{p}_{j} \right) {\rm QFT}_j^\dagger\\
U & = -\sum_{j=1}^\eta\sum_{\ell =1}^{L}  \sum_{p\in G}\frac{\zeta_\ell}{\left\|R_\ell - r_p\right\|} \ket{p}\!\!\bra{p}_{j}\\
V & = \sum_{j\neq k=1}^\eta \sum_{p,q\in G}\frac{1}{2\left\|r_p - r_q\right\|} \ket{p}\!\!\bra{p}_{j} \ket{q}\!\!\bra{q}_{k}
\end{align}
where ${\rm QFT}_j$ is the usual quantum Fourier transform applied to register $j$.
We emphasize that $T$ is only approximately given by the expression involving the QFT. This relation is exact in the continuum limit where $N\rightarrow\infty$. For finite-sized grids $N$, it cannot be the case that the QFT completely diagonalizes the momentum operator. Instead, writing $T$ this way represents something similar to the approximations made by so-called ``discrete value representation'' methods.
Using the QFT means that the evolution can be broken into a product of the evolution under $T$ and the one under $U+V$.

In the above expression, $\ell$ and $\kappa$ index nuclear degrees of freedom; thus, $R_\ell$ represents the positions of nuclei and $\zeta_\ell$ the atomic numbers of nuclei.
In this appendix, we use $L$ to denote the number of nuclei in our simulation (elsewhere, $L$ is the number of time points).
Furthermore, we have the following definition of grid points and their frequencies in the dual space defined by the QFT:
\begin{equation}
r_p = \frac{p \, \Omega^{1/3}}{N^{1/3}} \qquad \qquad k_p = \frac{2 \pi p}{\Omega^{1/3}} \qquad \qquad
p \in G \qquad \qquad G = \left[-\frac{N^{1/3}-1}{2},\frac{N^{1/3}-1}{2}\right]^3 \subset \mathbb{Z}^3 \, ,
\end{equation}
where $\Omega$ is the volume of the simulation cell and $N$ is the number of grid points in the cell. Although it is defined here in more precise terms, this is essentially the same representation used in the first work on quantum simulating chemistry in first quantization, by Kassal \emph{et al.}~\cite{Kassal2008}, well over a decade ago.

We consider simulation performed using high-order product formulas with a split-operator Trotter step.
What we mean by the latter is that we will alternate evolution under $T$ (using the QFT) and evolution under $U + V$.
In fact, the implementation of each Trotter step that we will pursue is essentially identical to the Trotter steps proposed by Kassal \emph{et al.}~\cite{Kassal2008}. The Trotter step requires $\widetilde{\cal O}(\eta^2)$ gates, with the complexity being dominated by computing the ${\cal O}(\eta^2)$ different interactions in the two-electron term.
Recently, Low \emph{et al.}~\cite{Low2022b} have shown that the number of Trotter steps required in second quantization using arbitrarily high order formulas can be as low as
\begin{equation}\label{eq:lowcomplex}
\left(N^{1/3} \eta^{1/3} + \frac{N^{2/3}}{\eta^{2/3}}\right)\frac{t^{1 + o(1)} N^{o(1)}}{\epsilon^{o(1)}} \, .
\end{equation}
We note that, curiously, this also closely matches our bound for the norm of the Fock operator (see \eq{fock_norm}) proved in \app{scanorm}.
The first term in brackets similarly corresponds to a contribution to the potential from electrons grouped as closely as possible in real space, but the reason why this quantity is relevant is very different between the two calculations.

The results for the Trotter error in second quantization also hold for first quantization.
As a general principle, we can consider the effect of $\sum_j \ket{p}\bra{q}_j$ on a computational basis state consisting of an anti-symmetric combination of lists of electron positions.
This removes an electron from orbital $q$ and places it in $p$.
This is performed for every part of the anti-symmetric state, preserving its sign.
However, for the starting anti-symmetric state the sign is based on whether the permutation is even or odd (as compared to ascending order).
If moving an electron from $q$ to $p$ passes over an \emph{odd} number of electrons, then the parity of each permutation flips.
That means that there is an overall sign flip in the basis state.

Similarly, if we consider the action of $a_p^\dagger a_q$ on a state $a_{q_1}^\dagger \cdots a_{q_\eta}^\dagger\ket{0}$, then the $a_q$ can be anti-commuted to the right to give several sign flips corresponding to the number of $a_{q_j}^\dagger$ operators that are anti-commuted through.
This corresponds to the number of occupied orbitals before $q$.
Then $a_q a_q^\dagger$ gives the identity.
Next, anti-commute $a_p^\dagger$ to the appropriate location in the list of operators.
The sign that is obtained corresponds to the number of $a_{q_j}^\dagger$ operators that are anti-commuted through, which is the number of electrons before $p$.
There is an overall sign flip if there is an odd number of electrons between $p$ and $q$.

This can then be extended to products such as
\begin{equation}
    \sum_j \ket{p}\bra{q}_j  \sum_k \ket{r}\bra{s}_k.
\end{equation}
The first sum corresponds to $a_p^\dagger a_q$ in second quantization, and the second sum corresponds to $a_r^\dagger a_s$.
This means that we have the equivalence
\begin{equation}
    \sum_{pqrs} V_{pqrs} \sum_j \ket{p}\bra{q}_j  \sum_k \ket{r}\bra{s}_k \equiv \sum_{pqrs} V_{pqrs} a_p^\dagger a_q a_r^\dagger a_s.
\end{equation}
The action on an anti-symmetric computational basis state in first quantization has exactly the same effects as that on the corresponding second-quantization state with $\eta$ electrons.
Moreover, the action of the operators always preserves the electron number in second quantization, so there is a corresponding state in first quantization.
Similarly, because we are using anti-symmetric states in first quantization, it is impossible to obtain a state with multiple electrons on the same orbital.
That is because two registers with the same orbital number will give cancellation of terms.

As a result all operators and states in second-quantization map directly to first quantization, preserving the norms, and in particular the error bounds derived in second-quantization hold for first quantization.
Therefore, multiplying the number of steps in \eq{lowcomplex} by the $\widetilde{\cal O}(\eta^2)$ gate complexity required of the first quantized Trotter step from \cite{Kassal2008} gives the following gate complexity for the product formula based time evolution in first quantization:
\begin{equation}
    \left(N^{1/3} \eta^{7/3} + N^{2/3}\eta^{4/3}\right)\frac{t^{1 + o(1)} N^{o(1)}}{\epsilon^{o(1)}} \, .
\end{equation}
This is the complexity given in \eq{Trotcomp}.

\section{Constant factors for time-evolution in the interaction-picture plane-wave algorithm}

\label{app:constant_factors}

Here we analyze the constant factors in the scaling of the interaction picture based plane wave algorithm from Babbush \emph{at al.}~\cite{BabbushContinuum} which was analyzed in detail for use in phase estimation by Su \emph{et al.}~\cite{Su2021}.
As explained on page 30 of \cite{Su2021}, the number of steps to give total time $T$ using the time evolution approach is $\lambda_B T/\ln 2$, but with a factor of 3 overhead for amplitude amplification.  Using the qubitization approach the number of steps is $e \lambda_B T$.  That means simulating the time evolution gives an overhead of $3/(e\ln 2)\approx 1.59$ over the qubitization.
Then in Eq.~(154) of \cite{Su2021}, the total time of evolution is approximately $\pi/(2\epsilon_{\rm pha})$ to give precision $\epsilon_{\rm pha}$ of the phase estimation.
There is moreover a (small) term $\mathcal{O}((\lambda_U+\lambda_V)^2\Delta E^2)$ in the expression for the number of steps $\mathcal{N}$ in \cite{Su2021} that originates from the nonlinearity of the sine function in phase estimation, which is not used here.

As a result, the complexity given in Theorem 5 of \cite{Su2021} can be modified to be appropriate for time evolution simply by replacing the formula for the number of steps in Eq.~(174) of \cite{Su2021} with
\begin{equation}
\mathcal{N} = \frac{ 3T (\lambda^1_U+\lambda^1_V/(1-1/\eta))}{P_{\rm eq} \ln 2} + \mathcal{O}(1) \, .
\end{equation}
Here we have replaced $\pi/(2\epsilon_{\rm pha})$ with $T$, replaced $e$ with $e/\ln 2$, and removed $\mathcal{O}((\lambda_U+\lambda_V)^2\Delta E^2)$.
Note that in this expression
\begin{align}
    \lambda_U &= \frac{\eta\sum_\ell \zeta_\ell}{\pi\Omega^{1/3}}\lambda_\nu, \\
    \lambda_V &= \frac{\eta(\eta-1)}{2\pi\Omega^{1/3}}\lambda_\nu, \\
    \lambda_\nu &= \sum_{\nu\in G_0} \frac 1{\|\nu\|^2} \le 4\pi N^{1/3},
\end{align}
$\lambda_U^1\approx \lambda_U$, $\lambda_V^1\approx \lambda_V$, and $P_{\rm eq}$ is close to 1.
This expression together with an appropriate choice of constant factor in $\Omega \propto \eta$ gives
the constant factor for the number of steps to use for time evolution.
It needs to be multiplied by a further complicated expression in Theorem 5 of \cite{Su2021} for the gate complexity of a single step to provide the full
constant factor for the gate complexity in \eq{interaction_scaling}.

\section{Smoothing the Coulomb operator to exponentially suppresses quantum scaling in basis size}
\label{app:softened_potential}

Here we discuss the fact that if one is willing to introduce a slight systematic bias into the Coulomb operator, it is possible to further improve the speedup in $N$ of the quantum algorithm. The ${N^{1/3}}$ dependence enters into the cost from the 1-norm of the two-body Coulomb operator, which scales as $\lambda = {\cal O}(\eta^2 V_{\rm max})$ where $V_{\rm max}$ is the maximum value of the electron-electron interaction for a single pair of electrons.

For typical plane wave or grid discretizations we have that $V_{\rm max} = {\cal O}(N^{1/3} / \Omega^{1/3}$) where $\Omega$ is the size of the computational cell (for the purpose of the analysis in this paper we assume that $\Omega = {\cal O}(\eta)$, since that is explicitly the case in condensed phase simulations). But we could also take steps to smooth out the cusp in the Coulomb operator and thus, lower the energy scale of $V_{\rm max}$. For example, this could be accomplished by taking $V_{\rm max}$ to be a constant and modifying the real-space form of the two-body Coulomb operator as
\begin{equation}
\frac{1}{\left | r_1 - r_2 \right |} \rightarrow \frac{1}{\left | r_1 - r_2\right | + V_{\rm max}}\, .
\end{equation}
Such a strategy has been explored in the context of first quantized quantum algorithms in real space in papers by Kivlichan \emph{et al.}~\cite{Kivlichan2016} and Childs \emph{et al.}~\cite{Childs2022}.

In principle, one could choose $V_{\rm max} = {\cal O}(\log N)$ and this would lead to the quantum algorithm scaling exponentially better than classical algorithms in $N$. This would also slightly reduce the cost of classical mean-field algorithms from scaling as $N^{4/3}$ to scaling as $N$. Of course, using such a drastic cutoff will introduce a significant bias into the overall dynamics. In order to avoid this, papers such as \cite{smoothCoulomb2006,smoothCoulomb2016} have sought to develop Richardson extrapolation type schemes where simulations are run with a series of smoothing or cutoff parameters in order to extrapolate the value of the observable with zero cutoff. However, questions remain about the convergence of such procedures and it seems likely to re-introduce some polynomial dependence on $N$ in order to reach convergence with the continuum limit.

Nevertheless, the context of this paper is that one might be interested in getting a speedup over low accuracy classical algorithms. In that spirit, one could probably make the case that if merely trying to improve in speed over mean-field algorithms, the error introduced in imposing a cutoff in the Coulomb operator might be less significant than the error due to making the mean-field approximation. Thus, this is perhaps a valid approach when competing with such classical methods, and thus might provide an exponential speedup.

\section{Gate complexity and speedup in various regimes}
\label{app:regimes}

\begin{table*}[h]
\begin{tabular}{c|c|c|c|c}
Processor
& Algorithm for sampling $\ket{\psi(t)}$
& Regime of optimality
& Space
& Effective gate complexity\\
\hline\hline
classical 
& zero temp mean-field with occ-RI-K/ACE \cite{Manzer2015Jul,Lin2016May}
& $N \leq \Theta(\eta^3)$
& $\widetilde{\cal O}(N \eta)$
& $N^{4/3} \eta^{7/3} t (N t / \epsilon)^{o(1)}$\\
classical 
& zero temp mean-field with occ-RI-K/ACE \cite{Manzer2015Jul,Lin2016May}
& $N \geq \Theta(\eta^3)$
& $\widetilde{\cal O}(N \eta)$
& $N^{5/3} \eta^{4/3} t (N t / \epsilon)^{o(1)}$\\
quantum
& second quantized Trotter grid algorithm \cite{Low2022b}
& $N \leq \Theta(\eta^2)$
& ${\cal O}(N \log N)$
&  $N^{4/3} \eta^{1/3} t (N t / \epsilon)^{o(1)}$ \\
quantum
& first quantized Trotter grid algorithm here
& $ \Theta(\eta^2) \leq N \leq \Theta(\eta^3)$
& ${\cal O}(\eta \log N)$
&  $N^{1/3}\eta^{7/3}t (N t / \epsilon)^{o(1)}$ \\
quantum
& first quantized Trotter grid algorithm here
& $ \Theta(\eta^3) \leq N < \Theta(\eta^4)$
& ${\cal O}(\eta \log N)$
&  $N^{2/3}\eta^{4/3}t (N t / \epsilon)^{o(1)}$ \\
quantum
& qubitization algorithms from \cite{BabbushContinuum} or \cite{Su2021} 
& $ N = \Theta(\eta^4)$
& ${\cal O}(\eta \log N)$
&  $\widetilde{\cal O}(N^{2/3}\eta^{4/3}t)$ \\
quantum
& interaction picture algorithms from \cite{BabbushContinuum} or \cite{Su2021}  
& $N > \Theta(\eta^4)$
& ${\cal O}(\eta \log N)$
&  $\widetilde{\cal O}(N^{1/3}\eta^{8/3} t)$ \\
\hline
\end{tabular}
\caption{\label{tab:regimes} Best known gate complexities of exact quantum algorithms and classical mean-field algorithms for sampling the output of time-evolution, by ratio of basis size to particle number. Here we use the asymptotic $\Theta(\cdot)$ notation, which implies the union of both an asymptotic upper-bound and an asymptotic lower-bound on the scaling. $N$ is number of basis functions, $\eta$ is number of particles, $\epsilon$ is target precision, and $t$ is duration of evolution. ``Effective gate complexity'' is the leading order scaling in the stated regime. All quantum algorithms discussed here require either a plane wave or grid basis. For those basis sets, the large space overhead of second quantization likely makes second quantized approaches infeasible in practice. When $N = \eta^4$, the quantum algorithms with the best asymptotic scaling are the plane wave or grid basis qubitization algorithms from \cite{BabbushContinuum} or \cite{Su2021}, respectively, as opposed to the interaction picture algorithms of those same works. This is due to lower polylogarithmic factors in the scaling that are suppressed by the $\widetilde{\cal O}(\cdot)$ notation.}
\end{table*}

Another way to express the results of \tab{regimes} is as a formula for the leading order scaling if assume that $N = \Theta(\eta^{\alpha})$. Then, for the classical algorithm we have that the leading gate complexity of the best approach is
\begin{equation}
\left(\eta^\beta t\right)\left(\frac{N t}{\epsilon}\right)^{o(1)} \qquad \textrm{where} \qquad N = \Theta\left(\eta^{\alpha}\right) \qquad \textrm{and} \qquad \beta = \begin{cases}
\frac{4 \alpha + 7}{3} & \alpha \leq 3\\
\frac{5 \alpha + 4}{3} & \alpha \geq 3\\
\end{cases} \, .
\end{equation}
By contrast, for the quantum algorithm we have that the leading order gate complexity of the best approach is
\begin{equation}
\left(\eta^\beta t\right)\left(\frac{N t}{\epsilon}\right)^{o(1)} \qquad \textrm{where} \qquad N = \Theta\left(\eta^{\alpha}\right) \qquad \textrm{and} \qquad \beta = \begin{cases}
\frac{4 \alpha + 1}{3} & \alpha \leq 2\\
\frac{\alpha + 7}{3} & 2 \leq \alpha \leq 3\\
\frac{2\alpha + 4}{3} & 3 \leq \alpha \leq 4\\
\frac{\alpha + 8}{3} & \alpha \geq 4\\
\end{cases} \, .
\end{equation}
For both classical and quantum expressions, these complexities are sometimes loose by sub-polynomial factors. Finally, we compare the speedup that exact quantum algorithms offer over classical mean-field algorithms. We report this as
\begin{equation}
\label{eq:speedup}
\frac{\textrm{exponent of $\eta$ scaling of classical complexity}}{\textrm{exponent of $\eta$ scaling of quantum complexity}} = \begin{cases}
\left(4\alpha + 7\right)/\left(4 \alpha + 1\right) & \alpha \leq 2\\
\left(4\alpha + 7\right)/\left(\alpha + 7\right) & 2 \leq \alpha \leq 3\\
\left(5\alpha + 4\right)/\left(2\alpha + 4\right) & 3 \leq \alpha \leq 4\\
\left(5\alpha + 4\right)/\left(\alpha + 8\right) & \alpha \geq 4\\
\end{cases}  
\qquad \textrm{if} \qquad N = \Theta\left(\eta^{\alpha}\right) \, .
\end{equation}
We plot numerical values of this speedup in \fig{speedup}.
\begin{figure}
  \centering
  \includegraphics[width=0.5\textwidth]{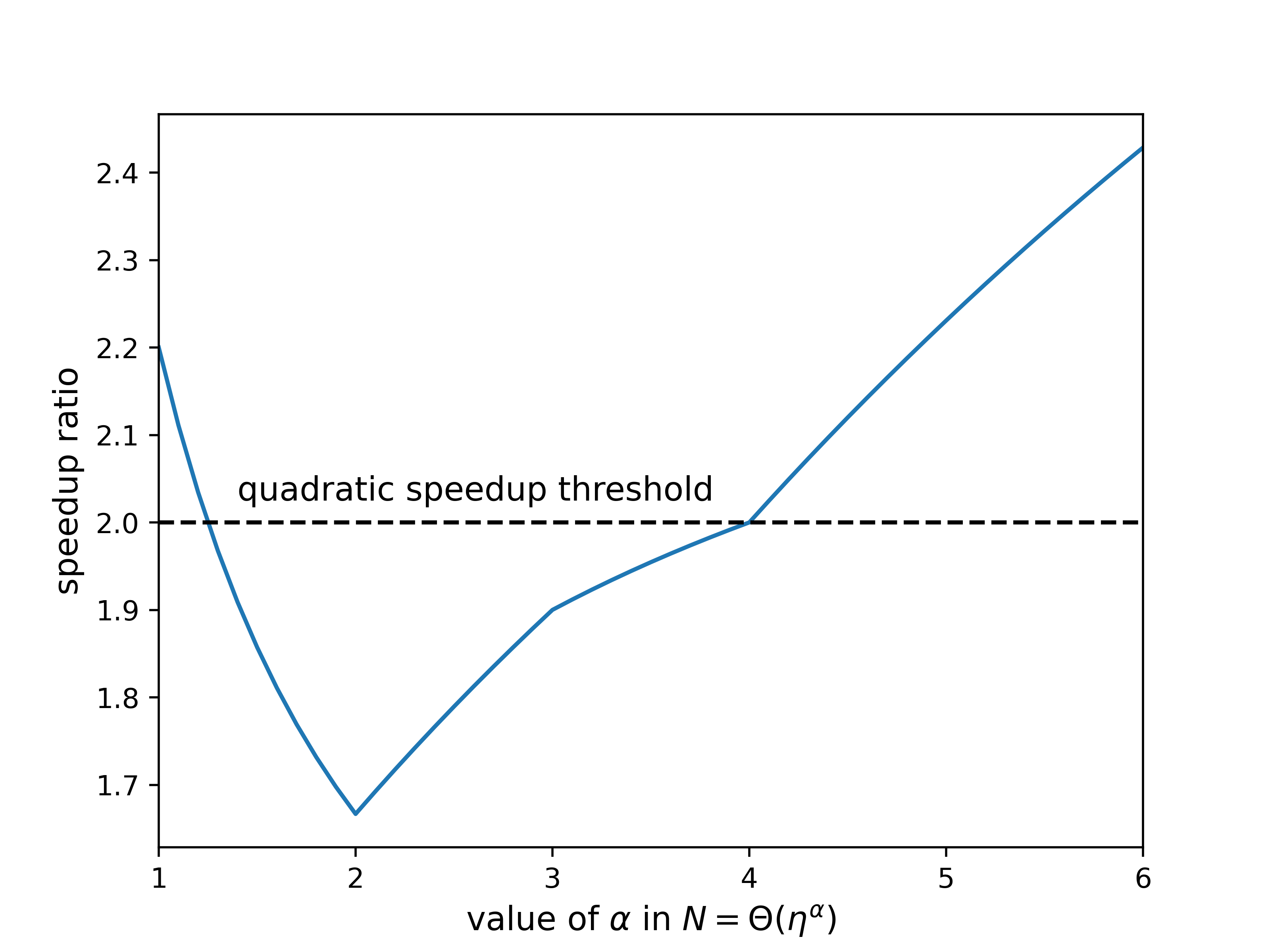}
  \caption{\label{fig:speedup} Plot showing the numerical values of the speedup exponent ratio given in \eq{speedup}. We see that a super-quadratic speedup of exact quantum algorithms over mean-field classical algorithms is realized when $\alpha < 5/4$ and when $\alpha > 4$.}
\end{figure}

Finally, we discuss the hope that Trotter based first quantized algorithms might be sped up by a factor of $\widetilde{\cal O}(\eta)$ by developing more efficient Trotter steps. The bottleneck for Trotter steps is the computation of the Coulomb operator since the simulation of the kinetic operator scales as $\widetilde{\cal O}(\eta)$. Thus, it seems promising that fast-multipole \cite{Rokhlin1985} Barnes-Hut \cite{BarnesHut}, or particle-mesh Ewald \cite{Ewald1993} type algorithms for computing the Coulomb potential require $\widetilde{\cal O}(\eta)$ operations in the classical random access memory (RAM) model. By contrast, the standard way of computing the Coulomb potential (involving summing up all ${\eta \choose 2}$ pairs of electrons) scales as $\widetilde{\cal O}(\eta^2)$. Thus, if one can figure out how to extend these better scaling methods to first quantization with $\widetilde{\cal O}(\eta)$ operations in the reversible circuit model (the cost model of relevance for this subroutine if executed on a quantum computer), the quantum algorithm would scale as
\begin{equation}
\left(N^{1/3}\eta^{4/3}t +  N^{2/3}\eta^{1/3}t\right)\left( \frac{N t}{\epsilon}\right)^{o\left(1\right)} \, .
\label{eq:fmm_cost}
\end{equation}

We note that it is straightforward to adapt these algorithms to second quantization with $\widetilde{\cal O}(N)$ gate complexity \cite{Low2018,Low2022b}. However, translating such algorithms to first quantization with $\widetilde{\cal O}(N)$ gate complexity in the quantum circuit model is highly non-trivial. This is due to nuances of how adaptive tree-like data structures are constructed and used in these algorithms, and it is why the work of \cite{Childs2022} decided to invoke the impractical assumption of QRAM in order to leverage the fast multipole algorithms. Note further that some of these algorithms such as the original fast multipole \cite{Rokhlin1985} and particle-mesh Ewald \cite{Ewald1993} make further assumptions on the state. In particular, if space is partitioned into ${\cal O}(\eta)$ boxes, then these methods require that no more than $k$ electrons are present in any box, in any configuration on which the wavefunction has support. Since electrons tend to repel one another this is often a good assumption at low energies, but it is not true for general states. It seems possible to implement a first quantized algorithm with $\widetilde{\cal O}(\eta \, k)$ space complexity and $\widetilde{\cal O}(\eta \, \textrm{poly}(k))$ gate complexity by keeping $k$ electron registers for each of these ${\cal O}(\eta)$ boxes of space. But there also exist versions of these algorithms, e.g. described in \cite{adaptiveFMM}, which use RAM and an adaptive tree structure to give $\widetilde{\cal O}(\eta)$ complexity without any assumptions on the state. Such approaches appear quite challenging to port to the quantum circuit model with the same complexity. However, if possible, the first quantized fast multipole-based Trotter would scale better than all other known approaches as long as $N < \eta^7$. When $N > \eta^7$, the first quantized interaction picture algorithm has better scaling.

\section{Efficient reduced density matrix estimation using classical shadows in first quantization}
\label{app:shadows}
\subsection{Problem statement}

We consider a system of \(\eta\) identical fermions occupying \(N \gg \eta\) orbitals.
In first-quantization, we represent the state of such a system as a wavefunction on \(\eta\) registers of \(n = \ceil{\log(N)}\) qubits.
We demand that this wavefunction is antisymmetric under the exchange of any two registers in order for it to represent a valid physical state.

Most physically interesting observables of such a system are captured by the few-body marginals, the reduced density matrices.
In this section, we concern ourselves with efficiently estimating elements of the \(k\)-body reduced density matrix ($k$-RDM) of the first-quantized state \(\ket{\psi}\) defined on \(\eta\) identical fermion particles,
\begin{equation}
    \prescript{k}{}D_{i_1, \ldots, i_k}^{j_1, \ldots, j_k} = \frac{\eta !}{\left( \eta - k \right)!} \tr \left[\ketbra{\psi}\prod_{\ell=1}^k \ketbra{i_\ell}{j_\ell}_\ell \right],
    \label{eq:1_register_k_rdm}
\end{equation}
where \(\ketbra{i}{j}_\ell\) indicates the tensor product of \(\ketbra{i}{j}\) on the \(\ell\)th register with the identity on the other \(\eta-1\) registers.
We can write an equivalent definition (equivalent due to the antisymmetry of the wavefunction),
\begin{equation}
    \prescript{k}{}D_{i_1, \ldots, i_k}^{j_1, \ldots, j_k} = \sum_{\bm{x} \in S_k^\eta} \tr \left[\ketbra{\psi} \prod_{\ell=1}^k \ketbra{i_\ell}{j_\ell}_{x_\ell} \right],
    \label{eq:k_register_k_rdm}
\end{equation}
where \(S_k^\eta\) is the set composed of all possible sequences of length \(k\) generating by drawing without replacement from \([\eta] \coloneqq \left\{ 1,\ldots,\eta \right\}\).

Our goal is to use measurements of the state \(\ket{\psi}\) to obtain a classical description of the state with enough information to approximate all \(N^{2k}\) elements of the $k$-RDM. We would like all of these estimates to accurate up to some additive error \(\epsilon\) with probability at least \(1 - \delta\). Ideally, our protocol will be efficient not only in terms of the number of measurements, but also in terms of the (gate) complexity of implementing each measurement and the classical complexity of the required post-processing.

We will accomplish our goal by applying the classical shadows formalism of Ref.~\citenum{Huang2020}.
We propose and analyze a protocol that requires at most
\begin{equation}
    m = 64 e^3 \log\left(N/\delta\right) k  \left( 2k + 2e \right)^k \eta^k \epsilon^{-2}
    \label{eq:real_m}
\end{equation}
measurements to estimate the $k$-RDM.
Performing these measurements requires acting on each of the particle registers with a randomly sampled Clifford circuit and performing a measurement in the computational basis.
These circuits can be implemented using \(\mathcal{O}(\eta n^2)\) one- and two-qubit Clifford gates on a linearly connected array of qubits in depth \(\mathcal{O}(n)\).
Each element of the $k$-RDM requires performing a number of classical operations that scales as 
\begin{equation}
    m' = \mathcal{O}\left(\left(n^4 + \log\left(1/\delta\right)\right) \eta^{2k} 2^k \epsilon^{-2}\right).
    \label{eq:real_m_prime}
\end{equation}

\subsection{The measurement protocol}

The classical shadows formalism of Huang et al.~works by choosing an ensemble of random unitaries \(\mathcal{U}\) on \(n\) qubits and defining a measurement channel
\begin{equation}
    \mathcal{M}(\sigma) \coloneqq \mathbb{E}_{U \sim \mathcal{U}} \sum_{b \in \{0, 1\}^n}
    U^\dagger \ketbra{b} U \ev{U \sigma U^\dagger}{b}.
\end{equation}
For specific choices of \(\mathcal{U}\), the channel \(\mathcal{M}\) is analytically invertible.
Operationally, we obtain the classical shadow of \(\sigma\) by repeatedly sampling a unitary \(U\) from \(\mathcal{U}\), applying the sampled \(U\) to a copy of \(\sigma\), and measuring in the computational basis (obtaining the bitstring \(b\)).
If we collect \(m\) such samples, then we call the (potentially unphysical) state
\begin{equation}
    \hat{\sigma} \coloneqq \frac{1}{m} \sum_{i=1}^m \mathcal{M}^{-1}\left(U^\dagger_i \ketbra{b_i} U_i \right)
\end{equation}
a classical shadow of \(\sigma\).
For an arbitrary observable \(O\), we can define an estimator \(\hat{o}\) of the quantity \(\tr \left[O \rho \right]\) using the classical shadow of \(\rho\),
\begin{equation}
    \hat{o} \coloneqq \tr \left[O \hat{\rho}\right].
\end{equation}

In expectation, we have that
\begin{equation}
    \ev{\hat{\sigma}} = \mathbb{E}_{U \sim \mathcal{U}} \sum_{b \in \{0, 1\}^n} \mathcal{M}^{-1}\left(U^\dagger \ketbra{b} U\right) \ev{U \sigma U^\dagger}{b} = \mathcal{M}^{-1}\left(\mathcal{M}\left(\sigma\right)\right) = \sigma.
    \label{eq:shadow_average}
\end{equation}
When we take \(\mathcal{U}\) to be the uniform distribution over the Clifford group on \(n\) qubits, the classical shadows measurement channel and its inverse have particularly simple forms~\cite{Huang2020},\footnote{Actually, a substantial constant factor savings in the number of gates can be obtained by using the canonical form of Ref.~\citenum{Bravyi2021-rn} and simply dropping the permutation at the end of the circuit.
    See, e.g., Ref.~\citenum{Huggins2022-nc}.
}
\begin{align}
    \mathcal{M}(A)      & = \frac{1}{2^n + 1}
    A + \frac{\tr\left[A\right]}{2^n + 1} \mathbb{I},
    \\
    \mathcal{M}^{-1}(A) & = (2^n + 1) A - \tr\left[A\right] \mathbb{I}.
    \label{eq:inverse_channel_exp}
\end{align}
Here, and throughout our analysis of the measurement protocol, we use the symbol \(\mathbb{I}\) to denote the identity operator on a Hilbert space whose dimension is appropriate for the context.

In this work, we propose and analyze the impact of using an ensemble \(\mathcal{U}\) that consists of a tensor product of \(\eta\) copies of the uniform distribution over \(n\) qubit Clifford circuits,
\begin{equation}
    \mathcal{U} = \bigotimes_{j=1}^\eta \mathrm{Cl}(2^n).
    \label{eq:U_def}
\end{equation}
That is to say, we perform our measurements by independently sampling \(\eta\) \(n\)-qubit Clifford unitaries, applying one to each particle register, and measuring in the computational basis.
We can consider the action of the corresponding classical shadow measurement channel and its inverse on an operator \(X_1 \otimes \cdots \otimes X_\eta\) that factorizes across the \(\eta\) registers. The channel is defined on the whole Hilbert space by linear extension.
For the classical shadow measurement channel, we have
\begin{align}
    \mathcal{M}(X_1 \otimes \cdots \otimes X_\eta) 
    &=
    \bigotimes_{j=1}^\eta \left( \mathbb{E}_{U_j \sim \mathrm{Cl}(2^n)} \sum_{b_j \in \{0,1\}^n} U_j^\dagger \ketbra{b_j} U_j \ev{U_j X_j U_j^\dagger}{b_j}\right)
    \\ &=
    \bigotimes_{j=1}^\eta \left( \frac{X_j + \tr\left[ X_j \right]\mathbb{I}}{2^n + 1} \right).
    \label{eq:M_channel}
\end{align}
The inverse, similarly, is given by
\begin{equation}
    \mathcal{M}^{-1}(X_1 \otimes \cdots \otimes X_\eta) = \bigotimes_{j=1}^\eta \left(\left(2^n + 1\right)X_j - \tr\left[X_j\right]\mathbb{I} \right).
    \label{eq:M_channel_inverse}
\end{equation}

Due to the antisymmetry of the wavefunction, we have the freedom to choose between a number of different observables when estimating the elements of the $k$-RDM.
Consider an arbitrary operator \(O\), and the operator \(P O P^\dagger\), where \(P\) is an operator that permutes the particle registers.
The expectation values of \(O\) and \(P O P^\dagger\) with respect to a first-quantized wavefunction are the same (to see this, observe that any sign picked up by acting \(P^\dagger\) on the ket is cancelled out by a corresponding sign obtained from acting \(P\) on the bra).
We can use this degree of freedom to minimize the variance of our measurement protocol.
Using the observable from \eq{1_register_k_rdm} to construct a classical shadow estimator of a $k$-RDM element would lead to an unnecessarily large variance, essentially because the observable doesn't take advantage of all of the information present in the state.
In contrast, \eq{k_register_k_rdm} defines the $k$-RDM element in terms of a sum over many different permutations of the registers.
We conjecture that a measurement protocol based on the observable in \eq{k_register_k_rdm} would perform well, but the analysis could be tedious due to the many different cases that would arise.

Rather than using the observables implied by either \eq{1_register_k_rdm} or \eq{k_register_k_rdm} in our classical shadow measurement procedure, we instead choose to estimate the $k$-RDM elements using an observable that involves a sum over a simpler set of permutations.
Essentially, we break the \(\eta\) registers up into \(k\) groups of size \(\eta /k\) and measure the $k$-RDM element using registers from each group.
For ease of notation, let us assume that \(\eta\) is divisible by \(k\).\footnote{In the event that \(\eta\) is not exactly divisible by \(k\), one could modify the protocol to either use groups of slightly different sizes or to only perform the measurements using \(\eta' = k \lfloor\eta / k \rfloor\) registers.}
Formally, we can define a set of sequences
\begin{equation}
    R_k =
    \left\{ 1,\ldots, \eta / k \right\} \times
    \left\{ \eta / k + 1,\ldots, 2\eta / k \right\} \times
    \cdots \times
    \left\{ \left( k - 1 \right)\eta / k + 1,\ldots, \eta \right\}.
    \label{eq:R_k_def}
\end{equation}
Due to the antisymmetry of the wavefunction, we have that 
\begin{equation}
    \prescript{k}{}D_{i_1, \ldots, i_k}^{j_1, \ldots, j_k} = 
    \frac{k^k \left( \eta ! \right)}{\eta ^ k\left( \eta - k \right)!} 
    \sum_{\bm{x} \in R_k} \tr \left[\ketbra{\psi} \prod_{\ell=1}^k \ketbra{i_\ell}{j_\ell}_{x_\ell} \right].
    \label{eq:R_k_register_k_rdm}
\end{equation}
We define an estimator \(\hat{d}\) for the $k$-RDM element \(\prescript{k}{}D_{i_1, \ldots, i_k}^{j_1, \ldots, j_k}\) using the classical shadow \(\hat{\rho}\) of \(\ket{\psi}\),
\begin{equation}
    \hat{d} = 
    \frac{k^k \left( \eta ! \right)}{\eta ^ k\left( \eta - k \right)!} 
    \sum_{\bm{x} \in R_k} \tr \left[\hat{\rho} \prod_{\ell=1}^k \ketbra{i_\ell}{j_\ell}_{x_\ell} \right].
    \label{eq:def_hat_d}
\end{equation}

In \app{shadows_k}, we prove that the single-shot variance of this estimator is bounded by
\begin{equation}
    \mathrm{Var}(\hat{d}) \leq
    e^3 \eta^k \left( 2k + 2e \right)^k.
    \label{eq:d_hat_variance_bound_intro}
\end{equation}
In order to guarantee that our estimates are close to the true value of the $k$-RDM elements with high probability, we need to proceed along the same lines as Ref.~\citenum{Huang2020} and construct a median-of-means estimator to obtain the desired rigorous guarantees~\cite{Lerasle2019-zm}.
To be precise, using Proposition 12 from Ref.~\citenum{Lerasle2019-zm}, we can consider an estimator that divides the \(m\) total classical shadow samples into \(K\) groups of size \(b\), and takes the median of the sample mean obtained by averaging the estimates within each group.
The probability that this median of means estimator has an error larger than \(2 \sqrt{\mathrm{Var}(\hat{d}) / b}\) is at most \(e^{-K/8}\).
To bound the error in our estimate by \(\epsilon\) with a success probability of at least \(1 - \delta\), this implies that we need 
\begin{align}
    b &= 4 \, \mathrm{Var}(\hat{d}) / \epsilon^2, 
    \\
    K &= 8 \log\left( 1/\delta \right).
\end{align}
The overall number of measurements claimed in \eq{real_m} follows directly from applying a union bound over the failure probabilities for estimating all \(N^{2k}\) $k$-RDM elements.

The measurement protocol can be summarized as follows.
We take a classical shadow of \(\ket{\psi}\) with the \(\mathcal{U}\) defined in \eq{U_def} using a number of samples \(m\) chosen according to \eq{real_m}.
For each sample, we evaluate the expectation values of the \(\left( \eta/k \right)^k\) different terms in the sum over \(R_k\) (see \eq{def_hat_d}) using generalizations of Gottesman-Knill theorem that account for the phase of the quantities involved~\cite{Garcia2017-rx,Aaronson2004-yb,Gottesman1998-rv}. Breaking the samples into \(K\) groups of size \(b\), averaging within the groups, and then taking the median of these means then yields the final estimate.
The classical post-processing costs quoted in \eq{real_m_prime} come from counting the number of \(n\)-qubit sized Clifford circuits that need to be simulated classically to carry out this procedure.

\subsection{Notation and preliminaries}
\label{app:simplification}

Before we proceed to bound the variance of the estimator \(\hat{d}\) for an arbitrary $k$-RDM element, it is helpful to recall a few useful expressions and prove some identities that we will use later.

We will make use of a formula for the two-fold twirl over the Clifford group and partial trace obtained from Ref.~\citenum{Huang2020},
\begin{equation}
    \label{eq:2_twirl_k}
    \mathbb{E}_{U \sim \textrm{Cl}(2^n)} U^\dagger \ketbra{x} U \ev{U A U^\dagger}{x} = \frac{A + \tr(A) \mathbb{I}}{2^n(2^n+1)}.
\end{equation}
For the three-fold twirl and partial trace, we find it convenient to use the identity
\begin{multline} \label{eq:3_twirl_k}
    \mathbb{E}_{U \sim \mathrm{Cl}(2^n)} U^\dagger \ketbra{x} U \ev{U B U^\dagger}{x} \ev{U C U^\dagger}{x} =
    \\
    \frac{1}{2^n \left( 2^n + 1 \right)\left( 2^n +2 \right)}
    \left( \mathbb{I}\left( \tr\left[ B C \right] +
        \tr\left[ B \right]\tr\left[ C \right] \right) + B \tr\left[ C \right] + C\tr\left[ B \right] +BC
    +CB \right).
\end{multline}
This equation is different from the corresponding one considered in previous work (Eq.~(S36) of Ref.~\citenum{Huang2020}),
in that it allows for \(B\) and \(C\) to have non-zero trace. It can be obtained directly from the analysis of Ref.~\citenum{Gross2015-ck}.\footnote{Note that while the proof of Lemma 7 in Ref.~\citenum{Gross2015-ck} is technically for Hermitian matrices, the same proof holds exactly in the non-Hermitian case.}

Another small departure we make from some prior work is that we directly consider the variance of estimators for the expectation values of non-Hermitian observables.
For a classical shadow \(\hat{\rho}\) of a state \(\rho\) and an estimator \(\hat{o} = \tr\left[ \hat{\rho} O \right]\) of the expectation value of a (not necessarily Hermitian) operator \(O\), we have
\begin{align}
    \mathrm{Var}(\hat{o}) &= \tr \left[ \rho \sum_{b} \mathbb{E}_{U \sim \mathcal{U}}  U^\dagger\ketbra{b} U \ev{U \mathcal{M}^{-1}(O) U^\dagger}{b} \ev{U \mathcal{M}^{-1}(O^\dagger) U^\dagger}{b}\right] - \left| {\tr \left[ O \rho \right]} \right|^2
    \\
    & \leq 
    \tr \left[ \rho \sum_{b} \mathbb{E}_{U \sim \mathcal{U}}  U^\dagger\ketbra{b} U \ev{U \mathcal{M}^{-1}(O) U^\dagger}{b} \ev{U \mathcal{M}^{-1}(O^\dagger) U^\dagger}{b}\right].
    \label{eq:var_general}
\end{align}
This expression can be arrived at from the definition of the variance of a complex-valued random variable applied to the classical shadow formalism.
We refer the reader to Ref.~\citenum{Wan2022} for a thorough discussion.

In the course of calculating the variance for the higher-order RDMs, we will find that we repeatedly need to simplify certain expressions.
Before describing those expressions and showing how they may be simplified, let us define some notation used for convenience throughout the rest of our analysis:
\begin{align}
    P_x &= \ketbra{x},
    \\
    P_{xy} &= \ketbra{x}{y},
    \\
    \mathbb{E}_U &= \mathbb{E}_{U \sim \mathrm{Cl}(2^n)},
    \\
    \sum_{b} &= \sum_{b \in \left\{ 0,1 \right\}^n}.
\end{align}

One class of expressions that we will need to simplify are of the form
\begin{equation}
    A = \mathbb{E}_{U} \sum_{b} U^\dagger P_b U \ev{U \mathcal{M}^{-1}(P_{ij}) U^\dagger}{b}.
    \label{eq:A_2_generic}
\end{equation}

We can use \eq{inverse_channel_exp} and \eq{2_twirl_k} to simplify \eq{A_2_generic},
\begin{align}
    A &= \mathbb{E}_{U} \sum_{b} U^\dagger P_b U \ev{U \mathcal{M}^{-1}(P_{ij}) U^\dagger}{b}
    \\ & = 
    \mathbb{E}_{U} \sum_{b} U^\dagger P_b U \ev{U \left( \left( 2^n + 1 \right) P_{ij} - \delta_{i,j}\mathbb{I} \right) U^\dagger}{b}
    \\ & = 
    P_{ij}.
    \label{eq:A_2_simplify}
\end{align}

Another kind of expression that we will need to simplify is of the form
\begin{equation}
    A = \mathbb{E}_{U} \sum_{b} U^\dagger P_b U \ev{U \mathcal{M}^{-1}(P_{ij}) U^\dagger}{b} \ev{U \mathcal{M}^{-1}(P_{kl}) U^\dagger}{b}.
    \label{eq:A_3_generic}
\end{equation}

Let us consider the first case, and simplify \(A\) as defined below,
\begin{equation}
    A = \mathbb{E}_{U} \sum_{b} U^\dagger P_b U \ev{U \mathcal{M}^{-1}(P_{i}) U^\dagger}{b} \ev{U \mathcal{M}^{-1}(P_{i}) U^\dagger}{b}.
    \label{eq:A_3_case_iiii}
\end{equation}
We have
\begin{align}
    \mathcal{M}^{-1}\left( P_{i}\right) = \left( 2^n + 1 \right)P_i - \mathbb{I}
\end{align}
by an application of \eq{inverse_channel_exp}.
Now we can apply \eq{3_twirl_k} with \(B = C = \left( 2^n + 1 \right)P_i - \mathbb{I}\).

Let us simplify the pieces of \eq{3_twirl_k} separately before combining them.
We have 
\begin{align}
    BC &= CB = \left( \left( 2^n + 1 \right)P_i - \mathbb{I} \right)^2
    \\
    &= \left( 2^n + 1 \right)\left( 2^n - 1 \right)P_i + \mathbb{I},
    \\
    \tr \left[ BC \right] &= \tr \left[ CB \right] = 2^{n}(2^n + 1) - 1,
    \\
    \tr \left[ B \right] &= \tr \left[ C \right] = 1.
\end{align}
As a result, 
\begin{align}
    &
    \mathbb{I}\left( \tr\left[ B C \right] +
    \tr\left[ B \right]\tr\left[ C \right] \right) + B \tr\left[ C \right] + C\tr\left[ B \right] +BC
    +CB
    \\
    = \,&
    2^n \left( 2^n + 1 \right) \mathbb{I} + 2\left( 2^n + 1 \right) P_i - 2\mathbb{I} + 2 \left( 2^n + 1 \right)\left( 2^n - 1 \right)P_i + 2 \mathbb{I}
    \\
    = \,&
    2^n \left( 2^n + 1 \right) \left(\mathbb{I} + 2 P_i  \right).
\end{align}

Putting everything together, we have 
\begin{align}
    A &=
    \mathbb{E}_U \sum_b U^\dagger P_b U \ev{U \mathcal{M}^{-1}(P_{i}) U^\dagger}{b} \ev{U \mathcal{M}^{-1}(P_{i}) U^\dagger}{b}
    \\
    &= \frac{2^n}{2^n + 2}\left( \mathbb{I} +2 P_{i}\right).
    \label{eq:A_3_case_iiii_simplified}
\end{align}

Now we consider simplifying the expression
\begin{equation}
    A = \mathbb{E}_{U} \sum_{b} U^\dagger P_b U \ev{U \mathcal{M}^{-1}(P_{ij}) U^\dagger}{b} \ev{U \mathcal{M}^{-1}(P_{ji}) U^\dagger}{b}.
    \label{eq:A_3_case_ijji}
\end{equation}
In this case, we can again use \eq{3_twirl_k} with
\begin{align}
    B &= \mathcal{M}^{-1}(P_{i j}) = \left( 2^n + 1 \right)P_{i j},
    \\
    C &= \mathcal{M}^{-1}(P_{j i}) = \left( 2^n + 1 \right)P_{j i}.
\end{align}

Working out some of the pieces, we have
\begin{align}
    BC &= \left( 2^n + 1 \right)^2 P_i,
    \\
    CB &= \left( 2^n + 1 \right)^2 P_j,
    \\
    \tr \left[ BC \right] &= \left( 2^n + 1 \right)^2,
    \\
    \tr \left[ B \right] &= \tr \left[ C \right] = 0.
\end{align}
Therefore,
\begin{align}
    &
    \mathbb{I}\left( \tr\left[ B C \right] +
    \tr\left[ B \right]\tr\left[ C \right] \right) + B \tr\left[ C \right] + C\tr\left[ B \right] +BC
    +CB
    \\
    &= \left( 2^n + 1 \right)^2 \left( \mathbb{I} + P_i + P_j \right).
\end{align}

Finally, we have 
\begin{align}
    A &=
    \mathbb{E}_U \sum_b U^\dagger P_b U \ev{U \mathcal{M}^{-1}(P_{ij}) U^\dagger}{b} \ev{U \mathcal{M}^{-1}(P_{ji}) U^\dagger}{b}
    \\
    &= \frac{2^n + 1}{2^n + 2}\left( \mathbb{I} + P_i + P_j \right).
    \label{eq:A_3_case_ijji_simplified}
\end{align}

\subsection{Variance of the $k$-RDM with a restricted sum}
\label{app:shadows_k}

Now we are ready to turn to the task of bounding the variance \(\hat{d}\) as defined in \eq{def_hat_d}.
For now, we neglect the coefficient in order to simplify the presentation.
Let
\begin{align}
    O &= \sum_{\bm{x} \in R_k} O_{\bm{x}},
    \\
    O_{\bm{x}} &= \prod_{\ell=1}^k \ketbra{i_\ell}{j_\ell}_{x_\ell}.
\end{align}
The variance of the classical shadow estimator \(\hat{o}\) of \(\ev{O}\) is bounded by 
\begin{align}
    \mathrm{Var}(\hat{o}) &\leq \sum_{\bm{x} \in R_k} \sum_{\bm{y} \in R_k}\tr \left[ \ketbra{\psi} A_{\bm{x}\bm{y}}\right], \label{eq:o_variance_tr_A}
    \\
    A_{\bm{x}\bm{y}} &= \sum_{b} \mathbb{E}_{U \sim \mathcal{U}}  U^\dagger\ketbra{b} U \ev{U \mathcal{M}^{-1}(O_{\bm{x}}) U^\dagger}{b} \ev{U \mathcal{M}^{-1}(O_{\bm{y}}^\dagger) U^\dagger}{b}.
\end{align}

Because the inverse channel, the random unitaries, and the \(O_{\bm{x}}\) all factorize across the registers, we can rewrite \(A_{\bm{x}\bm{y}}\) as a tensor product,
\begin{align}
    A_{\bm{x}\bm{y}} & = \bigotimes_{z = 1}^\eta A_{\bm{x}\bm{y}}^z,
\end{align}
where \(A_{\bm{x}\bm{y}}^z\) takes one of three forms depending on whether neither, one of, or both of \(O_{\bm{x}}\) and \(O_{\bm{y}}\) act non-trivially on the \(z\)th register.
If \(z \notin \bm{x}\) and \(z \notin \bm{y}\), then
\begin{equation}
    A_{\bm{x}\bm{y}}^z = \mathbb{I}.
\end{equation}

If exactly one of \(z \in \bm{x}\) or \(z \in \bm{y}\) is true, then we can use \eq{A_2_simplify} to simplify our expression for \(A_{\bm{x}\bm{y}}^z\).
The cases are symmetric between \(z \in \bm{x}\) and \(z \in \bm{y}\), so we can treat only the first case without loss of generality.
Let \(\ell\) denote the index of \(z\) in \(\bm{x}\) (i.e., \(x_\ell = z\)).
We have
\begin{align}
    A_{\bm{x}\bm{y}}^z &= \sum_{b} \mathbb{E}_{U \sim \mathcal{U}}  U^\dagger\ketbra{b} U \ev{U \mathcal{M}^{-1}(\ketbra{i_\ell}{j_\ell}) U^\dagger}{b}
    \\
    &= \ketbra{i_\ell}{j_\ell}_z.
\end{align}
If \(z \in \bm{y}\) we instead have \(A_{\bm{x}\bm{y}}^z = \ketbra{j_\ell}{i_\ell}_z\).

The third case we must consider is where \(z \in \bm{x}\) and \(z \in \bm{y}\).
Let \(\ell\) denote the index of \(z\) in \(\bm{x}\) and \(\bm{y}\) (they must be the same because of the way we construct \(\bm{x}\) and \(\bm{y}\)).
In this case,
\begin{align}
    A_{\bm{x}\bm{y}}^z &= \sum_{b} \mathbb{E}_{U \sim \mathcal{U}}  U^\dagger\ketbra{b} 
    U \ev{U \mathcal{M}^{-1}(\ketbra{i_\ell}{j_\ell}) U^\dagger}{b}
    U \ev{U \mathcal{M}^{-1}(\ketbra{j_\ell}{i_\ell}) U^\dagger}{b}.
\end{align}
If \(i_\ell = j_\ell\) we can simplify this expression using \eq{A_3_case_iiii_simplified}, otherwise we can use \eq{A_3_case_ijji_simplified}. The combination of these two formulas lets us write 
\begin{align}
    A_{\bm{x}\bm{y}}^z &= \sum_{b} \mathbb{E}_{U \sim \mathcal{U}}  U^\dagger\ketbra{b} 
    U \ev{U \mathcal{M}^{-1}(\ketbra{i_\ell}{j_\ell}) U^\dagger}{b}
    U \ev{U \mathcal{M}^{-1}(\ketbra{j_\ell}{i_\ell}) U^\dagger}{b}
    \\
    &= \frac{2^n + 1 - \delta_{i_\ell, j_\ell}}{2^n + 2}\left( \mathbb{I} + \ketbra{i_\ell} +
    \ketbra{j_\ell} \right).
    \label{eq:A_3_both_cases}
\end{align}

Now we will use the antisymmetry of \(\ket{\psi}\) to bound the quantity \(\left|\tr \left[ \ketbra{\psi} A_{\bm{xy}} \right]\right|\).
Let 
\begin{equation}
    a = |\bm{x} \cap \bm{y}|,
    \qquad \qquad
    b = 2k - 2a. 
    \label{eq:a_b_def}
\end{equation}
The operator \(A_{\bm{xy}}\) acts non-trivially on \(a + b\) registers.
On \(a\) registers, it acts with an operator of the form given in \eq{A_3_both_cases}.
On the other \(b\) registers, it acts as \(\ketbra{c}{d}\) for some \(c, d\) (that can vary per register).
Due to the antisymmetry of \(\ket{\psi}\), we can freely permute the registers without affecting the expectation value.

We can therefore rewrite the expectation value of interest as
\begin{align}
    \left|\tr \left[ \ketbra{\psi} A_{\bm{xy}} \right]\right| &= 
    \left|\ev{\left(\bigotimes_{\ell=1}^a \frac{2^n + 1 - \delta_{c_\ell, d_\ell}}{2^n + 2}\left( \mathbb{I} + \ketbra{c_\ell} + \ketbra{d_\ell} \right)
    \bigotimes_{\ell = a + 1}^{a + b} \ketbra{c_\ell}{d_\ell} 
    \bigotimes_{\ell = a + b + 1}^\eta \mathbb{I}\right)}{\psi}\right|
    \\
    &\leq 
    \left|\ev{\left(\bigotimes_{\ell=1}^a \left( \mathbb{I} + \ketbra{c_\ell} + \ketbra{d_\ell} \right)
    \bigotimes_{\ell = a + 1}^{a + b} \ketbra{c_\ell}{d_\ell} 
    \bigotimes_{\ell = a + b + 1}^\eta \mathbb{I}\right)}{\psi}\right|.
    \label{eq:A_bound_with_offdiags}
\end{align}

\subsubsection{Removing the off-diagonal terms}

We can simplify the bound in \eq{A_bound_with_offdiags} by replacing the off-diagonal matrix elements with projectors.
To do so, we will need the following lemma.

\begin{lemma}
    Let \(\ket{\psi}\) be an arbitrary normalized pure quantum state on \(n\) qubits.
    Let \(O\) be an arbitrary positive semidefinite operator on \(a\) qubits, and let \(\ket{\alpha}\) and \(\ket{\beta}\) be arbitrary orthonormal quantum states on \(n - a\) qubits.
    Then,
    \begin{equation}
        \left| \ev{\left(O \otimes \ketbra{\alpha}{\beta}\right)}{\psi} \right| \leq 
        \left| \ev{\left(O \otimes \ketbra{\phi}{\phi}\right)}{\psi} \right|
    \end{equation}
    for \(\ket{\phi} = \ket{\alpha}\) or \(\ket{\phi} = \ket{\beta}\).
    \label{lem:offdiag_to_diag}
\end{lemma}
\begin{proof}
    To begin the proof, expand \(\ket{\psi}\) as
    \begin{equation}
        \ket{\psi} = \sum_{ij} c_{ij} \ket{i}\ket{j},
    \end{equation}
    where the states \(\left\{ \ket{i} \right\}\) form an eigenbasis for \(O\) and the states \(\left\{ \ket{j} \right\}\) are an orthonormal basis such that
    \(\ket{\alpha}, \ket{\beta} \in \left\{ \ket{j} \right\}\).
    Then 
    \begin{align}
        \left| \ev{\left(O \otimes \ketbra{\alpha}{\beta}\right)}{\psi} \right| &=
        \left| \sum_{i} c_{i\alpha}^* c_{i \beta} O_{ii} \right|
        \\ &= 
        \sum_{i} k_{i\alpha}^* k_{i \beta},
        \label{eq:k_time}
    \end{align}
    where \(O_{ii}\) denotes the eigenvalue of \(O\) corresponding to the eigenvector \(\ket{i}\) and \(k_{ij}\) is defined implicitly as \(k_{ij} = c_{ij}\sqrt{O_{ii}}\).
    We can consider the quantity in \eq{k_time} as the inner product of two vectors \(\vec{k}_{\alpha}\) and \(\vec{k}_{\beta}\). The Cauchy-Schwarz inequality tells us that
    \begin{equation}
        \left| \sum_{i} k_{i\alpha}^* k_{i \beta} \right| \leq \sqrt{\left(  \sum_{i} k_{i\alpha}^* k_{i \alpha} \right) \left(  \sum_{i} k_{i\beta}^* k_{i \beta} \right)}.
    \end{equation}
    We can choose \(\gamma \in \left\{ \alpha, \beta \right\}\) such that
    \begin{align}
        \left|\sum_{i} k_{i\gamma}^* k_{i \gamma}\right| \geq & \left|\sum_{i} k_{i\alpha}^* k_{i \alpha}\right| \text{ and} \nonumber \\
        \left|\sum_{i} k_{i\gamma}^* k_{i \gamma}\right| \geq & \left|\sum_{i} k_{i\beta}^* k_{i \beta}\right|.
    \end{align}
    Therefore, we have that 
    \begin{align}
        \left| \ev{\left(O \otimes \ketbra{\alpha}{\beta}\right)}{\psi} \right| & \leq \left|\sum_{i} k_{i\gamma}^* k_{i \gamma}\right|
        \\
        &=
        \left| \sum_{i} c_{i\gamma}^* c_{i \gamma} O_{ii} \right|
        \\
        &= \ev{\left(O \otimes \ketbra{\gamma}{\gamma}\right)}{\psi}
    \end{align}
    for either \(\ket{\gamma} = \ket{\alpha}\) or \(\ket{\gamma} = \ket{\beta}\)
    We can remove the absolute value bars in the final line because \(O \otimes \ketbra{\gamma}\) is a positive semidefinite operator.
\end{proof}

Now we can return to our bound from \eq{A_bound_with_offdiags},
\begin{equation}
    \left|\tr \left[ \ketbra{\psi} A_{\bm{xy}} \right]\right|
    \leq 
    \left|\ev{\left(\bigotimes_{\ell=1}^a \left( \mathbb{I} + \ketbra{c_\ell} + \ketbra{d_\ell} \right)
    \bigotimes_{\ell = a + 1}^{a + b} \ketbra{c_\ell}{d_\ell} 
    \bigotimes_{\ell = a + b + 1}^\eta \mathbb{I}\right)}{\psi}\right|.
\end{equation}
By rearranging the registers, we can apply \lem{offdiag_to_diag}.
Taking \(\ket{\alpha}\) to be \(\bigotimes_{\ell = a + 1}^{a + b} \ket{c_\ell}\) and \(\bra{\beta}\) to be \(\bigotimes_{\ell = a + 1}^{a + b} \bra{d_\ell}\), 
we can show that either
\begin{equation}
    \left|\tr \left[ \ketbra{\psi} A_{\bm{xy}} \right]\right|
    \leq 
    \ev{\left(\bigotimes_{\ell=1}^a \left( \mathbb{I} + \ketbra{c_\ell} + \ketbra{d_\ell} \right)
    \bigotimes_{\ell = a + 1}^{a + b} \ketbra{c_\ell}{c_\ell} 
    \bigotimes_{\ell = a + b + 1}^\eta \mathbb{I}\right)}{\psi}
    \label{eq:d_or_c_smaller}
\end{equation}
holds, or an equivalent expression with \(\ketbra{d_\ell}\) instead of \(\ketbra{c_\ell}\) in the second set of registers.
Both cases are identical, so we will proceed using the label \(g_\ell\) for whichever choice is valid in each register.

We can also simplify the expression in the first registers.
We claim that, for each register, we can replace the term \(\ketbra{c_\ell} + \ketbra{d_\ell}\) with either \(2\ketbra{c_\ell}\) or \(2\ketbra{d_\ell}\) without making the expectation value any smaller.
This can be seen by proceeding register by register, using the linearity of the expectation value.
Here again, the choice of \(c_\ell\) or \(d_\ell\) in each register is immaterial, so we use the label \(g_\ell\) to denote whichever one is appropriate for each register.
Making this simplification, we have that 
\begin{equation}
    \left|\tr \left[ \ketbra{\psi} A_{\bm{xy}} \right]\right|
    \leq 
    \ev{\left(
    \bigotimes_{\ell=1}^a \left( \mathbb{I} + 2\ketbra{g_\ell}\right)
    \bigotimes_{\ell = a + 1}^{a + b} \ketbra{g_\ell}
    \bigotimes_{\ell = a + b + 1}^\eta \mathbb{I}
    \right)}{\psi}.
    \label{eq:g_thangs}
\end{equation}

\subsubsection{Taking advantage of antisymmetry}

Now we will take advantage of the antisymmetry of \(\ket{\psi}\) to bound the expectation values in \eq{g_thangs}.
It is helpful to rewrite the expression in the first set of registers in a different form:
\begin{align}
    \bigotimes_{\ell=1}^a \left( \mathbb{I} + 2\ketbra{g_\ell}\right) = 
    \sum_{w=0}^a 2^w \sum_{S \subseteq \left[ a \right]: |S| = w} \bigotimes_{\ell=1}^a W_\ell^S,
\end{align}
where \(W_\ell^S = \ketbra{g_\ell}\) if \(\ell \in S\) and \(W_\ell = \mathbb{I}\) otherwise.
This then leads us to the bound
\begin{equation}
    \left|\tr \left[ \ketbra{\psi} A_{\bm{xy}} \right]\right|
    \leq 
    \sum_{w=0}^a 2^w \sum_{S \subseteq \left[ a \right]: |S| = w} 
    \ev{\left(
    \bigotimes_{\ell=1}^a W_\ell^S
    \bigotimes_{\ell = a + 1}^{a + b} \ketbra{g_\ell}
    \bigotimes_{\ell = a + b + 1}^\eta \mathbb{I}
    \right)}{\psi}.
    \label{eq:g_thangs_now_with_more_W}
\end{equation}

Now that we have obtained this bound, we will proceed to use the antisymmetry of \(\ket{\psi}\) to show that
\begin{equation}
    \ev{\left(
    \bigotimes_{\ell=1}^a W_\ell^S,
    \bigotimes_{\ell = a + 1}^{a + b} \ketbra{g_\ell}
    \bigotimes_{\ell = a + b + 1}^\eta \mathbb{I}
    \right)}{\psi} \leq \frac{1}{P (\eta, |S| + b)} = \frac{\left( \eta - |S| - b \right)!}{\eta !}.
    \label{eq:W_g_bound}
\end{equation}

To do so, let us prove the following lemma,
\begin{lemma} \label{lem:projector_expectation_bound}
    Let \(\ket{\psi}\) be a normalized pure state on \(\eta\) registers of \(n\) qubits each. Furthermore, let \(S \ket{\psi} = -\ket{\psi}\) for any operator \(S\) that swaps the states of two of the registers. Let \(\{P_i\}_{i \in \left[ k \right]}\) be a set of projectors onto orthonormal \(n\) qubit states.
    Then 
    \begin{equation}
        0 \leq \ev{\left(\bigotimes_{i=1}^k P_i \bigotimes_{i=k + 1}^{\eta} \mathbb{I}\right)}{\psi} \leq \frac{1}{P(\eta, k)} = 
        \frac{\left( \eta - k \right)!}{\eta!},
        \label{eq:projector_product_bound}
    \end{equation}
    where \(P(\eta, k)\) denotes the number of ways to choose a sequence of \(k\) items from a set of size \(\eta\).
\end{lemma}
\begin{proof}
    Let \(S_k\) denote the set of all sequences obtained by choosing \(k\) items from the set \(\left[ \eta \right]\). Note that two sequences with the same elements in different orders are treated as distinct elements of \(S_k\).
    For a sequence \(s \in S_k\) we define the operator \(A_s\) as the operator that acts on register \(s_i\) with the projector \(P_i\) for all \(i \in \left[ k \right]\) and acts on the other \(\eta - k\) registers with the identity operation. Note that all of the operators \(A_s\) are defined using the same set of \(k\) projectors acting on (potentially) different registers.

    We will prove the claim by showing that 
    \begin{equation}
        \sum_{s \in S_k} \ev{A_s}{\psi} \leq 1.
    \end{equation}
    Clearly the operators \(\left\{ A_s \right\}_{s \in S_k}\) are all projectors onto different subspaces.
    In general, these projectors are not orthogonal (under the Hilbert-Schmidt inner product).
    Equivalently, we could say that the \(+1\) eigenspaces of these operators are not orthogonal in general.

    However, we can show that \(\ket{\psi}\) has no support on states that are in the \(+1\) eigenspace of more than one of these projectors.
    Consider \(A_x\) and \(A_y\) for \(x \neq y\).
    There must be some register \(\ell\) on which they act differently.
    If \(A_x\) and \(A_y\) both act on register \(\ell\) with distinct projectors \(P_i\) and \(P_j\) then \(A_x A_y\) = 0 and their eigenspaces have no overlap, so we are done.
    Assume that only one of \(A_x\) and \(A_y\) acts on register \(\ell\).
    Without loss of generality we consider the case where \(A_x\) acts on register \(\ell\) with the projector \(P_i\).
    Then, by definition, \(A_y\) acts on a different register \(\ell'\) with \(P_i\) (since \(A_y\) acts with exactly the same projectors as \(A_x\), just on a potentially different set of registers).
    Due to the antisymmetry of \(\ket{\psi}\), we therefore have \(\ev{A_x A_y}{\psi} = 0\).

    Therefore, we can assert that
    \begin{equation}
        \sum_{s \in S_k} \ev{A_s}{\psi} \leq 1.
        \label{eq:sum_A_s}
    \end{equation}
    This could be seen in more detail by expanding \(\ket{\psi}\) in the basis that diagonalizes all of the \(\left\{ A_s \right\}\) and applying the fact that if \(A_x \ket{\phi} = 1\) then \(A_y \ket{\phi} = 0\) for all \(x \neq y\).
    The antisymmetry of \(\ket{\psi}\) also implies that \(\ev{A_x}{\psi} = \ev{A_y}{\psi}\) for all \(x, y\). Therefore, we have that 
    \begin{equation}
        |S_k|\ev{A_s}{\psi} \leq 1
    \end{equation}
    for any \(A_s\).
    The \(\left\{ A_s \right\}\) are all positive semidefinite, so we can bound the expectation value of the particular one from \eq{projector_product_bound} below by zero and divide by \(|S_k| = P(\eta, k)\) to yield
    \begin{equation}
        0 \leq \ev{\left(\bigotimes_{i=1}^k P_i \bigotimes_{i=k + 1}^{\eta} \mathbb{I}\right)}{\psi} \leq \frac{1}{P(\eta, k)} = 
        \frac{\left( \eta - k \right)!}{\eta!},
    \end{equation}
    completing the proof.
\end{proof}

\eq{W_g_bound} follows directly from this lemma and the fact that we can freely permute the observables between registers without changing the expectation value.
Now we can return to \eq{g_thangs_now_with_more_W} and apply \eq{W_g_bound} to show that
\begin{align}
    \left|\tr \left[ \ketbra{\psi} A_{\bm{xy}} \right]\right|
    & \leq 
    \sum_{w=0}^a 2^w \sum_{S \subseteq \left[ a \right]: |S| = w} 
    \ev{\left(
    \bigotimes_{\ell=1}^a W_\ell^S
    \bigotimes_{\ell = a + 1}^{a + b} \ketbra{g_\ell}
    \bigotimes_{\ell = a + b + 1}^\eta \mathbb{I}
    \right)}{\psi}
    \\ & \leq
    \sum_{w=0}^a 2^w \sum_{S \subseteq \left[ a \right]: |S| = w} 
    \frac{\left( \eta - w - b \right)!}{\eta !}
    \\ & = 
    \sum_{w=0}^a 2^w \sum_{S \subseteq \left[ a \right]: |S| = w} 
    \frac{\left( \eta - w\right)!}{\eta !}
    \frac{\left( \eta - w - b \right)!}{\left( \eta - w\right)!}
    \\ & \leq
    \sum_{w=0}^a 2^w \sum_{S \subseteq \left[ a \right]: |S| = w} 
    \frac{\left( \eta - w\right)!}{\eta !}
    \frac{\left( \eta - a - b \right)!}{\left( \eta - a\right)!},
\end{align}
    with the last inequality following from the fact that \(\eta - a \leq \eta - w\).
    Then we have that 
\begin{align}
    \left|\tr \left[ \ketbra{\psi} A_{\bm{xy}} \right]\right|
    & \leq 
    \sum_{w=0}^a 2^w \sum_{S \subseteq \left[ a \right]: |S| = w} 
    \frac{\left( \eta - w\right)!}{\eta !}
    \frac{\left( \eta - a - b \right)!}{\left( \eta - a\right)!},
    \\ & =
    \frac{\left( \eta - a - b \right)!}{\left( \eta - a\right)!}
    \sum_{w=0}^a 2^w \sum_{S \subseteq \left[ a \right]: |S| = w} 
    \frac{\left( \eta - w\right)!}{\eta !}
    \\ & =
    \frac{\left( \eta - a - b \right)!}{\left( \eta - a\right)!}
    \sum_{w=0}^a 2^w 
    \binom{a}{w}
    \frac{\left( \eta - w\right)!}{\eta !}
    \\ & \leq
    \frac{\left( \eta - a - b \right)!}{\left( \eta - a\right)!}
    \sum_{w=0}^a 2^w 
    \binom{a}{w}
    \frac{\left( a-w \right)!}{a!}
    \\ & =
    \frac{\left( \eta - a - b \right)!}{\left( \eta - a\right)!}
    \sum_{w=0}^a
    \frac{2^w}{w!}
    \\ & \leq
    \frac{\left( \eta - a - b \right)!}{\left( \eta - a\right)!}
    \sum_{w=0}^\infty
    \frac{2^w}{w!} \label{eq:series_fun}
    \\ & =
    \frac{\left( \eta - a - b \right)!}{\left( \eta - a\right)!} e^2,
\end{align}
where the last step is obtained by the application of a well-known formula for the infinite sum of the sequence in \eq{series_fun}.

\subsubsection{Putting the pieces together}

Having shown that 
\begin{equation}
    \left|\tr \left[ \ketbra{\psi} A_{\bm{xy}} \right]\right|
    \leq 
    \frac{e^2\left( \eta - a - b \right)!}{\left( \eta - a\right)!},
\end{equation}
we are ready to return to the bound in \eq{o_variance_tr_A}, which we recall below:
\begin{equation}
    \mathrm{Var}(\hat{o}) \leq \sum_{\bm{x} \in R_k} \sum_{\bm{y} \in R_k}\tr \left[ \ketbra{\psi} A_{\bm{x}\bm{y}}\right]. \label{eq:o_variance_tr_A_2}
\end{equation}
We then have that 
\begin{align}
    \mathrm{Var}(\hat{o}) &\leq \sum_{\bm{x} \in R_k} \sum_{\bm{y} \in R_k} \left|\tr \left[ \ketbra{\psi} A_{\bm{x}\bm{y}}\right]\right|
    \\ & \leq
    \sum_{\bm{x} \in R_k} \sum_{\bm{y} \in R_k} 
    \frac{e^2\left( \eta - a - b \right)!}{\left( \eta - a\right)!}. \label{eq:o_variance_e_2}
\end{align}
Recall that we defined \(a\) and \(b\) in \eq{a_b_def} in the following way,
\begin{equation}
    a = |\bm{x} \cap \bm{y}|,
    \qquad \qquad
    b = 2k - 2a. 
    \label{eq:a_b_def_2}
\end{equation}
Recall also the definition of the set of sequences \(R_k\) from \eq{R_k_def},
\begin{equation}
    R_k =
    \left\{ 1,\ldots, \eta / k \right\} \times
    \left\{ \eta / k + 1,\ldots, 2\eta / k \right\} \times
    \cdots \times
    \left\{ \left( k - 1 \right)\eta / k + 1,\ldots, \eta \right\}.
    \label{eq:R_k_def_2}
\end{equation}
Colloquially, a sequence in \(R_k\) indexes a set of \(k\) registers, one from the first group of \(\eta / k\), one from the second group of \(\eta / k\), and so on.

Let us consider a fixed sequence \(\bm{x} \in R_k\) and determine how many sequences \(\bm{y} \in R_k\) exist for a specific value of \(a\).
For a fixed value of \(a\), \(\bm{x}\) and \(\bm{y}\) share \(a\) elements.
By construction, there are \(\binom{k}{a}\) different choices for these \(a\) elements (because there are \(k\) groups and \(\bm{x}\) and \(\bm{y}\) can either match or fail to match in each group).
In each of the \(k - a\) groups of registers where \(\bm{x}\) and \(\bm{y}\) don't match, there are exactly \(\eta / k - 1\) ways to choose the corresponding element of \(\bm{y}\).
Therefore, for a given \(a\) and \(\bm{x}\), we have that
\begin{equation}
    \left| \left\{ \bm{y} \in R_k : \left| \bm{x} \cap \bm{y} \right| = a\right\} \right| = \binom{k}{a}\left( \eta / k - 1 \right)^{k-a}.
\end{equation}

The only way that a particular \(\bm{x}\) or \(\bm{y}\) enters into \eq{o_variance_e_2} is through \(a\) and \(b\), so we can use this fact to take the sums over \(\bm{x}\) and \(\bm{y}\), yielding
\begin{align}
    \mathrm{Var}(\hat{o}) & \leq
    \sum_{\bm{x} \in R_k} \sum_{\bm{y} \in R_k} 
    \frac{e^2\left( \eta - a - b \right)!}{\left( \eta - a\right)!}
    \\ & \leq
    e^2 \sum_{\bm{x} \in R_k} \sum_{a=0}^k \binom{k}{a}\left(  \eta / k - 1  \right)^{k-a}
    \frac{\left( \eta - 2k + a \right)!}{\left( \eta - a\right)!}
    \\ & \leq
    e^2 \left( \eta/k \right)^k  \sum_{a=0}^k \binom{k}{a}\left(  \eta / k - 1  \right)^{k-a}
    \frac{\left( \eta - 2k + a \right)!}{\left( \eta - a\right)!},
\end{align}
under the assumption that \(\eta > 2k\) so that we don't have to restrict the sum over \(a\).

Simplifying the inequality further, we find that
\begin{align}
    \mathrm{Var}(\hat{o}) & \leq
    e^2 \left( \eta/k \right)^k \left(  \eta / k - 1  \right)^{k} \sum_{a=0}^k \binom{k}{a}\left(  \eta / k - 1  \right)^{-a}
    \frac{\left( \eta - 2k + a \right)!}{\left( \eta - a\right)!}
    \\ & \leq
    e^2 \left( \eta/k \right)^k \left(  \eta / k - 1  \right)^{k} \sum_{a=0}^k \binom{k}{a}\left(  \eta / k - 1  \right)^{-a}
    \frac{\left( \eta - 2k + a \right)!}{\left( \eta - k\right)!}.
    \label{eq:messy_factorials}
\end{align}

Now we employ the upper and lower bounds from Stirling's formula (that hold for any integer \(n > 0\)),
\begin{equation}
    \sqrt{2 \pi n} \left( \frac{n}{e} \right)^n < n! < e \sqrt{2 \pi n} \left( \frac{n}{e} \right)^n.
\end{equation}
We can use these bounds to simplify the ratio of factorials in \eq{messy_factorials},
\begin{align}
    \frac{\left( \eta - 2k + a \right)!}{\left( \eta - k\right)!} & \leq
    e \sqrt{2 \pi \left( \eta - 2k + a \right)} \left( \frac{\eta - 2k + a}{e} \right)^{\eta - 2k + a} \frac{1}{\left( \eta-k \right)!}
    \\ & \leq
    e \sqrt{2 \pi \left( \eta - k \right)} \left( \frac{\eta - k}{e} \right)^{\eta - 2k + a} \frac{1}{\left( \eta-k \right)!}
    \\ & \leq
    e \left( \frac{\eta - k}{e} \right)^{\eta - 2k + a} \left(\frac{e}{\eta - k}\right)^{\eta - k}
    \\ & = 
    e \left(\frac{e}{\eta - k}\right)^{k-a}.
\end{align}
Using the assumption that \(\eta > 2k\) we can proceed further, yielding
\begin{align}
    \frac{\left( \eta - 2k + a \right)!}{\left( \eta - k\right)!} & \leq
    e \left(\frac{e}{\eta - k}\right)^{k-a} 
    \\ & \leq
    e \left(\frac{2e}{\eta}\right)^{k-a} 
    \\ & = 
    e \left(\frac{2e}{\eta}\right)^{k} \left(\frac{2e}{\eta}\right)^{-a},
    \label{eq:nice_factorials}
\end{align}
where we have used the fact that \(\eta > 2k\) implies that \(\eta - k > \eta / 2\).

We can use \eq{nice_factorials} to further simplify \eq{messy_factorials}, finding that,
\begin{align}
    \mathrm{Var}(\hat{o}) & \leq
    e^2 \left( \eta/k \right)^k \left(  \eta / k - 1  \right)^{k} \sum_{a=0}^k \binom{k}{a}\left(  \eta / k - 1  \right)^{-a}
    \frac{\left( \eta - 2k + a \right)!}{\left( \eta - k\right)!}
    \\ & \leq
    e^3 \left( \eta/k \right)^k \left(  \eta / k - 1  \right)^{k} \sum_{a=0}^k \binom{k}{a}\left(  \eta / k - 1  \right)^{-a}
    \left(\frac{2e}{\eta}\right)^{k} \left(\frac{2e}{\eta}\right)^{-a}
    \\ & \leq
        e^3 \left( \eta^2/k^2 \right)^{k} \sum_{a=0}^k \binom{k}{a}\left(  \frac{\eta}{2k} \right)^{-a}
    \left(\frac{2e}{\eta}\right)^{k} \left(\frac{2e}{\eta}\right)^{-a}
    \\ & = 
    e^3 \left(\frac{2e \eta}{k^2} \right)^k \sum_{a=0}^k \binom{k}{a}\left( \frac{e}{k} \right)^{-a}.
\end{align}
Note that we again used the fact that \(\eta > 2k\) implies that \(\eta - k > \eta / 2\) to simplify the part of the bound involving \(\left( \eta/k - 1 \right)\).
Applying the binomial theorem to the sum yields the bound
\begin{align}
    \mathrm{Var}(\hat{o}) & \leq
    e^3 \left(\frac{2e \eta}{k^2} \right)^k e^{-k} \left( k + e \right)^k
    \\ & = 
    e^3 \left(\frac{2 \eta \left( k + e \right)}{k^2} \right)^k.
\end{align}

Recall that we defined the estimator \(\hat{o}\) by neglecting the coefficient \(\frac{k^k \left( \eta ! \right)}{\eta ^ k\left( \eta - k \right)!}\) in \eq{R_k_register_k_rdm}'s expression for the $k$-RDM element \(\prescript{k}{}D_{i_1, \ldots, i_k}^{j_1, \ldots, j_k}\).
If we let \(\hat{d}\) be the estimator for this $k$-RDM element with the coefficient included, we have that
\begin{equation}
    \mathrm{Var}(\hat{d}) =
    \left(\frac{k^k \left( \eta ! \right)}{\eta ^ k\left( \eta - k \right)!}\right)^2 \mathrm{Var}(\hat{o}).
\end{equation}
Therefore, we can bound the desired variance by
\begin{align}
  \mathrm{Var}(\hat{d}) & \leq \left(\frac{k^k \left( \eta ! \right)}{\eta ^ k\left( \eta - k \right)!}\right)^2  
    e^3 \left(\frac{2 \eta \left( k + e \right)}{k^2} \right)^k.
\end{align}
Simplifying this expression, we obtain
\begin{align}
  \mathrm{Var}(\hat{d}) & \leq 
  \left(\frac{k^k \left( \eta ! \right)}{\eta ^ k\left( \eta - k \right)!}\right)^2  
    e^3 \left(\frac{2 \eta \left( k + e \right)}{k^2} \right)^k
    \\ & =
    e^3
    \left(\frac{\left( \eta ! \right)}{\left( \eta - k \right)!}\right)^2  
    \left(\frac{2 \left( k + e \right)}{\eta} \right)^k
    \\ & \leq
    e^3 \eta^k \left( 2k + 2e \right)^k,
\end{align}
which is the bound advertised in \eq{d_hat_variance_bound_intro}.

\section{More efficient Slater determinant state preparation in first quantization}
\label{app:state_prep}

The general principle is to prepare the state in second quantization, then convert it to first quantization.
To avoid needing to store all $N$ qubits for the second quantized state as it is produced, we convert its qubits to the first quantized representation.

To explain this, we will first explain how a state in the second quantized representation can be converted to the first quantized representation.
A computational basis state in second quantization consists of a string of $N$ bits with $\eta$ ones and $N-\eta$ zeros.
The procedure is to run through these qubits in sequence and store the locations in $\eta$ registers of size $\lceil\log N\rceil$.
Let us call the qubit number we consider from the second quantized representation $q$ and also record the number of electrons (ones) found so far as $\xi$.
The value of $\xi$ will be stored in an ancilla register of size $n_\eta=\lceil\log (\eta+1)\rceil$.

We initialize all $\eta$ registers for the first quantized representation and the $\xi$ register as zero.
Then, for $q=1$ to $N$ we perform the following.
\begin{enumerate}
    \item Add the value in qubit $q$ to the $\xi$ register, with Toffoli cost $n_\eta-1$.
    If the qubit is in the state $\ket 1$ then $\xi$ is incremented.
    \item Now use qubit $q$ to control unary iteration \cite{BabbushSpectra} on the register $\xi$, which has cost $\eta-1$.
    \item Use this unary iteration to write the value $q$ into register $\xi$ using CNOTs.
    Because $q$ is iterated classically, only CNOTs are needed, with no further Toffolis beyond that needed for the unary iteration.
    Because the unary iteration is controlled by qubit $q$, in the case where qubit $q$ is in state $\ket 0$, the unary iteration does not proceed and the value of $q$ is not written out.
    \item Now perform unary iteration on $\xi$ again that is \emph{not} controlled; the cost is $\eta-2$.
    \item We use the unary iteration on $\xi$ to check if the value in register number $\xi$ is $q$; if it is then we perform a NOT on qubit $q$.
    This multiply-controlled Toffoli is controlled by $\lceil \log N\rceil +1$ qubits (including the qubit from the unary iteration), so it has a cost of $\lceil \log N\rceil$.
    But, this is done for each of the $\eta$ registers, for a total cost $\eta\lceil\log N\rceil$.
\end{enumerate}
The last operation ensures that qubit $q$ is set to $\ket 0$.
+That is because, if it is initially $\ket 0$, then value $q$ is not written in register $\xi$, and the value is not flipped.
If it is initially $\ket 1$, then $q$ is written in register $\xi$, and the multiply-controlled Toffoli flips this qubit to $\ket 0$.

So far this procedure gives an ordered list of the electron positions, but we need an antisymmetrized state.
To obtain that, we apply the procedure in \cite{Berry2018} to antisymmetrize with cost $\mathcal{O}(\eta\log\eta \log N)$.
The total Toffoli cost is
\begin{equation}
    N \left( 2\eta +n_\eta - 3 + \eta\lceil\log N\rceil\right) + \mathcal{O}(\eta\log\eta \log N).
\end{equation}
The dominant cost here is $\eta N\log N$ from erasing the qubits in the second quantized representation, with the factor of $\log N$ coming from the need to check all qubits of each register to check if it is $q$.
However, recall that in unary iteration it is possible to check if a register is equal to a consecutive sequence of values without this logarithmic overhead, and we are considering consecutive values of $q$.

To eliminate that overhead, we, therefore, consider simultaneous unary iteration on all of the $\eta$ registers.
That is, for each register for the first quantized representation, we also store the qubits needed for unary iteration, as well as a control register to ensure we do not iterate on registers that do not have value written into them yet.
The control qubits will correspond to the value of $\xi$ in unary.
Our modified procedure is as follows (with the iteration of $q$ from 1 to $N$).
\begin{enumerate}
    \item Perform a single step of unary iteration on all $\eta$ registers with cost $\eta$ Toffolis.
    \item Add the value in qubit $q$ to the $\xi$ register, with Toffoli cost $n_\eta-1$.
    \item Use qubit $q$ to control unary iteration on the register $\xi$, which has cost $\eta-1$.
    \item Use this unary iteration to write the value $q$ into register $\xi$, \emph{as well} as the $\lceil\log N\rceil$ ancilla qubits for the unary iteration and the control qubit.
    Again this is performed with CNOTs.
    \item Convert the control qubits to one-hot unary using a sequence of CNOTs.
    \item For each of the $\eta$ registers, use the control qubit and the unary iteration output to control a NOT on qubit $q$.
    This has a cost of a single Toffoli for each register, for a toal of $\eta$.
    \item Convert the control qubits to from one-hot unary with CNOTs.
\end{enumerate}
As a result, we have eliminated the $\log N$ factor and also eliminated the cost of $\eta-2$ for the unary iteration on $\xi$ (because the control qubits are a unary representation of $\xi$).
One might ask if the binary representation of $\xi$ is still needed; however, it would be more costly to add increment $\xi$ in unary (about $\eta$ cost instead of $\log\eta$).
The total Toffoli cost of this procedure is now
\begin{equation}
    N \left( 3\eta +n_\eta - 2\right) + \mathcal{O}(\eta\log\eta \log N),
\end{equation}
where the order term is the cost for antisymmetrizing. Note that this reduces the Toffoli cost, but there is still a Clifford cost of $N\eta\log N$ from the CNOTs to place the value of $q$ in the first quantized registers.

Now to efficiently prepare the Slater determinant, we can perform the sequence of Givens rotations on the qubits for the second quantized representation.
The Givens rotations are performed in a sequence where Givens rotations are performed in a layer on qubits 1 to $\eta+1$, then on qubits 2 to $\eta+2$, then 3 to $\eta+3$, and so on. One can find the details of the Givens rotations that must be applied in \cite{Kivlichan2018QuantumConnectivity}.
Generally, layer $q$ of Givens rotations is performed on qubits $q$ to $\eta+q$.
After the first layer there are only $\eta+1$ qubits being used, and the first qubit is not accessed again in the preparation.
Therefore we can convert this qubit to the first quantized representation and erase it.
Then there are only $\eta$ qubits actively being used in the second quantized representation, and the next layer will be performed on qubits $2$ to $\eta+2$, bringing on one more qubit.

In this way, each time we perform a layer of Givens rotations to prepare the state, we can convert one qubit to the first quantized representation, and only $\eta+1$ qubits of the second quantized representation need be used at once, which is trivial compared to the number of qubits used for the first quantized representation.

\end{document}